\documentclass[12pt,oneside,english]{article}
\usepackage{authblk}
\usepackage{babel}
\usepackage[T1]{fontenc}
\usepackage[latin9]{inputenc}
\usepackage{simpler-wick}
\usepackage{geometry}
\usepackage{amstext,amsthm,amssymb}
\usepackage[justification=centering]{caption}
\usepackage{stmaryrd}
\usepackage{graphicx}
\usepackage{mathtools}
\usepackage{wasysym}
\usepackage{mathrsfs}
\usepackage{bbm}
\usepackage{esint}
\usepackage{dsfont}
\usepackage{xcolor}
\usepackage{verbatim}
\usepackage{subfigure}
\usepackage{setspace}
\usepackage{url}
\usepackage[percent]{overpic}
\usepackage{tikz}
\usepackage{relsize}
\usepackage{titlesec}
\usepackage{simpler-wick}
\usepackage{yhmath}

\makeatletter

\numberwithin{equation}{section}
\numberwithin{figure}{section}
\newtheorem{thm}{Theorem}[section]
\newtheorem*{thm*}{Theorem}
\newtheorem{lemma}[thm]{Lemma}
\newtheorem{prop}[thm]{Proposition}
\newtheorem{coro}[thm]{Corollary}

\theoremstyle{definition}

\theoremstyle{remark}
\newtheorem{rmk}{Remark}[section]
\newtheorem{ex}{Example}[section]
\numberwithin{thm}{section}

\definecolor{kallecol}{rgb}{.99,.1,.5}
\definecolor{davidcol}{rgb}{.10,.1,.99}
\definecolor{delaracol}{rgb}{.1,.70,.30}
\definecolor{sketchcol}{rgb}{.4,.4,.8}
\definecolor{outlinecol}{rgb}{.8,.4,.3}

\newcommand{\zu}{\mathbf{u}}
\newcommand{\zubis}{\mathbf{v}}
\newcommand{\zs}{\mathbf{z}}
\newcommand{\zsbis}{\mathbf{w}}

\newcommand{\bu}{\mathbf{b}}

\newcommand{\vertical}{\textnormal{v}}
\newcommand{\horizontal}{\textnormal{h}}
\newcommand{\DomBall}[1]{\mathbf{D}(r)}
\newcommand{\DomBallInt}[1]{\mathbf{D}^\circ(r)}
\newcommand{\DomCirc}[1]{\mathbf{C}(r)}

\newcommand{\Neumann}{\mathtt{N}}
\newcommand{\Dirichlet}{\mathtt{D}}
\newcommand{\DorN}{{\Dirichlet/\Neumann}}

\newcommand{\LatticeSign}{{\epsilon}}
\newcommand{\mybar}[3]{%
	\mathrlap{\hspace{#2}\overline{\scalebox{#1}[1]{\phantom{\ensuremath{#3}}}}}\ensuremath{#3}
}
\newcommand{\Null}{\textnormal{\small Null}}
\newcommand{\Poles}[1]{\textnormal{S}^{[#1]}}
\newcommand{\tens}{\otimes}
\newcommand{\SymmAlg}{\mathrm{S}}
\newcommand{\SymmPow}[1]{\SymmAlg^{#1}}
\newcommand{\dgffptwise}{\phi}
\newcommand{\dgff}[1]{\langle\Phi,#1\rangle}

\newcommand{\gff}{\varphi}
\newcommand{\idField}{I}
\newcommand{\dOrigLocLinFi}{\mathcal{P}}
\newcommand{\dLocLinFi}{{\dOrigLocLinFi_{\nabla}^{\textnormal{lin}}}}
\newcommand{\dLinFields}{\dOrigFields_\nabla^{\textnormal{lin}}}
\newcommand{\dLinFieldsRad}[1]{\dOrigFields_\nabla^{\textnormal{lin}}{(#1)}}
\newcommand{\dOrigLocFi}{\mathcal{P}}
\newcommand{\dLocFi}{{\dOrigLocFi_\nabla}}
\newcommand{\dLocFiDeg}[1]{{\SymmPow{{#1}} \dLocLinFi}}

\newcommand{\dNuFi}{\mathcal{N}}

\newcommand{\dOrigFields}{\mathcal{F}}
\newcommand{\dFields}{{\dOrigFields_\nabla}}
\newcommand{\dFieldsDeg}[1]{{\dOrigFields_\nabla^{(#1)}}}

\newcommand{\ChiralFock}{\mathscr{F}}
\newcommand{\AntiChiralFock}{\mybar{0.85}{2.5pt}{\mathscr{F}}}
\newcommand{\FullFock}{\ChiralFock \tens \AntiChiralFock}
\newcommand{\dHolCurr}{\textnormal{J}}

\newcommand{\dHolCurrMode}[1]{\mathsf{J}_{#1}}
\newcommand{\dAntHolCurrMode}[1]{\mybar{0.95}{1pt}{\mathsf{J}}_{#1}}
\newcommand{\dHomFi}[2]{(\dFields)_{#1,#2}}
\newcommand{\dHomLinFi}[2]{(\dLinFields)_{#1,#2}}
\newcommand{\RadSupp}[1]{R_{\mathrm{Supp}}(#1)}
\newcommand{\dRep}[1]{J_{#1}}
\newcommand{\dRepBar}[1]{\mybar{0.75}{2.7pt}{J}_{#1}}
\newcommand{\FockId}{\mathbbm{1}}
\newcommand{\dLocLinFiRad}[1]{\mathcal{P}_{\nabla}^{\textnormal{lin}}{(#1)}}

\newcommand{\FunSupp}{\textnormal{supp }}
\newcommand{\PolySupp}{\textnormal{Supp }}
\newcommand{\HolCurr}{\textnormal{J}}
\newcommand{\AntiHolCurr}{\mkern 1.5mu\overline{\mkern-2mu {\textnormal{J}\mkern-1.7mu}\mkern 1.5mu}}
\newcommand{\holcurrfield}{\textnormal{j}}
\newcommand{\antiholcurrfield}{\mkern 1.5mu\overline{\mkern 0.5mu {\textnormal{j}\mkern-1.7mu}\mkern 1.5mu}}
\newcommand{\field}{X}
\newcommand{\primary}{\bullet}
\newcommand{\dual}{*}
\newcommand{\medial}{\textnormal{m}}
\newcommand{\corner}{\textnormal{c}}
\newcommand{\SqLat}{\Z^2}
\newcommand{\SqLatMesh}{\meshsz \Z^2}
\newcommand{\SqLatMeshPrim}{\meshsz \Z^2_\primary}
\newcommand{\SqLatMeshDual}{\meshsz \Z^2_\dual}
\newcommand{\SqLatMeshMedial}{\meshsz \Z^2_\medial}

\newcommand{\SqLatMeshDiamond}{\meshsz \Z^2_\diamond}
\newcommand{\SqLatCorner}{\Z^2_\corner}
\newcommand{\SqLatMedial}{\Z^2_\medial}
\newcommand{\SqLatDiamond}{\Z^2_\diamond}

\newcommand{\ZPrimary}{\Z^2_\primary}
\newcommand{\ZDual}{\Z^2_\dual}
\newcommand{\ZMedial}{\Z^2_\medial}
\newcommand{\ZDiamond}{\Z^2_\diamond}
\newcommand{\ZCorner}{\Z^2_\corner}

\newcommand{\CFTcorr}[3]{\left\langle #3 \right\rangle_{#1}^{#2}}
\newcommand{\CFTcorrBig}[3]{\Big\langle #3 \Big\rangle_{#1}^{#2}}
\newcommand{\GFFkernelWO}[2]{\mathrm{K}_{#1}^{#2}}
\newcommand{\GFFkernel}[4]{\GFFkernelWO{#1}{#2} (#3 ; #4)}

\newcommand{\testfun}[1]{\mathcal{D}(#1)}
\newcommand{\testfunZA}[1]{\mathcal{D}_\nabla(#1)}
\newcommand{\distrib}[1]{\mathcal{D}'(#1)}
\newcommand{\distribMC}[1]{\mathcal{D}'_\nabla(#1)}

\newcommand{\ddistribMC}[1]{\mathrm{Fun}_{\nabla}(#1)}

\newcommand{\term}[1]{{\bf #1}}

\newcommand{\ev}{\textnormal{ev}}

\newcommand{\interior}{\textnormal{int}^\sharp}

\newcommand{\Ball}{\mathbf{B}^\sharp}

\newcommand{\norm}[1]{\Vert #1\Vert}

\newcommand{\Green}{\mathrm{G}}
\newcommand{\GreenGrad}{\Green_\nabla}

\newcommand{\psmallbulletbar}{{\tikz[baseline=0ex]\draw[black,fill={black}] (0,0) circle (1.1pt) ;}}%
\newcommand{\psmallbullet}{{\tikz[baseline=-0.55ex]\draw[black,fill={black}] (0,0) circle (1.1pt) ;}}

\newcommand{\Vir}{\mathbf{Vir}}
\newcommand{\Hei}{\mathfrak{hei}}
\newcommand{\AntHei}{\overline{\Hei}}
\newcommand{\UHei}{\mathrm{Hei}}
\newcommand{\AntUHei}{\overline{\UHei}}

\newcommand{\ii}{\mathbbm{i}}
\newcommand{\re}{\textnormal{Re}\,}
\newcommand{\im}{\textnormal{Im}\,}

\newcommand{\eps}{\varepsilon}
\newcommand{\bdry}{\partial}
\newcommand{\diskRC}[2]{\mathbf{D}_{#1}(#2)}

\newcommand{\domain}{\Omega}
\newcommand{\ddomain}{\domain_\meshsz}
\newcommand{\ConfigSp}[2]{\mathrm{Conf}_{#1} \! \left( #2 \right)}
\newcommand{\ContFun}{C}

\newcommand{\cconj}[1]{\overline{#1}}

\newcommand{\half}{\frac{1}{2}}

\newcommand{\SymmGrp}{\mathfrak{S}}

\newcommand{\dist}{\mathrm{dist}}

\newcommand{\OO}{\mathcal{O}}
\newcommand{\oo}{o}

\newcommand{\id}{\mathsf{id}}
\newcommand{\idof}[1]{\id_{#1}}
\newcommand{\dmn}{\mathrm{dim}}
\newcommand{\spn}{\mathrm{span}}

\newcommand{\Kern}{\mathrm{Ker}}

\newcommand{\isom}{\cong}

\newcommand{\C}{\mathbb{C}} 
\newcommand{\R}{\mathbb{R}} 
\newcommand{\Z}{\mathbb{Z}} 
\newcommand{\Znn}{\Z_{\geq 0}} 
\newcommand{\Zpos}{\Z_{> 0}} 
\newcommand{\N}{\mathbb{N}} 
\newcommand{\UD}{\mathbb{D}} 
\newcommand{\bC}{\C} 
\newcommand{\bR}{\R} 
\newcommand{\bZ}{\Z} 
\newcommand{\bN}{\N} 

\newcommand{\dcint}[1]{\int^{\sharp}_{{#1}}} 
\newcommand{\dcoint}[1]{\sqint^{\sharp}_{{#1}}} 
\newcommand{\dd}[1]{\ud^\sharp {#1}}
\newcommand{\ud}{\mathrm{d}} 

\newcommand{\meshsz}{\delta}

\newcommand{\dnabla}{\nabla_{\sharp}}
\newcommand{\deebar}{\mkern 1.5mu\overline{\mkern -1.5mu {\dee\mkern-1.7mu}\mkern 1.5mu}}

\newcommand{\gdeebar}{\deebar_{\sharp}}
\newcommand{\dee}{\partial}

\newcommand{\gdee}{\dee_{\sharp}}
\newcommand{\pdee}{\dee^\psmallbullet_{\sharp}}
\newcommand{\pdeebar}{\deebar^\psmallbulletbar_{\sharp}}

\newcommand{\pder}[1]{\frac{\partial}{\partial{#1}}}
\newcommand{\pdder}[1]{\frac{\partial^2}{\partial{#1}^2}}

\newcommand{\pderof}[2]{\frac{\partial{#2}}{\partial{#1}}}
\newcommand{\Lapl}{\triangle}

\newcommand{\gLapl}{\Lapl_{\sharp}}
\newcommand{\dDLapl}{\Lapl_{\ddomain}^{\Dirichlet}}
\newcommand{\dDLaplDual}{\Lapl_{\ddomain^*}^{\Dirichlet}}
\newcommand{\dNLapl}{\Lapl_{\ddomain}^{\Neumann}}

\newcommand{\dvirL}[1]{\mathsf{L}_{#1}}

\newcommand{\dvirBarL}[1]{\mybar{0.95}{-0.5pt}{\mathsf{L}}_{#1}}
\newcommand{\VirL}[1]{\mathsf{L}_{#1}}
\newcommand{\VirBarL}[1]{\mybar{0.95}{-0.5pt}{\mathsf{L}}_{#1}}
\newcommand{\VirC}{\mathsf{C}}
\newcommand{\HeiJ}[1]{\mathsf{j}_{#1}}
\newcommand{\HeiK}{\mathsf{k}}
\newcommand{\AntHeiJ}[1]{\bar{\mathsf{j}}_{#1}}

\newcommand{\barDelta}{\mybar{0.7}{1.2pt}{\Delta}}

\newcommand{\EX}{\mathbb{E}}

\newcommand{\set}[1]{\left\{ #1 \right\}}

\newcommand{\voavac}{[1]}

\newcommand{\confmap}{\varphi}

\newcommand{\no}[1]{\, \rotatebox[]{90}{\scalebox{.8}{$\ \circ\,\circ$}}\,#1\,\rotatebox[]{90}{\scalebox{.8}{$\ \circ\,\circ$}}\,}
\newcommand{\noQuo}[1]{\, \rotatebox[]{90}{\scalebox{.8}{$\ \bullet\,\bullet$}}\,#1\,\rotatebox[]{90}{\scalebox{.8}{$\ \bullet\,\bullet$}}\,}
\newcommand{\Pair}{\mathrm{Pair}}
\newcommand{\ParPair}{\mathrm{PartPair}}

\titleformat{\subsection}[runin]
{\normalfont\bfseries}{\thesubsection}{0pt}{\hspace{1.5em}}[.]
\titleformat{\section}
{\normalfont\Large\bfseries}{\thesection}{0pt}{\hspace{1.5em}}

\begin{document}

\title{\Large\scshape\bfseries
Fock space of local fields of the discrete GFF
and its scaling limit bosonic CFT\vspace{0.5cm}}

\author[$\!\,$]{David Adame$\,$-$\,$Carrillo\footnote{\texttt{david.adamecarrillo@aalto.fi}}}
\author[$\!\,$]{Delara Behzad\footnote{\texttt{d.behzad@uu.nl}}}
\author[$\!\,$]{Kalle Kyt\"ol\"a\footnote{\texttt{kalle.kytola@aalto.fi}}}

\affil[$*$ $\ddagger$]{%
\textit{
Department of Mathematics and Systems Analysis, Aalto University, \linebreak[4]
P.O. Box 11100, FI-00076 Aalto, Finland \linebreak[4]}}
\affil[$\dagger$]{%
\textit{Mathematical Institute, Utrecht University, \linebreak[4] Budapestlaan 6, 3584 CD Utrecht, The Netherlands}}

\date{}

\maketitle

\begin{abstract}
To connect conformal field theories (CFT) to probabilistic lattice models, recent works
\cite{HKV, AdameCarrillo-discrete_symplectic_fermions}
have introduced a novel definition of local fields of the lattice models. Local fields in this
picture are probabilistically concrete: they are built from random variables in the model. The
key insight is that discrete complex analysis ideas allow to equip the space of local fields
with the main structure of a CFT: a representation of the Virasoro algebra.

In this article, for the first time, we fully analyze the structure of the space of
local fields of a lattice model as a representation, and use this to establish a
one-to-one correspondence between the local fields of a lattice model and those of a
conformal field theory.
The CFT we consider is probabilistically realized in terms of the gradient of the
Gaussian Free Field~(GFF). Its space of local fields is just a bosonic Fock space for two chiral
symmetry algebras. The corresponding lattice model is the discrete Gaussian Free Field.
Our first main result is that the space of local fields of polynomials in the gradient of the
discrete GFF is isomorphic to the Fock space. These local fields make sense
with both Dirichlet and Neumann boundary conditions.
Our second main result is that
with the appropriate renormalization,
correlation functions of local fields of the discrete GFF
converge in the scaling limit to the correlation functions of the CFT.
The renormalization needed is, conceptually correctly,
according to the eigenvalue of the Virasoro generator~$\VirL{0} + \VirBarL{0}$
on the local field.
\vfill
\end{abstract}

\tableofcontents

\newpage

\section{Introduction}
	\label{sec:intro}
	\subsection{Lattice models as discretizations of conformal field theories}

Making mathematical sense of quantum field theories
in continuum space(-time) is often very challenging.
Many of the physicists' commonly used ways (such as path integrals) of specifying field theories
are not as such well-defined in con\-tinuum space.
One typically needs to introduce a short-distance cutoff (ultraviolet cutoff)
to write down proper mathematical definitions,
and then one should analyze what happens as the cutoff is removed.
One standard way of implementing the cutoff is to discretize the theory to a lattice.
The mesh size of the lattice then serves as the short-distance cutoff length scale.
The advantage is that defining the discretized field theories as probabilistic models
in (finite) lattice domains is usually straightforward. The difficulty then
lies in removing the ultraviolet cutoff by forming a scaling limit of these well-defined models,
i.e., proving that relevant limits exist when one lets the lattice mesh size tend to zero
and proving field-theoretically desired properties of these limits.
Despite still being challenging, the approach of lattice discretizations and their
scaling limits is one of the main robust and standard ways to construct field
theories~\cite{GJ-quantum_physics}.

Since the directly physically important quantities in a field theory
are its correlation functions, for a scaling limit to be viewed as a construction of a field
theory, it should address the convergence of correlation functions in the disretized
models as the lattice mesh size tends to zero.
In a field theory, correlation functions are assigned to any so-called local fields
at any number of spatial positions. Note that the terminology here can be slightly confusing:
``local fields'' does not simply refer to the ``basic fields'', in terms of which
for example the path integral of the field theory would be written. Local fields are meant to
represent \emph{all} of the locally observable quantities in the theory. The basic fields
are typically among those, but the notion of local fields is vastly more general.

Both the physics and mathematics research places a lot of emphasis on
two-dimensional conformal field theories (CFT) \cite{DMS-CFT, Gawedzki-lectures_on_CFT},
i.e., field theories in two-dimensional space(-time) with certain conformal invariance constraints.
One reason is that by the renormalization group philosophy,
under physically reasonable assumptions, universal
macroscopic (large-scale) behavior of general theories should be governed
by renormalization group fixed points which are CFTs.
For example statistical physics models at their critical points should renormalize
exactly to CFTs, and their universal behavior near the critical point should
be governed by the linearized renormalization group flow near
the fixed point, which in turn can be analyzed in terms of the CFT fixed point itself.
Another reason is that two-dimensional CFTs have remarkable structures that make them both
interesting and significantly more tractable than most other field theories,
and therefore they are great example cases to study in detail.
In particular, by virtue of conformal invariance, the state space of a two-dimensional CFT
carries a representation of (two commuting copies of) the Virasoro algebra.
In turn, by what is known as the state-field correspondence, the state space is exactly identified
with the space of local fields. Therefore the rich algebraic structure of the CFT is
present in the observable quantities which are to be inserted in correlation functions.
The Virasoro algebra action on local fields describes the (infinitesimal) changes of correlation
functions as (infinitesimal) conformal transformations are applied.
Simple but important special cases 
of conformal transformations
are scaling transformations, and consequently the eigenvalues of the corresponding
Virasoro generators~(${\VirL{0} + \VirBarL{0}}$) 
describe exactly how local fields should be renormalized as functions of
the lattice mesh when forming the scaling limit
(or more generally renormalization with any ultra-violet cutoff scale).
This is crucial to the constructive field theory approach via scaling limits, and it
furthermore contains the information about the renormalization group transformation in the vicinity
of the CFT fixed point, and via that, the critical exponents that govern the universal large-scale
behavior of the model near the critical point. The physicists' exact
(albeit not always rigorous) determination
of critical exponents in a number of interesting models via the use of the
representation theory of the Virasoro algebra has indeed been one of the
spectacular successes of two-dimensional conformal field theories. 

The objective in this article is to make concrete sense of
the above picture in the case of one particular conformal field theory and
its probabilistic lattice model discretization.
The field theory is essentially the simplest imaginable one,
the massless free boson in two dimensions, and it correspondingly 
provides the simplest theory in which the general question of (re)constructing
a CFT as a scaling limit is meaningful.
The lattice model is also arguably its simplest discretization,
the discrete Gaussian Free Field.
What is noteworthy is how much detailed structure we can, in this case, match
between the CFT and the
discrete probabilistic model. In fact, we view our results as the first ones to
establish a full conformal field theory as the scaling limit of a probabilistic
lattice model.
We believe that it will eventually be possible to obtain
similar complete CFT scaling limit results also for at least a few other theories~---
among which the most promising candidates are the Ising model and discrete
symplectic fermions,
building on~\cite{HKV} and~\cite{AdameCarrillo-discrete_symplectic_fermions},
respectively.

\subsection{The massless free boson and the Gaussian Free Field}
The massless free boson is, first of all, a fundamental example among conformal
field theories, and it is common to start discussing CFTs with the free boson
as the prototype, see, e.g.,
\cite{DMS-CFT, Gawedzki-lectures_on_CFT, KM-GFF_and_CFT}.
The physics description of the free boson field theory is usually given in terms
of path integral, so that for example the partition function (in the Euclidean signature)
in a planar domain~$\domain$ would be written as
\begin{align}\label{eq: free boson path integral}
\text{``} \,
\int e^{-S(\gff)} \, \mathscr{D} \gff \;
\text{''} , \qquad \text{ where } \qquad
S(\gff) = \frac{1}{8\pi}\int_{\domain} \| \nabla \gff (z) \|^2 \; \ud^2 z \, ,
\end{align}
and with the path integral taken formally over field
configurations~$\gff \colon \domain \to \R$.
The quantity~$S(\gff)$ in the exponential~--- the action of the theory~---
is a quadratic form in~$\gff$, the Dirichlet energy of~$\gff$.
By virtue of this gaussianity, it is actually not very difficult to give a precise
meaning to the massless free boson path integral
as a probability
measure~\cite{Gawedzki-lectures_on_CFT, Sheffield-GFF_for_mathematicians,
PW-lecture_notes_on_GFF, BP-GFF_and_LQG}.
In this probabilistic incarnation, the theory is referred to as the Gaussian Free Field (GFF).
If the domain~$\domain$ has boundaries, 
then in the path-integral
and in the corresponding Gaussian probability measure,
boundary conditions for~$\gff$ should be specified.
Dirichlet and Neumann boundary conditions are the most common choices.
Also unless there are at least some Dirichlet boundary conditions,
the zero-mode (``average'') of~$\gff$ requires special treatment;
the GFF~$\gff$ is then only well-defined up to additive constants.
We will actually work with the gradient~$\nabla \gff$ of the GFF
(see Section~\ref{sec:scal_lim} for the precise definition),
thus removing the additive constant ambiguity, and making the theory
well-defined with both Dirichlet and Neumann boundary conditions.

More complicated field theories can be obtained building on the simple
case of the free boson.
For example, allowing multiple components for the bosonic fields
and compactifying the target space gives rise to interesting
CFTs~\cite{Gawedzki-lectures_on_CFT}.
The GFF is furthermore
at the heart of constructive field theory beyond free fields, because one
can view its well-defined Gaussian measure as a reference and then include
interactions by adding a potential,
see~\cite{GJ-quantum_physics}.\footnote{Often such interacting field theories
would be treated perturbatively around the GFF.
We emphasize, however, that both in principle and in practice,
also nonperturbative constructive treatments of interacting theories
build on the GFF in this way.}
One prominent recent example
is giving a mathematical
meaning to the intricate Liouville conformal field theory via a path integral
\cite{DKRV-LQG_on_Riemann_sphere}, and proving that the path-integral construction indeed gives a conformal field
theory \cite{GKRV-Segals_axioms_and_bootstrap_for_Liouville_theory},
see~\cite{GKR-review_probabilistic_construction_and_bootstrap} for a review.

The GFF fundamentally underlies also many fascinating aspects of random geometry.
Via probabilistic couplings, the GFF is intimately related to SLE-type random curves
\cite{Dubedat-SLE_and_the_free_field_partition_functions_and_couplings,
SS-contour_lines_of_the_DGFF, IK-Hadamard_formula_and_coupling_of_SLE_with_free_field,
SS-contour_line_of_the_continuum_GFF, MS-imaginary_geometry_1, PW-global_and_local_multiple_SLEs}
and conformal loop
ensembles~\cite{QW-coupling_the_GFF_with_free_and_zero_bc,
ASW-bounded_type_thin_local_sets_of_GFF, PW-lecture_notes_on_GFF}.
These random curves can be studied in terms of the GFF and vice versa.

The ubiquity of the GFF, especially its relations to various discrete models,
may appear surprising. This is, however, at least
partly explained by its simplicitly and universality.
The articles \cite{BPR-a_characterization_of_GFF, BPR-one_plus_epsilon_moments_suffice,
AP-a_characterization_of_GFF} indeed show how modest assumptions suffice to characterize
the GFF.

\subsection{Conformal field theory: local fields and their correlation functions}

To properly specify a field theory, it is not sufficient to just give
an action~--- one also has to address what are the (locally)
observable quantities in it, i.e., what is the space of local fields of the theory.
By the state-field correspondence, this is the same data as the state space of the theory,
or its ``spectrum''.

Typically in a field theory there are some basic degrees of freedom, or basic fields,
in terms of which for example the action of the
theory is written in path integral formulations.
The free boson path integral~\eqref{eq: free boson path integral}
involves a single scalar field~$\gff$ in this
role of a basic field.
To illustrate the notion of local fields, it is easiest to start with examples.
It seems natural that the value $\gff(z)$ of the basic field at a point~$z$
would be an observable quantity located at~$z$.
Also the values of the derivatives~$\partial_\mu \gff (z)$,
$\partial_\mu \partial_\nu \gff (z)$, \ldots, of any order,
are determined by the configuration of the field~$\gff$ in an
infinitesimal neighborhood of~$z$, so they can be viewed as locally
observable quantities at~$z$.
Furthermore, one could consider polynomials (with suitable regularization)
in any of the above, such as~$\gff(z)^4$
or~$\big(\partial_\mu \gff(z)\big)^2$, etc.,
or even more complicated functions (again suitably regularized)
such as~$e^{\ii \alpha \gff(z)}$.
All of these examples are indeed quantities that one might want to assign
correlation functions to, and the informal 
examples illustrate
the idea that a local field is an observable quantity whose value is
determined locally (only accessing an infinitesimal neighborhood of a point)
from the appropriate basic degrees of freedom.

Note, however, that the basic degrees of freedom are
not necessarily themselves observable quantities~--- for example in gauge theories,
observable quantities must be independent of the (unphysical) gauge choices, whereas
the basic fields in the path integral depend on the gauge.
Similarly, for the GFF with Neumann boundary conditions (as well as for the GFF
on Riemann surfaces without boundaries), there is an ambiguity of
an additive constant in~$\gff$, so it makes sense to exclude the field value~$\gff(z)$
while allowing the values of its derivatives as local fields. We will indeed
do so: we will consider a bosonic CFT in which the local fields are the
\emph{polynomials in the derivatives of~$\gff$}. There are two reasons for this choice.
First, this allows us to talk about the same CFT
with different boundary conditions (in this article we explicitly state
results for Dirichlet and Neumann boundary conditions, but for example
mixtures of these could be handled similarly).
Second, this choice leads to a simple and concrete CFT.
The space of local fields of this CFT will simply be the chargeless Fock space
for two chiral Heisenberg algebras\footnote{The chiral part~$\ChiralFock$
of this state space is exactly the Heisenberg vertex operator algebra (VOA),
see e.g. \cite{Kac-VAs_for_beginners,
FB-vertex_algebras_and_algebraic_curves, LL-introduction_to_VOAs}.
General CFTs have more modules (for the appropriate VOA) in their state space~---
the bosonic CFT of interest here is particularly simple in this respect.},
which we will denote by 
\begin{align*}
\FullFock
\qquad\qquad \text{(the precise definition will be given in Section~\ref{sec:Fock background}).}
\end{align*}
In any domain~$\domain$, with Dirichlet or Neumann boundary conditions (abbreviated $\DorN$),
any local fields $F_1, \ldots, F_n \in \FullFock$ inserted at any distinct points
$z_1, \ldots, z_n \in \domain$ have correlation functions 
\begin{align*}
\CFTcorrBig{\domain}{\DorN}{F_1(z_1) \cdots F_n(z_n)}
\qquad \text{(the precise definition will be will be given in Section~\ref{sec:scal_lim})}
\end{align*}
in this free bosonic CFT.
Indeed a field theory should fundamentally
consist of the above kind of data:
a space of local fields, and correlation functions assigned to any
local fields at any finite number of distinct points in any domain.
This data should of course also satisfy
suitable axioms of a field theory, notably concerning the existence of
operator product expansions (OPE),
i.e., short-distance expansions of correlation functions
of local fields
in terms of other local fields contained in the theory\footnote{In
typical conformal field theories, there are particularly important local fields
called \emph{primary fields}, but they alone do not form a full conformal field
theory for example for the reason that their OPEs involve other local fields
(\emph{descendants}) besides just the primary fields themselves.}.

As we allow only polynomials in the gradient of~$\gff$ as local fields, we are
excluding not just the field~$\gff(z)$ itself, but also for example the
``vertex operators''~$e^{\ii \alpha \gff(z)}$.
Such vertex operators would be in many ways very interesting local fields.
Still, reasons to exclude them are as above~---
we want local fields that make sense with more general boundary
conditions (including Neumann), and we prefer to focus on
a simple well-behaved CFT.
In the cases where free field is defined without additive constant ambiguities
(for example with Dirichlet boundary conditions),
extending our calculations to cases involving vertex operators seems
feasible. But that would be a distraction from our primary objective, which is to
provide a first complete CFT scaling limit result for a lattice model.

\subsection{Local fields in the lattice model}

The discrete Gaussian Free Field (DGFF)
is the most straightforward discretization of the GFF
to finite subgraphs~$\ddomain \subset \meshsz \bZ^2$ of the square lattice
with mesh size~$\meshsz>0$~\cite{Sheffield-GFF_for_mathematicians}.
It is a Gaussian measure associated with the 
discrete analogue
\begin{align*}
S(\dgffptwise) =
  \frac{1}{8\pi} \sum \big( \dgffptwise(\zs) - \dgffptwise(\zsbis) \big)^2
\end{align*}
of~\eqref{eq: free boson path integral}, where the sum ranges over
pairs of nearest neighbor vertices $\zs,\zsbis \in \ddomain$ in the discrete domain
(the precise definition will be given in Section~\ref{sec:currents}). Again, we
consider~$\dgffptwise$ as defined only up to an additive constant,
i.e., our Gaussian process is just the discrete gradient of~$\dgffptwise$.
This makes sense with both Dirichlet and Neumann boundary conditions~($\DorN$)
in arbitrary discrete domains~$\ddomain$.

\begin{figure}[h!]
	\centering
	\begin{overpic}[scale=0.757, tics=10]{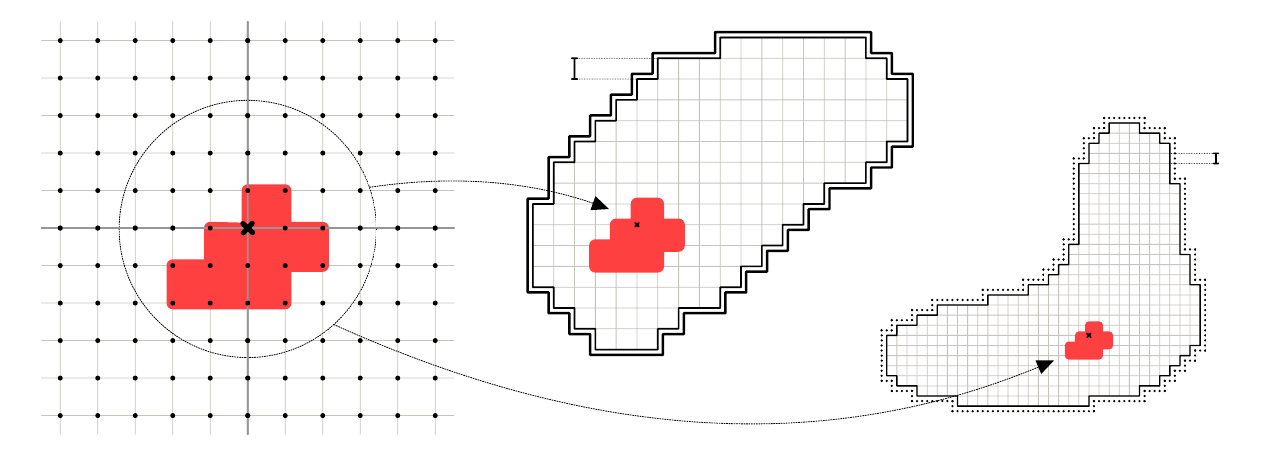}
		\put(0.4,32.2){$\Z^2$}
		\put(43,29.6){\footnotesize$\meshsz$}
		\put(71,31.6){$\domain_\meshsz$}
		\put(96,22.6){\footnotesize$\meshsz'$}
		\put(81.5,25.6){$\domain'_{\meshsz'}$}
		\put(75.2,1.5){\footnotesize$\mathtt{Dirichlet}$}
		\put(58.3,33.8){\footnotesize$\mathtt{Neumann}$}
	\end{overpic}
	\centering
	\caption{%
	A local field for a lattice model 
	encodes a rule to construct concrete random variables at arbitrary
	points~$\zs$ in 
	arbitrary discrete domains~$\ddomain$ of any mesh size~$\meshsz>0$
	and for any choice of boundary conditions. These random variables
	take into account 
	some fixed finite set of lattice sites around their points of insertion.}
	\label{fig: local fields intuition}
\end{figure}

In a lattice model, the notion of a local field obviously should not refer to
infinitesimal neighborhoods of a point~--- rather any finite lattice distance is
to be considered microscopic, and we should allow local fields to ``see''
finite lattice regions around their positions of insertion.
An example is provided by
the following discretization of the gradient squared of the DGFF.
Given a discrete domain~${\ddomain \subset \meshsz \bZ^2}$,
a choice of boundary conditions (Dirichlet or Neumann), and a point~$\zs \in \ddomain$
(such that all of its four neighbors also belong to~$\ddomain$),
the discrete gradient squared is the random variable
\begin{align}\label{eq: discrete gradient squared}
\Big( \frac{\dgffptwise(\zs+\meshsz) - \dgffptwise(\zs-\meshsz)}{2} \Big)^2
    + \Big( \frac{\dgffptwise(\zs+\ii\meshsz) - \dgffptwise(\zs-\ii\meshsz)}{2} \Big)^2
\end{align}
on the probability space of the DGFF in~$\ddomain$.
For different choices of mesh size~$\meshsz$, discrete domain, boundary conditions, and
insertion point, we will want to view the correspondingly obtained random variable as an
incarnation of the same abstract local field. So,
following~\cite{GHP-pattern_probabilities} and~\cite{HKV},
the idea of a local field of the lattice model is a fixed rule,
as illustrated schematically in Figure~\ref{fig: local fields intuition},
to construct random variables in a natural translation
invariant fashion from the model's random configuration restricted to a finite
lattice neighborhood.
The same rule is to be applied
in all discrete domains, with any boundary conditions,
and at any point~--- provided just that the domain is large
enough around the point so that the given finite neighborhood fits in it,
which in scaling limit considerations is guaranteed for all small
enough lattice mesh sizes~$\meshsz$.

For specificity and in accordance with the notion of local fields in our
chosen free boson CFT, we require the rule defining a DGFF local field moreover be
given by polynomials in the differences of the DGFF values.
The space of such local \emph{field polynomials}
is denoted by~$\dLocFi$
(the precise definition will be given in Section~\ref{sec:currents}).
A crucial subtlety, emphasized and treated in~\cite{HKV},
is that among such field polynomials, there may be
different ones which represent the same observable quantity, and should
therefore be identified as local fields.
Namely, it is possible that for two different field
polynomials $P_1, P_2 \in \dLocFi$,
in all sufficiently large discrete domains~$\ddomain$,
with all allowed boundary conditions, we have the coincidence of
all expected values involving the random variables associated
with $P_1$ and $P_2$ at a point~$\zs$ and other random variables
supported at least some microscopic distance~$\OO(\meshsz)$ away from
that point.
In this case $P_1$ and $P_2$ were called \emph{correlation equivalent}
in~\cite{HKV}, and the space~$\dNuFi \subset \dLocFi$ of field polynomials
which are correlation equivalent to zero was called~\emph{null fields}.
The appropriate definition of the space of \emph{local fields},
corresponding to distinct observable quantities, is then the quotient
\begin{align}
\dFields := \dLocFi / \dNuFi \, . 
\end{align}

In \cite{HKV} it was proven that the space~$\dFields$ carries representations
of the Heisenberg algebra and the Virasoro algebra: the Heisenberg algebra
generators act by discrete Laurent-monomial weighted discrete contour
integrations of discrete holomorphic currents, and the Virasoro generators
are obtained from them by a Sugawara construction.
By repeating the same with antiholomorphic variants, one
also obtains second copies of Heisenberg and Virasoro algebras acting
on the space~$\dFields$, which commute\footnote{%
To obtain the commutation of these holomorphic and antiholomorphic chiral
algebras acting of lattice model local fields,
however, it is crucial to fix one convention
about the discrete Laurent monomials from the published version of~\cite{HKV}.
The fixed convention will be given in our Proposition~\ref{prop: monomials},
and with this, we provide the details of the commutation in
Proposition~\ref{prop: comm relations}.} with the original copies.

\subsection{Main results}

As described above, both the space~$\dFields$ of lattice model local fields and
the space~$\FullFock$ of local fields of the free boson CFT
carry representations of two commuting copies of Heisenberg and Virasoro algebras.
Our first main result is the following.
\begin{thm*}[Theorem~\ref{thm: main theorem about Fock space structure}, informally stated] \ \\
The Fock space~$\FullFock$ of local fields of the free boson CFT and the
space~$\dFields$ of local fields of the gradient of the discrete Gaussian Free Field
are isomorphic,
\begin{align*}
\dFields \isom \FullFock \, ,
\end{align*}
as representations of two commuting copies of the Heisenberg algebra and as
representations of two commuting copies of the Virasoro algebra
with central charge~$c=1$.
\end{thm*}
To spell out the meaning in the most concrete terms,
the isomorphism of Theorem~\ref{thm: main theorem about Fock space structure}
provides a one-to-one correspondence
\begin{align*}
\set{\begin{array}{c}
      \text{lattice model} \\
      \text{local fields}
     \end{array}}
\; \xleftrightarrow{\ \text{ 1-to-1 }\ }\;
\set{\begin{array}{c}
      \text{conformal field theory} \\
      \text{local fields}
     \end{array}}
\end{align*}
whose two directions have the following interpretations:
\begin{itemize}
\item
Any random variable, which is a
polynomial in differences of the DGFF values at finitely many lattice neighbors
of a point, determines by a translation invariant rule an abstract field polynomial
$P \in \dLocFi$, and then by the above isomorphism also a unique associated 
local field $F \in \FullFock$ of the free boson CFT.
The CFT local field~$F$ is zero in the Fock space
if and only if the field polynomial is null\footnote{In practical terms,
being null means that the DGFF random variables given by that formula
in general discrete domains and with general boundary conditions have
vanishing correlations with anything else at least some microscopic distance away.},
$P \in \dNuFi$.
\item To any Fock space local field $F \in \FullFock$ of the free boson CFT,
one can associate a corresponding lattice model local field and
a representative abstract field polynomial $P \in \dLocFi$,
and therefore corresponding
random variables in the DGFF at any point in any sufficiently large discrete
domain. These random variables are 
polynomials in differences of the DGFF values at finitely many lattice neighbors
of the point, given by a translation invariant formula.
The field polynomial~$P$ associated to~$F$ is unique modulo
null fields.\footnote{In practical terms again, this
means uniqueness up to the addition of random variables that do not affect
correlations with anything else at least some microscopic distance away.}
\end{itemize}

The raison d'\^etre of local fields is the formation of correlation
functions. In turn, lattice discretizations of
quantum field theories are meant to recover such correlation functions as
appropriately renormalized scaling limits.
Our second main result concerns the
scaling limits of probabilistic correlations of local fields
in the lattice model when conceptually correct renormalization
is applied, and it identifies these scaling limits as CFT correlations.
By general principles, the right renormalization of a local field
should be determined via the structure of the Virasoro representations,
in terms of its $\VirL{0} + \VirBarL{0}$ eigenvalue, and our result
does precisely this.

Let us concretely illustrate scaling limit considerations with the field
polynomial~\eqref{eq: discrete gradient squared} from our earlier example.
Perhaps contrary to a naive expectation for a squared discrete gradient,
this local field has nonvanishing limits for its correlation functions without
any renormalization. These limits, however, are trivial in the
sense that they are constant as functions of the position where the
field is inserted~--- they in fact accidentally capture a component
corresponding to the CFT identity field!
A field polynomial which actually corresponds
to the regularized gradient squared field in the CFT
via the isomorphism~${\FullFock \isom \dFields}$ of our first main result, is
obtained by subtracting from~\eqref{eq: discrete gradient squared}
a suitable constant, $4 \pi - 8$.
After thus getting rid of the unwanted identity field component, we can apply the
renormalization anticipated for the gradient squared, i.e., divide by the
square of the lattice mesh~$\meshsz$. The random variables
\begin{align*}
\meshsz^{-2} \Bigg( \Big( \frac{\dgffptwise(\zs+\meshsz) - \dgffptwise(\zs-\meshsz)}{2} \Big)^2
    + \Big( \frac{\dgffptwise(\zs+\ii\meshsz) - \dgffptwise(\zs-\ii\meshsz)}{2} \Big)^2
    - 4 \pi + 8 \Bigg)
\end{align*}
then indeed have nontrivial scaling limits for their correlation functions.
For more general field polynomials,
the right procedure of subtracting accidental
components and applying the right renormalization is not a priori obvious.
Serendipitously, the isomorphism of our
first main result directly gives
the right counterparts to any CFT fields.
Our second main theorem, given below, then identifies the right renormalization
and scaling limits of correlations.
It states, in particular, that the homogeneous components for
renormalization are exactly the eigenspaces of
$\VirL{0} + \VirBarL{0}$ in our space~$\dFields$
of lattice model local fields. Any Fock space field can
be written as a sum of such homogeneous components.

\begin{thm*}[Theorem~\ref{thm: main theorem about scaling limits}, informally stated] \ \\
Let $F_1 , \ldots, F_n \in \FullFock$ be $n$ eigenvectors of
$\VirL{0} + \VirBarL{0}$ with respective eigenvalues $D_1 , \ldots, D_n$.
Let $\domain \subset \bC$ be an open
simply-connected proper subset of the plane,
and let $z_1, \ldots, z_n \in \domain$ be $n$ distinct points in it.
Let $(\ddomain; \zs_1^\meshsz, \ldots, \zs_n^\meshsz)$ be discrete domains
$\ddomain \subset \meshsz \bZ^2$ with $n$ marked vertices approximating
$(\domain;z_1,\ldots,z_n)$ in Carath\'eodory sense as $\meshsz \to 0$.
Fix a choice of boundary conditions, Dirichlet or Neumann~($\DorN$).
Then we have the following scaling limit of discrete GFF expected values
\begin{align*}
\frac{1}{\meshsz^{D_1 + \cdots + D_n}} \,
\EX_{\ddomain}^{\DorN} \Big[ F_1^{\ddomain}(\zs_1^\meshsz)
        \cdots F_n^{\ddomain}(\zs_n^\meshsz) \Big]
    \; \underset{\meshsz \to 0}{\longrightarrow} \;
    \CFTcorrBig{\domain}{\DorN}{F_1(z_1) \cdots F_n(z_n)} \,
\end{align*}
where each
$F_j^{\ddomain}(\zs_j^\meshsz)$, is a random
variable at~$\zs_j^\meshsz$ associated with the CFT
field~$F_j$ via the isomorphism
$\FullFock \isom \dFields$, and the right hand side is the CFT
correlation of the Fock space fields $F_1, \ldots, F_n$
in~$\domain$ with the chosen boundary conditions.
\end{thm*}

We view the combination of the two main results as a complete
realization of the free boson CFT as the scaling limit of
the discrete Gaussian Free Field.

\subsection{Organization of the article and outlines of the proofs}

Let us briefly outline the structure and ideas of the proofs of the two
main results, and simultaneously describe the organization of the article.

There are two sections about the necessary preliminaries.
Section~\ref{sec:Fock background} contains definitions and background
related to the Heisenberg algebra, Virasoro algebra, and bosonic Fock spaces.
Section~\ref{sec:discrete background} contains definitions, conventions
and background about discrete complex analysis.
Many of the details in these sections could just be consulted
when they become relevant to the main arguments.

The gradient of the discrete Gaussian Free Field and its local fields are
defined in Section~\ref{sec:currents}, and
the constructions of the two commuting Heisenberg and Virasoro algebra
representations on the space of these lattice model local fields are given.
This mostly amounts to recalling results
from~\cite{HKV} but with slight modifications to the exact setup
(gradient of DGFF instead of GFF) and one minor change
which is necessary to ensure the commutation of the two
copies of both algebras.

The proof of the first main result, giving the
isomorphism $\FullFock \isom \dFields$, is carried out in
Sections~\ref{sec:linear}~--~\ref{sec:higher}.
The starting point (Corollary~\ref{cor: easy inclusion}) is to note
that an embedding ${\FullFock \subset \dFields}$ is obtained
just using the universal property and irreducibility of the
Fock space~${\FullFock}$, as soon as a few very simple properties
about the Heisenberg algebra actions
on the constant field polynomial~$1 \in \dLocFi$ are verified.
Therefore the essence of the proof is to
get the opposite inclusion $\dFields \subset \FullFock$. The rough idea
is to show that dimensions of our lattice local fields
are bounded from above by the corresponding dimensions
in the embedded Fock space, which then rules out the existence of any
superfluous local fields in the lattice model.
Both spaces are actually infinite-dimensional,
so such dimension bounds must be obtained in
some appropriate finite-dimensional subspaces instead.

The first step,
undertaken in Section~\ref{sec:linear},
is to consider only those lattice local fields which correspond to
linear polynomials in the gradient of the DGFF.
The space~$\dLinFields \subset \dFields$ of such linear local fields
of the lattice model is still infinite-dimensional, but it admits a
filtration by finite-dimensional subspaces~$\dLinFieldsRad{r}$ of
linear local fields which have representative field polynomials supported
in a lattice ball of size~$r \in \bN$. On the one hand, a version of the
domain Markov property of the DGFF can be used to find
essentially canonical representatives supported on the boundary of
the lattice ball. This gives dimension upper bounds
$\dmn \big( \dLinFieldsRad{r} \big) \le 4r - 1$ in the filtration
(Lemma~\ref{lem: dimension upper bounds in filtration}).
On the other hand, for particular linear local fields in the embedded
Fock space $\FullFock \subset \dFields$ we perform explicit
calculations with discrete contour integration and discrete Laurent
monomials to get representatives of small enough radius of support,
which allow us to conclude that such local fields already
saturate the obtained dimension upper bound for~$\dLinFieldsRad{r}$
(Lemma~\ref{lem: dimension lower bounds in filtration}).
Therefore the linear local fields in the embedded Fock
space exhaust all linear local fields of the DGFF
(Theorem~\ref{thm: basis lin loc fields}). As an important
by-product, we simultaneously obtain an explicit characterization
of all linear null
fields (Corollary~\ref{coro: linear nulls}).\footnote{The result is
that the linear null fields are essentially just the discrete Laplacians
of the DGFF at all possible locations. This is not unexpected:
these vanishing Laplacians are the ``equations of motion
of the corresponding (classical) field theory''.}

The remaining task of treating higher degree lattice local fields is
then undertaken in Section~\ref{sec:higher}.
The key tool here is a suitable version
of normal ordering on the lattice local fields. The combinatorics
of the appropriate normal ordering is identical to the usual Wick
products, but it is important to use infinite square grid quantities
in the contractions.\footnote{The interpretation is that the infinite square grid
plays a role similar to the specification of a local coordinate
at the position of the field insertion. Conceptually,
local field correlations in CFT require such local coordinate specifications;
see, e.g., \cite{FB-vertex_algebras_and_algebraic_curves} and~\cite{KM-GFF_and_CFT}.}
The fundamental observation about normal ordering is that field polynomials with
linear null field factors behave like an ideal with respect to it:
any linear null factor in a normal ordered product renders the full
result null (Lemma~\ref{lem: normal ordering with null fields}).
This makes normal ordering well-defined for local
fields $F \in \dFields = \dLocFi / \dNuFi$, not merely for field
polynomials~$P \in \dLocFi$.
The proof of this well-definedness relies on the full classification of linear
null fields from the previous section.
Finally, one needs to relate normal ordering to the Heisenberg algebra actions
on local fields (Lemma~\ref{lem: basis elements as normal ordered products}).
In combination with the earlier result
that all linear local fields are in the embedded Fock space,
this yields the proof of the nontrivial inclusion $\dFields \subset \FullFock$,
and finishes the proof of the first main result: $\FullFock \isom \dFields$
(Theorem~\ref{thm: main theorem about Fock space structure}).

The starting point of the proof of the second main result, about scaling limits, 
is to write formulas for the CFT correlations of Fock space fields as multiple
contour integrals
(Proposition~\ref{prop: folklore} and Equation~\eqref{eq: CFT correlation}),
since the corresponding lattice local fields are by
construction given by discrete contour integrals.
The straightforward idea,
then, is to expand expected values of DGFF local fields by Wick's formula 
and interpret the discrete contour integrals appearing in them
as Riemann sum approximations of the integrals which give the CFT correlations.
The discrete integrals involve derivatives of discrete Green's functions with
Dirichlet or Neumann boundary conditions as well as discrete Laurent monomials.
Both converge in the scaling limit to their continuum counterparts,
with suitable scaling factors extracted from the Laurent monomials.
In view of this, the convergence of the discrete integrals is entirely
unsurprising, and the remaining part of the proof amounts to taking care
of essentially combinatorial
details from the specific discretizations. The scaling limit result
(Theorem~\ref{thm: main theorem about scaling limits}) then follows.

\subsection{Scaling limit results in the literature}

We finish the introduction by discussing prior results
on conformally invariant scaling limits of correlation functions obtained
for probabilistic lattice models.
This is a major research topic with a long history, so
it would be impossible to give a comprehensive account even of results
on a single model. We simply aim to highlight some notable results for
comparisons to the present work.

The two-dimensional critical Ising model is a quintessential
lattice model which is believed to converge to a CFT in its scaling limit,
so it provides a fruitful point of first comparisons.

A landmark result on the correlation functions of the Ising model
was the proof by Chelkak, Hongler, and
Izyurov of the existence of conformally covariant scaling limits for all
spin correlation functions~\cite{CHI-spin_correlations}.
That result had been preceded by the breakthrough work of Smirnov and others on
conformally convariant scaling limits for fermionic
observables and their
generalizations~\cite{Smirnov-holomorphic_fermions,CS-universality_in_the_Ising_model,
Izyurov-thesis} (which was in particular a key input to conformally invariant
random geometry descriptions)
and of energy correlations~\cite{Hongler-thesis,HS-energy_density}.
Later the result was generalized to the proof of
conformally covariant scaling limits of any mixed correlations of
spins, disorders, fermions, and energies,
in arbitrary finitely connected domains, with the most general boundary
conditions allowed in the corresponding
CFT~\cite{CHI-primary_field_correlations},
see also~\cite{BIVW-bosonization_of_primary_fields_for_Ising}.
This is essentially the full extent to which one can hope to understand the
scaling limits of primary field correlation functions in the critical Ising
model in planar domains.
Despite the spectacular success, we argue that this does not yet
establish a full CFT as the scaling limit of the Ising model, nor
does it fully describe the scaling limits of all locally defined random
variables in the Ising model.

Regarding the former point, note that
besides the finitely many primary fields,
a full CFT contains a vast amount of other local fields, including
infinite-dimensional spaces of descendant fields to every primary field.
This is crucial, since it is this infinite-dimensional space of local fields that carries
the algebraic structures that are the hallmarks of conformal field theory.
Given the algebraic structure, one should furthermore proceed to verify that
the scaling limits satisfy the characterizing
conditions of CFT correlation functions.

Regarding the latter point, note that there are locally
defined random variables in the Ising model more general than the spin at one point,
the energy (the product of spins at two lattice neighbors), and the other
specific ones mentioned above.
Such general local random variables and their scaling limits, ``pattern probabilities'',
have been studied in~\cite{GHP-pattern_probabilities}. The scaling
limits there were, however, taken with the generic renormalization for spin-flip even
and spin-flip odd patterns separately, and the limits with these renormalizations
therefore only retain the projections to the corresponding CFT primary fields.
While impressive, the result does not, therefore, describe the probabilities of arbitrary
(finely-tuned to cancel the leading scaling components) local
patterns in their true asymptotic scale as a function of the lattice mesh,
and it does not correspond to the renormalization of local fields by their
scaling dimensions as required for field theory.

Complementarily,
the articles~\cite{HKV, AdameCarrillo-discrete_symplectic_fermions}
do identify the algebraic structures of CFTs acting
on the infinite-dimensional spaces of local fields for
three lattice models: the Ising model,
discrete Gaussian Free Field, and discrete symplectic fermions
(within the double-dimer model). On the other hand, in these works
is it not proven that the space of local fields corresponds exactly to the
CFT state space, and scaling limits of correlations of local fields are not
treated.

Let us then compare with scaling limit results of discrete models
believed to converge to the Gaussian Free Field.

Discrete-valued random height functions on a (discrete) lattice
are natural higher-dimen\-sional analogues of random walks, just like the GFF is
a higher-dimensional generalization of the Brownian motion.
It is natural to conjecture ``functional central limit theorems''
stating their convergence to GFF in the scaling limit,
for at least suitable exact models or under suitable assumptions.
Recent progress has been made, e.g.,
in~\cite{DHLRR-logarithmic_variance_for_square_ice,
CPST-delocalization_of_uniform_graph_homomorphisms, DKMO-delocalization_6V},
but scaling limits for these discretizations remain challenging.

Dimer configurations on suitable planar graphs can be encoded,
following Thurston,
into discrete-valued discrete height functions (modulo an additive constant)
which bear resemblance with the random walk analogues but can be studied
by more powerful combinatorial tools related to discrete complex analysis.
In a celebrated result,
Kenyon proved the conformal invariance of uniform domino tilings
(dimers on the square grid) in the scaling limit,
and showed that dimer height functions on Temperleyan domains tend to the
GFF~\cite{Kenyon-conformal_invariance_of_domino_tiling,
Kenyon-dominos_and_the_GFF}.
In more general domains, the same results still hold~\cite{Russ-hedgehog}, and, on non-simply connected surfaces, such height functions
converge in the scaling limit to compactified free
fields~\cite{Dubedat-dimers_and_familier_of_CR_operators, Basok-dimers_on_Riemann_surfaces}.
Moreover, double dimers, i.e., the superposition of two independent dimer
covers, form loops which have been shown to converge in the scaling limit
to the conformal loop ensemble $\mathrm{CLE}_4$
\cite{Kenyon-conformal_invariance_od_loops_in_the_double_dimer_model,
Dubedat-double_dimers_CLEs_and_isomonodromic_deformations,
BC-tau_functions_a_la_Dubedat_and_probabilities,
BW-crossing_estimates_of_simple_CLEs},
which is naturally coupled
with the GFF, and the loop scaling limit result can be seen as
a convergence result of appropriate height level lines.
Many of these results are stated in terms of convergence in law
of globally defined
random configurations, not in terms of local fields correlations.
These and related works contain, however, also scaling limit results for
correlations of dimer height
gradients~\cite{Kenyon-conformal_invariance_of_domino_tiling} and
for monomer correlations and the correlations of
so-called electric operators (dimer analogues of
vertex operators)~\cite{Dubedat-dimers_and_familier_of_CR_operators}.
From the point of view of CFT,
these should again be viewed as scaling limit results
of primary field correlations.

There are also recent results on scaling limits of correlations
in percolation \cite{Camia-percolation_primary, CF-percolation_log_OPE}
and in the abelian sandpile model~\cite{PR-multipoint_abelian_sandpile,
CCRR-fDGFF_sandpile_UST}.
The scaling limits of these models are believed to be logarithmic conformal
field theories: they have non-diagonalizable Virasoro generators $\VirL{0}$
and $\VirBarL{0}$, and correspondingly local fields with much more intricate
scaling behavior. Gaining full CFT
scaling limit results for such lattice models could be particularly valuable,
as it would shed light on a class of CFTs that remains extremely poorly
understood while containing important examples for statistical physics.


\textbf{Acknowledgments:}
DAC, DB, and KK were supported by the
Research Council of Finland (project 346309: Finnish Centre
of Excellence in Randomness and Structures, ``FiRST'').
The authors would like to thank Mikhail Basok, Nathana\"el Berestycki, Dmitry Chelkak,
Cl\'ement Hongler, Konstantin Izyurov, Richard Kenyon, Antti Kupiainen, and Wioletta Ruszel for valuable discussions.

\section{Preliminaries: the Fock space}
	\label{sec:Fock background}
	A central object in the present work is a Fock space,
which should be interpreted
as the space of polynomial local fields in
the gradient and higher order derivatives of the bosonic free field.
This interpretation will eventually be made concrete
in Section~\ref{subsec: local fields of GFF}.
By way of preliminaries, we start in this section by explicitly defining
the Fock space as a vector space
and as a representation of the key algebraic structures present in the
free boson conformal field theory (CFT).

In Section~\ref{ssec: Heisenberg basics}, we fix notation and conventions
related to the Heisenberg algebra and its chargeless Fock representation,
and we recall some essential properties of these.
The Heisenberg algebra plays the role of a chiral symmetry algebra of the free boson
CFT. In the full CFT with holomorphic and antiholomorphic
chiralities, one has two commuting copies of the
Heisenberg algebra,
and the full Fock space of interest to us is built from Fock representations for these two
chiral parts.
A textbook reference for (a more general version of) these topics is
\cite[Ch.~6.2~--~6.3]{LL-introduction_to_VOAs}.

In Section~\ref{ssec: Sugawara basics} we recall
how the Fock space becomes a representation of two commuting copies of Virasoro algebra
via the Segal--Sugawara construction.
Eigenvalues of particular Virasoro generators give rise to gradings
of the Fock space by conformal dimensions, or equivalently by the scaling dimension and
spin of the local fields. These scaling dimensions will later feature crucially in
the renormalization of fields in our scaling limit result.

To finish this background section, we make explicit observations
in Section~\ref{ssec: primary fields} about Virasoro primary fields in the Fock space.

\subsection{Heisenberg algebra and Fock space}
\label{ssec: Heisenberg basics}

The \term{Heisenberg algebra} is the $\bC$-vector space
\begin{align}\label{eq: Heisenberg basis}
\Hei = \Big( \bigoplus_{k \in \bZ} \bC \HeiJ{k} \Big) \oplus \, \bC \HeiK ,
\end{align}
equipped with the Lie algebra structure uniquely determined by the Lie brackets
\begin{align}\label{eq: Heisenberg bracket}
[\HeiJ{k},\HeiJ{\ell}] = k \, \delta_{k+\ell} \, \HeiK ,
\qquad\qquad
[\HeiK,\HeiJ{\ell}] = 0 ,
\end{align}
where we use the Kronecker delta with notation
\begin{align*}
\delta_n = \begin{cases}
            1 & \text{ if } n = 0 \\
            0 & \text{ if } n \ne 0 .
           \end{cases}
\end{align*}

As is common, we will only consider representations of $\Hei$ where the central element $\HeiK$ acts as
the identity, so the reader should feel free to think ``$\HeiK = 1$''. Somewhat more formally,
we consider the associative algebra
\begin{align}\label{eq: Heisenberg UEA}
 \UHei = \mathcal{U}(\Hei) \big/ (\HeiK - 1)
\end{align}
obtained as the quotient of the universal enveloping algebra $\mathcal{U}(\Hei)$
by the two-sided ideal generated by $\HeiK - 1$. A Poincar\'e--Birkhoff--Witt (PBW) basis
of~$\UHei$ consists of words
\begin{align*}
\cdots \HeiJ{-2}^{n_{-2}} \, \HeiJ{-1}^{n_{-1}} \, \HeiJ{0}^{n_{0}} \, \HeiJ{1}^{n_{1}} \, \HeiJ{2}^{n_{2}} \cdots 
\end{align*}
with only finitely many of the exponents $n_k \in \Znn$ nonzero.

The most fundamental representation of the Heisenberg algebra~$\Hei$ is the
(chargeless) \term{Fock representation}
\begin{align}\label{eq: Fock representation}
\ChiralFock = \UHei \big/ (\HeiJ{k} : k \ge 0)
\end{align}
obtained as a quotient of the $\UHei$-module $\UHei$ by the left ideal generated by the
$\HeiJ{k}$ with nonnegative indices~$k$.
The vector $\voavac \in \ChiralFock$,
i.e., the equivalence class of the algebra unit $1 \in \UHei$ in the
quotient~\eqref{eq: Fock representation},
is called the \term{vacuum} vector
in the Fock representation~$\ChiralFock$.
From the PBW basis of~$\UHei$, one gets a basis of $\ChiralFock$ consisting of
the vectors
\begin{align}\label{eq: chiral Fock basis}
\HeiJ{-k_m} \cdots \HeiJ{-k_2} \, \HeiJ{-k_1} \voavac 
\qquad \text{ with } \qquad
m \in \Znn, \; 0 < k_1 \le k_2 \le \cdots \le k_m .
\end{align}
Note that the zero-mode $\HeiJ{0}$ is in the centre of~$\Hei$ and
acts as zero on the whole chargeless Fock space~$\ChiralFock$.
Let us also mention, although this fact will not be directly used, that
the (chargeless) Fock space~$\ChiralFock$ admits the structure of
a \emph{vertex operator algebra (VOA)} \cite[Theorem~6.3.2]{LL-introduction_to_VOAs};
it plays the role of the chiral symmetry algebra of the free boson conformal field
theory (CFT).

As usual in bulk conformal field theory, we will in fact need two chiralities: the holomorphic
and the antiholomorphic. Algebraically both are exactly the same, but we distinguish the latter
by overline in the notation. The full (two-chiral) Lie algebra is the direct sum
\begin{align*}
\Hei \oplus \AntHei
\end{align*}
of two commuting copies of the Heisenberg algebra~\eqref{eq: Heisenberg basis}, and the appropriate associative algebra is
the tensor product
\begin{align*}
\UHei \tens \AntUHei 
\end{align*}
of two commuting copies of the associative algebra~\eqref{eq: Heisenberg UEA}.
The \term{(full) Fock space}
\begin{align*}
\FullFock
\end{align*}
is also the tensor product of two copies of the Fock representation~\eqref{eq: Fock representation},
with $\UHei$ acting on the first tensor factor and $\AntUHei$ acting on the second.
Let us denote by $\FockId = \voavac \otimes \voavac \in \FullFock$ the vector obtained as the
tensor product of the vacuum vectors of the two chiral Fock representations.
A natural basis of the full Fock space
is obtained by tensoring the bases of the form~\eqref{eq: chiral Fock basis} in both chiral
Fock representations; explicitly,
denoting the generators of the two commuting copies by $\HeiJ{k}$ and $\AntHeiJ{k}$,
that basis consists of vectors
\begin{align}\label{eq: full Fock basis}
& \HeiJ{-k_m} \cdots \HeiJ{-k_2} \, \HeiJ{-k_1} \, 
\AntHeiJ{-k'_{m'}} \cdots \AntHeiJ{-k'_2} \, \AntHeiJ{-k'_1} \FockId \; \in \, \FullFock \qquad
\text{ with } \\ \nonumber
& m, m' \in \Znn, \quad 0 < k_1 \le k_2 \le \cdots \le k_m, \quad 0 < k'_1 \le k'_2 \le \cdots \le k'_{m'} .
\end{align}

We note the following simple fact, which will be used for the easy half of our
first main result.
\begin{lemma}\label{lem: abstract nonsense}
Suppose that $V$ is a representation of $\UHei \tens \AntUHei$, 
and $v \in V \setminus \set{0}$ is a nonzero vector such that $\HeiJ{k} v = 0$ and $\AntHeiJ{k} v = 0$
for all $k \ge 0$. Then there exists a unique map
$\FullFock \to V$ of representations of $\UHei \tens \AntUHei$
such that $\FockId \mapsto v$. This map is injective and onto the subrepresentation
$(\UHei \tens \AntUHei) v \subset V$. In particular if $V$ is irreducible, then
the map gives an isomorphism $\FullFock \cong V$
of representations of $\UHei \tens \AntUHei$.
\end{lemma}
\begin{proof}
The condition $\FockId \mapsto v$ and the requirement of being a map of $\UHei \tens \AntUHei$
representations immediately fixes the values of the map on basis vectors~\eqref{eq: full Fock basis},
\begin{align*}
\HeiJ{-k_m} \cdots \HeiJ{-k_2} \, \HeiJ{-k_1} \, 
\AntHeiJ{-k'_{m'}} \cdots \AntHeiJ{-k'_2} \, \AntHeiJ{-k'_1} \FockId \; \mapsto \; 
\HeiJ{-k_m} \cdots \HeiJ{-k_2} \, \HeiJ{-k_1} \, 
\AntHeiJ{-k'_{m'}} \cdots \AntHeiJ{-k'_2} \, \AntHeiJ{-k'_1} v ,
\end{align*}
so the uniqueness of such a map $\FullFock \to V$ is clear.
The existence follows from a universal property (note that the properties assumed
of the vector~$v$ correspond exactly to the generators of the left ideal factored out in the
construction~\eqref{eq: Fock representation} of the Fock representation).

It is well-known and easy to check that $\ChiralFock$ is an irreducible representation of
$\UHei$, so $\FullFock$ is an irreducible representation of $\UHei \tens \AntUHei$. Therefore, the
kernel of the nonzero map ${\FullFock \to V}$ is zero, giving injectivity. By construction
the map is onto the subrepresentation ${(\UHei \tens \AntUHei) v \subset V}$.
If $V$ is irreducible, this nonzero subrepresentation must be the whole space
$(\UHei \tens \AntUHei) v = V$, so we get an isomorphism.
\end{proof}

\subsection{Virasoro algebra and the Sugawara construction}
\label{ssec: Sugawara basics}

By construction, the Fock representation $\ChiralFock$ is a representation of the Heisenberg algebra.
It also becomes a representation of the Virasoro algebra by the Sugawara construction, outlined below.

The \term{Virasoro algebra} is the $\bC$-vector space
\begin{align}\label{eq: Virasoro basis}
\Vir = \Big( \bigoplus_{n \in \bZ} \bC \VirL{n} \Big) \oplus \, \bC \VirC ,
\end{align}
equipped with the Lie algebra structure uniquely determined by the Lie brackets
\begin{align}\label{eq: Virasoro bracket}
[\VirL{n},\VirL{m}] = (n-m) \, \VirL{n+m} +  \delta_{n+m} \, \frac{n^3 - n}{12} \, \VirC ,
\qquad\qquad
[\VirC,\VirL{n}] = 0 ,
\end{align}
where the Kronecker delta notation convention is as in~\eqref{eq: Heisenberg bracket}. In CFT,
one considers representations where the central element~$\VirC$ acts as 
a fixed scalar multiple of identity, and the value~$c$ of the scalar is called
the \term{central charge} of the CFT. For CFT of interest to us ---the free boson---, the
central charge is~$c=1$.

The following is the most basic variant of Sugawara constructions.
\begin{lemma}[Sugawara construction]\label{lemma: sugawara Fock}
On the Fock represetation~$\ChiralFock$ of~\eqref{eq: Fock representation}, the
operators defined by the formulas
\begin{align}\label{eq: Sugawara on chiral Fock space}
\VirL{n} = \frac{1}{2} \Big( \sum_{k \ge 0} \HeiJ{n-k} \, \HeiJ{k} + \sum_{k < 0} \HeiJ{k} \, \HeiJ{n-k} \Big)
\qquad \text{ and } \qquad
\VirC = \idof{\ChiralFock} 
\end{align}
are well-defined (their action on any vector has only finitely many nonzero terms) and
they render~$\ChiralFock$ a representation of the Virasoro algebra with central charge~$c=1$.
\end{lemma}
\begin{proof}
This is well-known; see, e.g.,
\cite[Section~5.7]{Kac-VAs_for_beginners} or \cite[Theorem~6.2.16]{LL-introduction_to_VOAs}
for proofs in a more general vertex operator algebra setting, or
\cite[Theorem~4.10]{HKV} for a similar calculation in a specific setting
directly related to that of the present article.
\end{proof}

The eigenvalues of~$\VirL{0}$ give an important grading on the Fock representation~$\ChiralFock$.
A direct calculation shows that the basis vectors~\eqref{eq: Heisenberg basis} are
$\VirL{0}$-eigenvectors with eigenvalues $\Delta = k_1 + \cdots + k_m$. We thus
make $\ChiralFock$ an $\Znn$-graded vector space by
\begin{align}\label{eq: chiral Fock grading}
\ChiralFock = \; & \bigoplus_{\Delta \in \Znn} \ChiralFock_\Delta , \qquad \text{ where } \\ \nonumber
\ChiralFock_\Delta = \; & \spn \set{ \, \HeiJ{-k_m} \cdots \HeiJ{-k_2} \, \HeiJ{-k_1} \voavac \; \bigg| \;
    m \in \Znn, \; 0 < k_1\le \cdots \le k_m, \; \sum_{j=1}^m k_j = \Delta } ,
\end{align}
and then we have
\begin{align*}
\VirL{0} v = \Delta v \qquad \text{ for any } v \in \ChiralFock_\Delta .
\end{align*}
From the Virasoro commutation relations~\eqref{eq: Virasoro bracket} it also follows easily that
$ \VirL{n} \ChiralFock_\Delta \subset \ChiralFock_{\Delta - n} $, 
so $\VirL{n}$ is a homogeneous operator of degree~$-n$ with respect to the grading~\eqref{eq: chiral Fock grading}.

On the full (two-chiral) Fock space~$\FullFock$, we similarly obtain two commuting actions of the
Virasoro algebra, both with central charge~$c=1$. We denote the generators of these
by~$\VirL{n}$ and $\VirBarL{n}$ (we do not introduce notation for the central elements, since they simply act
as identity on~$\FullFock$).
On the full Fock space~$\FullFock$ we then have gradings by $\VirL{0}$ and $\VirBarL{0}$-eigenvalues,
\begin{align}\label{eq: Fock bigrading}
\FullFock = \; & \bigoplus_{\Delta, \bar{\Delta} \in \Znn} \ChiralFock_\Delta \otimes \AntiChiralFock_{\bar{\Delta}} .
\end{align}
The above bigrading by the pairs of $\VirL{0}$ and $\VirBarL{0}$-eigenvalues is said to be according
to \term{conformal dimensions}. The $\Znn$-grading by eigenvalues of the sum $\VirL{0} + \VirBarL{0}$
is an $\Znn$-grading by \term{scaling dimensions}, and it will be important for us to determine the
appropriate renormalization of lattice local fields in the scaling limit.
The $\bZ$-grading by the eigenvalues of the difference $\VirL{0} - \VirBarL{0}$ is also physically
significant, interpreted as giving the \emph{spins} of the corresponding local fields,
but this last $\bZ$-grading will not be needed in the present work.

\subsection{Primary fields in the Fock space}
\label{ssec: primary fields}

Let us finish with a comment about primary fields, as
these are generally of interest in CFTs.
The Fock space~$\FullFock$ is irreducible
as a representation of~$\UHei \tens \AntUHei$, and has the vacuum~$\FockId$ as the unique
(up to scalars) \term{Heisenberg primary state}, i.e., an eigenvector
for~$\HeiJ{0}$ and~$\AntHeiJ{0}$
annihilated by~$\HeiJ{k}$ and~$\AntHeiJ{k}$ for~$k > 0$.
The situation is different for
\term{Virasoro primary states}, i.e., eigenvectors for~$\VirL{0}$ and~$\VirBarL{0}$
which are annihilated by~$\VirL{n}$ and~$\VirBarL{n}$ for~${n > 0}$.
The Heisenberg primary state~$\FockId$ is also Virasoro primary, but
in addition there is a countably infinite number of linearly independent
Virasoro primaries in the Fock space, as made explicit in the following.

\begin{rmk}\label{rmk: Virasoro primaries in Fock space}
For every~$a,b \in \Znn$ there exists a Virasoro primary
state
in the Fock space ${\FullFock}$ with conformal weights
$\Delta = a^2$, $\bar{\Delta} = b^2$, and this state is
unique up to a multiplicative constant.
\hfill{$\diamond$}
\end{rmk}
The above is a known fact and not logically used in our main results, so we give only a brief
outline of the argument leading to it.
\begin{proof}[Sketch of proof of Remark~\ref{rmk: Virasoro primaries in Fock space}]
Due to the tensor product form of the
Fock space~$\FullFock$, it is enough to show that the (chiral) Fock
representation~$\ChiralFock$
has unique (up to scalars) Virasoro highest weight states of all 
highest weights~$\Delta=a^2$ for $a \in \Znn$.
The Fock representation~$\ChiralFock$ has an invariant inner product,
with respect to which $\VirL{n}^\dagger = \VirL{-n}$ for $n \in \Z$. Consequently,
any Virasoro submodule in~$\ChiralFock$ has a complementary submodule, and any Virasoro
submodule which is a highest weight module must be irreducible. The argument from here on
relies on dimension counting for the $\VirL{0}$-eigenspaces
in~$\ChiralFock$. On the one hand, the basis~\eqref{eq: chiral Fock basis}
readily shows that the dimension of $\Kern(\VirL{0}-\Delta) \subset \ChiralFock$
is $p(\Delta)$, the number of integer partitions of~$\Delta$.
On the other hand, the irreducible Virasoro highest weight modules of
central charge~$c=1$ and highest weight~$h=a^2$, with $a \in \Znn$, has
dimension~$p\big(\Delta-a^2\big)-p\big(\Delta-(a+1)^2\big)$ for the corresponding
eigenspace (this is the only needed fact whose proof is not completely elementary).
Starting from the fact that the vacuum~$\voavac \in \ChiralFock$ is a highest weight vector with
highest weight~$h=0$, and inductively looking for the complementary subspaces to
the subspace already found, and comparing dimensions of $\VirL{0}$-eigenspaces,
one sees that also (up to scalars unique) highest weight vectors of highest weights
$1,4,9,25,36,\ldots$ (and only these highest weights)
can be found in~$\ChiralFock$.
\end{proof}

Formulas for any of these Virasoro primary states in the Fock space~$\FullFock$
can be obtained by straightforward computation: examples include
\begin{align}
\label{eq: primary vacuum}
\FockId & & \text{ primary with conformal weights }
    & \; \Delta = 0, \; \bar{\Delta} = 0 \\
\label{eq: primary holom current}
\HeiJ{-1} \FockId & & \text{ primary with conformal weights }
    & \; \Delta = 1, \; \bar{\Delta} = 0 \\
\label{eq: primary antiholom current}
\AntHeiJ{-1} \FockId & & \text{ primary with conformal weights }
    & \; \Delta = 0, \; \bar{\Delta} = 1 \\
\label{eq: quadratic primary}
\HeiJ{-1} \AntHeiJ{-1} \FockId & & \text{ primary with conformal weights }
    & \; \Delta = 1, \; \bar{\Delta} = 1 \\
\HeiJ{-1} \, \Big(\AntHeiJ{-1}^4 + \frac{3}{2} \AntHeiJ{-2}^2 - 2 \AntHeiJ{-3}\AntHeiJ{-1}\Big) \, \FockId
    & & \text{ primary with conformal weights }
    & \; \Delta = 1, \; \bar{\Delta} = 4 
\end{align}
\begin{align} \nonumber
\bigg(
\HeiJ{-1}^9
+ 9 \, \HeiJ{-2}^2 \HeiJ{-1}^5 
- \frac{135}{4} \, \HeiJ{-2}^4 \HeiJ{-1}
- 12 \, \HeiJ{-3} \HeiJ{-1}^6 
+ 90 \, \HeiJ{-3} \HeiJ{-2}^2 \HeiJ{-1}^2
+ 40 \, \HeiJ{-3}^3 
- 90 \, \HeiJ{-4} \HeiJ{-2} \HeiJ{-1}^3 & \\ \nonumber
- 90 \, \HeiJ{-4} \HeiJ{-3} \HeiJ{-2} 
+ \frac{135}{2} \, \HeiJ{-4}^2 \HeiJ{-1} 
+ 36 \, \HeiJ{-5} \HeiJ{-1}^4 
+ 54 \, \HeiJ{-5} \HeiJ{-2}^2 
-72 \, \HeiJ{-5} \HeiJ{-3} \HeiJ{-1} &\,
\bigg) \, \FockId \\
\label{eq: primary nine}
\text{ primary with conformal weights } \Delta = 9, \; \bar{\Delta} = 0 \, .
\end{align}

A few of these will appear throughout the present work.
The vacuum~\eqref{eq: primary vacuum} will correspond to the identity
field of the CFT, as made concrete in Section~\ref{subsec: local fields of GFF}.
The fields corresponding to~\eqref{eq: primary holom current}
and~\eqref{eq: primary antiholom current}, will in various guises play central
roles both in the discrete and in the continuum; these fields are
called the holomorphic and antiholomorphic currents. The infinitely many other
Virasoro primaries will not directly show up in our construction, but we find it
worthwhile to remark that those primary fields are nevertheless present in our
space of local fields.

\section{Preliminaries: discrete complex analysis}
	\label{sec:discrete background}
	A key feature of the lattice models that enabled
equipping the spaces of their local fields with
representations of Virasoro algebra in~\cite{HKV} is a form of
exact solvability expressed in terms of discrete complex analysis.
This second background section collects notions and results about discrete
complex analysis that we will build on.

Section~\ref{subsec: lattices} contains the definitions of the sublattices of square grids
that we use as well as our conventions about discretizations of differential operators.
In Section~\ref{subsec: discrete domains}, we introduce discrete domains and
discrete Green's functions with both Dirichlet and Neumann boundary conditions.
The necessary scaling limit results for these discrete Green's functions
are addressed in Section~\ref{subsec: scaling limit Green functions}.
The final notions of discrete complex analysis needed
are discrete contour integration and discrete Laurent monomials,
and we give the precise conventions and results about these
in Section~\ref{subsec: discrete monomials}.

\subsection{Square grids and discrete differential operators}
\label{subsec: lattices}\label{subsec: difference operators}
Our discrete setup is based on the square grid~$\bZ^2$ and a few related graphs.
The classification of lattice model local fields can be
formulated on the unit-mesh square grid~$\bZ^2$,
but for the scaling limit results of Section~\ref{sec:scal_lim}, we need a general
\term{mesh size}~$\meshsz > 0$~---
the scaling limit amounts to letting~$\meshsz$ tend to zero.

All of our square-grid graphs are embedded in the complex plane~$\bC$. Throughout the article,
without further comments, we often identify vertices with complex numbers via the embedding,
and edges and faces (square plaquettes) with complex
numbers corresponding to the midpoints of the embedded edges or square plaquettes.

The infinite \term{square lattice} of mesh size $\meshsz>0$ is the set
\begin{align}\label{eq: square grid with mesh}
\SqLatMesh := \set{ n \meshsz + \ii m \meshsz \; \big| \; n,m \in \bZ} \subset \bC .
\end{align}
We view $\SqLatMesh$ as a nearest-neighbor graph: an edge connects $\zs,\zsbis \in \SqLatMesh$ when~$|\zs-\zsbis|=\meshsz$,
and this adjacency relation is denoted by $\zs \sim \zsbis$.
Since we use also various other related grids for discrete complex analysis, we, for clarity, occasionally denote $\SqLatMeshPrim := \SqLatMesh$ and call this the \term{primal lattice}
to distinguish it from the variants.
Two related grids are
\begin{align*}
\SqLatMeshDual := \; & \set{ k \meshsz + \ii \ell \meshsz \; \Big| \; k,\ell \in \bZ + \half}
    & \text{\term{(dual lattice)}} \\
\SqLatMeshMedial := \; & \set{ \frac{\zs+\zsbis}{2} \; \bigg| \; \zs,\zsbis \in \SqLatMesh, \; \zs \sim \zsbis}
    & \text{\term{(medial lattice)}}
\end{align*}
and we refer to the elements of~$\SqLatMeshDual$ as {dual vertices} or {plaquette centers},
and to the elements of~$\SqLatMeshMedial$ as {medial vertices} or {edge midpoints}~---
see Figure~\ref{fig: lattices} for an illustration.
There is a bipartition
$\SqLatMeshMedial = \big(\meshsz (\Z + \half) \times \meshsz \Z \big) \cup \big(\meshsz \Z \times \meshsz (\Z + \half) \big)$
of the medial lattice to midpoints of \term{horizontal and vertical edges}.
Finally, the bipartite grid
\begin{align*}
\SqLatMeshDiamond := \; & \SqLatMeshPrim \cup \SqLatMeshDual
    \qquad \qquad \qquad & \text{\term{(diamond lattice)}}
\end{align*}
is essentially dual to the medial lattice~$\SqLatMeshMedial$, and it is a natural domain of definition
of some of our discrete functions.
The elements of $\SqLatMeshDiamond$ are referred to as diamond vertices.

\begin{figure}[h!]
\centering
\begin{overpic}[scale=0.776, tics=10]{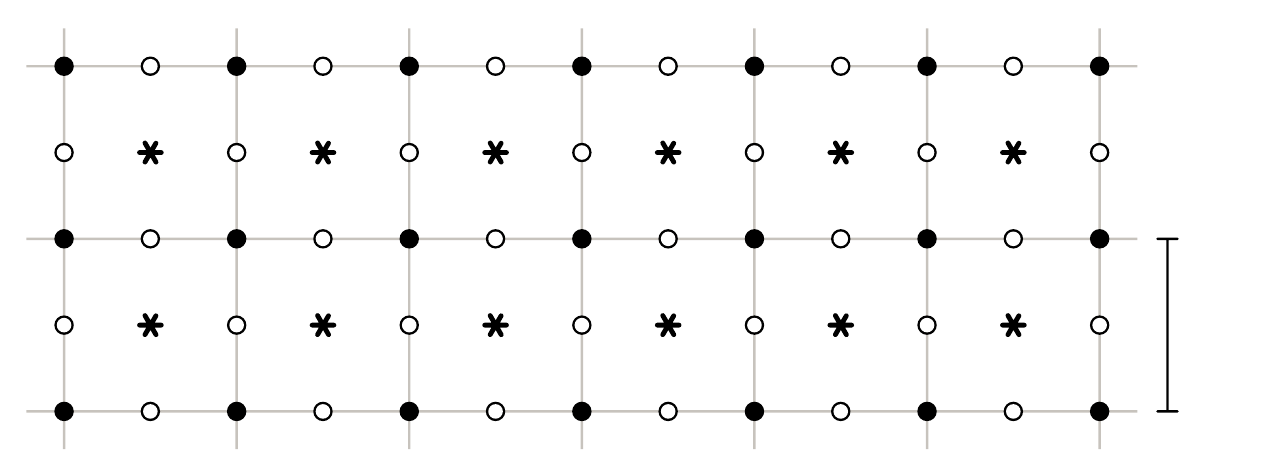}
	\put(93.5,10.2){$\meshsz$}
\end{overpic}
\centering
\caption{The $\meshsz$-mesh square grid with the sublattices
$\SqLatMeshPrim$, $\SqLatMeshDual$, and $\SqLatMeshMedial$.}
\label{fig: lattices}
\end{figure}

Finite-difference operators are analogues of differential operators in the
discrete setup, i.e., on the above lattices.
For concreteness, let us only write the defining formulas of these finite-difference
operators acting on complex-valued functions, although the action by the same formulas
will also be used on functions with values in other vector spaces.
The (combinatorially normalized) discrete 
\term{holomorphic} and \term{antiholomorphic deri\-vatives} of a function $f \colon \SqLatMeshDiamond \to \bC$
on the diamond lattice are functions $\gdee f , \gdeebar f \colon \SqLatMeshMedial \to \bC$
on the medial lattice given by the formulas
\begin{align}\nonumber
\gdee f(\zs) := \; &
    \frac{f(\zs+\frac{\meshsz}{2})-f(\zs-\frac{\meshsz}{2})}{2} - \ii\, \frac{f(\zs+\frac{\ii \meshsz}{2})-f(\zs-\frac{\ii \meshsz}{2})}{2}\,,
	\\
\label{eq: discrete Wirtinger derivatives}
\gdeebar f(\zs) := \; &
    \frac{f(\zs+\frac{\meshsz}{2})-f(\zs-\frac{\meshsz}{2})}{2} + \ii\, \frac{f(\zs+\frac{\ii \meshsz}{2})-f(\zs-\frac{\ii \meshsz}{2})}{2}\,.
\end{align}
Similarly, for a function $f \colon \SqLatMeshMedial \to \bC$ on the medial lattice,
$\gdee f, \gdeebar f \colon \SqLatMeshDiamond \to \bC$
are functions on the diamond lattice defined by exactly the same formulas as above.
If a function $f$ defined either on the diamond or the medial lattice 
satisfies $\gdeebar f(\zs) = 0$ at a medial or diamond vertex~$\zs$, it is said to be \term{discrete holomorphic} at~$\zs$.
If it satisfies $\gdee f(\zs) = 0$, then it is said to be \term{discrete antiholomorphic} at $\zs$.

The discrete \term{Laplacian} of a function~$f$ on any of the lattices
($\SqLatMeshPrim$, $\SqLatMeshDual$, $\SqLatMeshMedial$, or $\SqLatMeshDiamond$) is the
function $\gLapl f$ on the same lattice defined by
\begin{align}\label{eq: combinatorial Laplacian}
\gLapl f(\zs) := \; & 
    f(\zs + \meshsz) + f(\zs + \ii \meshsz) + f(\zs - \meshsz) + f(\zs - \ii \meshsz) - 4 \, f(\zs)\,.
\end{align}
Note that discrete Laplacian admits the following factorization:
\begin{align*}
\gLapl = 4 \, \gdee\gdeebar = 4 \, \gdeebar\gdee\,.
\end{align*}

At times we furthermore need versions of the holomorphic and antiholomorphic derivative
operators acting on functions defined only on the primal graph.
For $f \colon \SqLatMeshPrim \to \bC$
we define $\pdee f \colon \SqLatMeshMedial \to \bC$ by
\begin{align}\label{eq: primal lattice dee}
\pdee f(\zs) \coloneqq
	f \Big( \zs+\frac{\meshsz}{2} \Big) - f \Big( \zs-\frac{\meshsz}{2} \Big)
\quad \text{ and } \quad
\pdee f(\zs) \coloneqq
-\ii \, \bigg( f \Big( \zs+\frac{\ii\meshsz}{2} \Big) - f \Big( \zs-\frac{\ii\meshsz}{2} \Big) \bigg)
\end{align}
when $\zs$ is, respectively, a horizontal and a vertical edge.
Similarly, we define $\pdeebar f \colon \SqLatMeshMedial \to \bC$ by
\begin{align}\label{eq: primal lattice deebar}
\pdeebar f(\zs) \coloneqq
	f \Big( \zs+\frac{\meshsz}{2} \Big) - f \Big( \zs-\frac{\meshsz}{2} \Big)
\quad \text{ and } \quad
\pdeebar f(\zs) \coloneqq
\ii \, \bigg( f \Big( \zs+\frac{\ii\meshsz}{2} \Big) - f \Big( \zs-\frac{\ii\meshsz}{2} \Big) \bigg)
\end{align}
when $\zs$ is, respectively, a horizontal and a vertical edge.
For a function $f \colon \ZDiamond \to \C$, we then have
the following variants of factorizations of the discrete Laplacian
\begin{align}\label{eq: factorizations of Laplacian with primal graph derivatives}
\gdee\pdeebar f(\zs) = \gdeebar\pdee f(\zs)
    = \begin{cases} 
        \frac{1}{2}\gLapl f(\zs) & \zs\in\ZPrimary \\
        \mspace{20mu}0 & \zs\in\ZDual \, .
      \end{cases}
\end{align}

\subsection{Discrete domains and Green's functions}\label{subsec: discrete domains}

By a \term{discrete domain} 
with mesh size $\meshsz>0$ we mean a
planar region bounded by a polygonal Jordan curve made of the edges of the $\meshsz$-mesh square
grid~$\SqLatMesh$, 
such that the induced subgraph of~$\SqLatMesh$ consisting of the primal vertices
inside the Jordan curve is connected.
We write $\ddomain \subset \SqLatMesh$ for the set of primal vertices in the closure of the
Jordan domain, and $\bdry\ddomain$ for the set of primal vertices on the Jordan curve.
Vertices in $\ddomain \setminus \bdry \ddomain$ are called interior vertices.
By a mild abuse of terminology, we refer to $\ddomain$ as the discrete domain.

Discrete differential operators act naturally on functions defined on discrete domains, too,
once we specify some boundary conditions.

The \term{discrete Dirichlet
Laplacian} on $\ddomain$ is the operator $\dDLapl$ acting
by the formula~\eqref{eq: combinatorial Laplacian} on functions
$f \colon \ddomain \setminus \bdry \ddomain \to \R$
defined on the interior primal vertices
, with the interpretation that the values
of the function are zero outside the interior (in particular on the
boundary~$\bdry \ddomain$). This operator is negative definite, and its inverse
is up to a sign
the \term{discrete Dirichlet Green's function} in~$\ddomain$, defined as
\begin{align}
\Green^\Dirichlet_{\ddomain} \, = \, & - (\dDLapl)^{-1} \colon
\ddomain \times \ddomain \to \R,
\end{align}
setting again the values to zero if either of the arguments is not an interior vertex.
Concretely, the discrete Dirichlet Green's function 
is determined by
\begin{align}\label{eq: defining properties of discrete Dirichlet Green function}
\begin{cases}
\, \gLapl \Green_{\ddomain}^{\Dirichlet}(\cdot, \zsbis) = - \delta_\zsbis(\cdot) & \text{ on } \ddomain \setminus \bdry \ddomain \\
\, \Green_{\ddomain}^{\Dirichlet}(\zs, \zsbis) = 0 & \text{ if either $\zs \in \bdry \ddomain$ or $\zsbis \in \bdry \ddomain$}.
\end{cases}
\end{align}
Put differently, the discrete Green's function is a
kernel for the Dirichlet problem for the discrete Laplacian: 
for any $\psi \colon \ddomain \setminus \bdry \ddomain \to \R$, the function
$f(\zs) = \sum_{\zsbis \in \ddomain \setminus \bdry \ddomain} \Green_{\ddomain}^{\Dirichlet}(\zs, \zsbis) \, \psi(\zsbis)$
is the unique solution to $\dDLapl f = - \psi$ with $f|_{\bdry \ddomain} \equiv 0$.

The \term{discrete Neumann Laplacian} is the operator $\dNLapl$ acting
on functions $f \colon \ddomain \to \C$ by the formula
\begin{align}\label{eq: discrete Neumann Laplacian}
(\dNLapl f)(\zs) :=
    \sum_{\substack{\zsbis \in \ddomain \\ |\zs-\zsbis| = \meshsz}} \big( f(\zsbis) - f(\zs) \big) .
\end{align}
The crucial difference to the Dirichlet case is that the coefficient of the
``diagonal'' term $f(\zs)$ on the right is
proportional to the number of neighbors of~$\zs$ in the discrete domain~$\ddomain$ for the
discrete Neumann Laplacian. The operator~$\dNLapl$ is negative semidefinite
but not invertible, it has a one-dimensional kernel consisting of constant functions
(the connectedness of the discrete domain is essential here).
The degeneracy can be remedied by restricting to the space
\begin{align}\label{eq: discr zero avg fun}
\ddistribMC{\ddomain} := \set{ f \colon \ddomain \to \R \; \Bigg| \; \sum_{\zs \in \ddomain} f(\zs) = 0 }
\end{align}
of \term{zero-average functions}: the Neumann Laplacian~$\dNLapl$ is injective
on~$\ddistribMC{\ddomain}$ and its range lies in~$\ddistribMC{\ddomain}$, so
\begin{align*}
\dNLapl \colon \ddistribMC{\ddomain} \to \ddistribMC{\ddomain} 
\end{align*}
is invertible. 
By fixing any probability mass function on the boundary, i.e.,
${b \colon \bdry \ddomain \to [0,1]}$ such that $\sum_{x \in \bdry \ddomain} b(x) = 1$,
we can still define a (choice of the) \term{discrete Neumann Green's function}
\begin{align}\label{eq: definition of discrete Dirichlet Green function}
{\Green_{\ddomain}^{\Neumann} \colon \ddomain \times \ddomain \to \R}
\qquad \text{ by } \qquad
{\Green_{\ddomain}^{\Neumann}(\cdot , \zsbis) := (\dNLapl)^{-1} (b - \delta_\zsbis)} \, .
\end{align}
We then have ${\gLapl \Green_{\ddomain}^{\Neumann}(\cdot , \zsbis) = - \delta_\zsbis(\cdot)}$
in $\ddomain \setminus \bdry \ddomain$, and for any $\psi \in \ddistribMC{\ddomain}$ the function
${f(\zs) := \sum_\zsbis \Green_{\ddomain}^{\Neumann}(\zs,\zsbis) \psi(\zsbis)}$ is the unique zero-average solution to
$\dNLapl f = - \psi$. 
Note also that all of the discrete derivatives of Section~\ref{subsec: difference operators} are
zero-average linear combinations of values, so the discrete derivatives
of $\Green_{\ddomain}^{\Neumann}$ with respect to its second argument
are well-defined (independent of the choice of~$b$).

\subsection{Scaling limits of discrete Green's functions}
\label{subsec: scaling limit Green functions}

An ingredient of the proof of our scaling limit result for correlations of local fields
is the convergence of discrete double derivatives of the discrete Green's functions to their
continuum counterparts. Below we introduce these continuum objects, and then state the
convergence in the form that it will be used.

Let $\domain \subset \bC$ be a nonempty open simply-connected proper subset of the complex plane.
Both the Dirichlet and Neumann Green's functions in such domains~$\domain$
can be defined making use of conformal invariance.
Choose a conformal map
\begin{align*}
\confmap \colon \domain \to \UD 
\end{align*}
to the unit disk
$\UD = \set{ z \in \bC \; \big| \; |z|<1} $.
The Dirichlet Green's function in the unit disk is
\begin{align*}
\Green_\UD^\Dirichlet(z,w) := \frac{-1}{2 \pi} \, \log \bigg| \frac{z-w}{1-z\overline{w}} \bigg|
\qquad \text{ for } z, w \in \UD, \; z \ne w ,
\end{align*}
and (a choice of) the Neumann Green's function in the unit disk is
\begin{align*}
\Green_\UD^\Neumann(z,w) := \frac{-1}{2 \pi} \, \log \Big| (z-w) (1-z \overline{w}) \Big|
\qquad \text{ for } z, w \in \UD, \; z \ne w .
\end{align*}
We may then define the \term{Dirichlet Green's function} in $\domain$ by
\begin{align*}
\Green_\domain^\Dirichlet(z,w) :=
\Green_\UD^\Dirichlet \big( \confmap(z) , \confmap(w) \big)
\qquad \text{ for } z, w \in \domain, \; z \ne w ,
\end{align*}
since it is easy to check that the result does not depend on the chosen~$\confmap$.
For the \term{Neumann Green's function}, we can use
\begin{align*}
\Green_\domain^\Neumann(z,w) := 
\Green_\UD^\Neumann \big( \confmap(z) , \confmap(w) \big)
\qquad \text{ for } z, w \in \domain, \; z \ne w \, ,
\end{align*}
which depends on the chosen conformal map~$\confmap$, but the difference of
any two choices is of the form $h_1(z) + h_2(w)$,
with $h_1,h_2 \colon \domain \to \R$ harmonic.

These Green's functions are real-analytic functions on the configuration
space
\begin{align*}
\ConfigSp{2}{\domain} = \set{(z,w) \in \domain \times \domain \, \big| \, z \ne w} \; . 
\end{align*}
They are Green's functions for the Laplacian $\Lapl = \pdder{x} + \pdder{y}$ (where we write
$z = x + \ii y$ with $x,y \in \R$) in the sense that
\begin{align}\nonumber
\, \Lapl \Green_\domain^{\DorN}(\, \cdot\, , w) = \; & 0 \; \quad \text{ on } \domain \setminus \set{w} \, \\
\label{eq: Green's function log singularity}
\text{and } \qquad 
\Green_\domain^{\DorN}(z, w) = \; & \frac{1}{2\pi} \, \log \frac{1}{|z-w|} + g_\domain^\DorN(z,w) \, .
\end{align}
where $g_\domain^\DorN \colon \domain \times \domain \to \bR$ is real-analytic, and harmonic
in both variables separately.
By direct inspection, we also observe symmetricity of the Green's functions
\begin{align*}
\Green_\domain^{\DorN}(z, w) = \Green_\domain^{\DorN}(w, z)
\qquad \text{ for } z, w \in \domain, \; z \ne w .
\end{align*}

We need to consider directional derivatives of discrete and
continuum Green's functions
with respect to both variables.
The (combinatorially normalized)
discrete directional derivative of a function $f \colon \ddomain \to \bC$
in a direction~$\mu \in \set{+1,\ii,-1,-\ii}$ is
\begin{align}\label{eq: discrete directional derivative}
\dnabla^{\mu} f (\zs) = \frac{f(\zs+\half\mu\meshsz) - f(\zs-\half\mu\meshsz)}{\mu} .
\end{align}
In the continuum, the directional derivative of a function~$f \colon \domain \to \bC$
in the direction of a complex number $\mu \in \bC$ of unit modulus~$|\mu|=1$, is
\begin{align}\label{eq: continuum directional derivative}
\nabla^\mu f (z) = \frac{\ud}{\ud t} f(z + \mu t) \Big|_{t=0} .
\end{align}
When differentiating functions of several variables, we indicate
the index or label of the variable acted on by an extra subscript in a hopefully
self-explanatory way.

The following convergence results of discrete Dirichlet Green's functions are well-known.

\begin{lemma}\label{lem: discrete Dirichlet Green function convergence}
When discrete domains $(\ddomain)_{\meshsz>0}$ form an approximation to~$\domain$ in the
Carath\'eo\-dory sense and $\zs^\meshsz,\zsbis^\meshsz\in\ddomain$ are the closest
points in the discrete domain~$\ddomain$ to given points ${z,w\in\domain}$,
the discrete Dirichlet Green's functions converge
\begin{align*}
\Green_{\ddomain}^\Dirichlet(\zs^\meshsz,\zsbis^\meshsz) = \Green_{\domain}^\Dirichlet(z,w) + \oo(1),
\end{align*}
and the error term~$\oo(1)$ is uniformly small for $(z,w)$ on compact subsets of~$\ConfigSp{2}{\domain}$.
\end{lemma}

\begin{coro}
When discrete domains $(\ddomain)_{\meshsz>0}$ form an approximation to~$\domain$ in the
Cara\-th\'eo\-dory sense and $\zs_\medial^\meshsz,\zsbis_\medial^\meshsz$ are
the closest edges of directions~$\mu,\nu \in \set{\pm 1 , \pm \ii}$
in the discrete domain~$\ddomain$ to given points~${z,w\in\domain}$,
the discrete double derivatives of the Dirichlet Green's functions converge
\begin{align}\label{eq: Dirichlet Greens double derivative scaling limit}
\meshsz^{-2} (\dnabla^{\mu})_\zs (\dnabla^{\nu})_\zsbis
	\Green_{\domain_\meshsz}^\Dirichlet
	\big( \zs^\meshsz_\medial , \zsbis^\meshsz_\medial \big)
	\, = & \ 
	\nabla^\mu_z \nabla^\nu_w
	\Green_{\domain}^\Dirichlet
	\big( z , w \big)
	+
	o(1)\,,\phantom{\Big\vert}
\end{align}
and the error term~$\oo(1)$ is uniformly small for $(z,w)$ on compact
subsets of~$\ConfigSp{2}{\domain}$.
\end{coro}

We also need a counterpart of the
result~\eqref{eq: Dirichlet Greens double derivative scaling limit}
for Neumann boundary conditions. This is less well-known, so we provide the proof below.
\begin{lemma}\label{lem: discrete Neumann Green function convergence}
When discrete domains $(\ddomain)_{\meshsz>0}$ form an approximation to~$\domain$ in the
Carath\'eo\-dory sense and $\zs_\medial^\meshsz,\zsbis_\medial^\meshsz$ are
the closest edges of directions~$\mu,\nu \in \set{\pm 1 , \pm \ii}$
in the discrete domain~$\ddomain$ to given points~${z,w\in\domain}$,
the discrete double derivatives of the Neumann Green's functions converge
\begin{align}\label{eq: Neumann Greens double derivative scaling limit}
\meshsz^{-2} (\dnabla^{\mu})_\zs (\dnabla^{\nu})_\zsbis
	\Green_{\domain_\meshsz}^\Neumann
	\big( \zs^\meshsz_\medial , \zsbis^\meshsz_\medial \big)
	\, = & \ 
	\nabla^\mu_z \nabla^\nu_w
	\Green_{\domain}^\Neumann
	\big( z , w \big)
	+
	o(1)\,,\phantom{\Big\vert}
\end{align}
and the error term~$\oo(1)$ is uniformly small for $(z,w)$ on compact
subsets of~$\ConfigSp{2}{\domain}$.
\end{lemma}
\begin{proof}
Fix $w \in \domain$ and a direction~$\nu \in \set{\pm 1 , \pm \ii}$, and let
$\zsbis_\medial^\meshsz$ be the edge of direction~$\nu$ in~$\ddomain$ that is nearest to~$w$.
Consider the discrete Dirichlet Green's function~$\Green_{\ddomain^*}^\Dirichlet$
on the dual graph~$\ddomain^* \subset \meshsz \ZDual$, and particularly its discrete derivative
in the second variable at the medial vertex~$\zsbis_\medial^\meshsz$ in the direction~$-\ii \nu$,
\begin{align*}
F(\zs^\meshsz_*) :=
    (\dnabla^{-\ii \nu})_\zsbis \Green_{\ddomain^*}^\Dirichlet (\zs^\meshsz_*, \zsbis_\medial^\meshsz) .
\end{align*}
By construction, the function $F \colon \ddomain^* \to \R$ is discrete harmonic except at
$\zs_*^\meshsz = \zsbis_\medial^\meshsz \pm \ii \frac{\meshsz}{2} \nu$, where its discrete Laplacian is
\begin{align*}
\dDLaplDual F \big( \zsbis_\medial^\meshsz \pm \ii \frac{\meshsz}{2} \nu \big) = \pm 1 .
\end{align*}
In particular, we can define a function~$H$ (a discrete harmonic conjugate of~$F$), by
requiring
\begin{align*}
\dnabla^{\mu} H (\zs_\medial^\meshsz) = \dnabla^{-\ii \mu} F (\zs_\medial^\meshsz)
\end{align*}
for all edges $\zs_\medial^\meshsz \ne \zsbis_\medial^\meshsz$;
the single-valuedness of~$H$ around~$\zsbis_\medial^\meshsz$ relies on the
opposite values
of the Laplacian of~$F$ at the two adjacent dual
vertices $\zsbis_\medial^\meshsz \pm \ii \frac{\meshsz}{2}\nu$.
From the construction, routine combinatorial considerations yield
\begin{align*}
\dNLapl H ( \zs^\meshsz ) = 0 \quad
    \text{ for } \zs^\meshsz \ne \zsbis_\medial^\meshsz \pm \frac{\meshsz}{2} \nu 
\qquad \text{ and } \qquad 
\dNLapl H \big( \zsbis_\medial^\meshsz \pm \frac{\meshsz}{2} \nu \big) = \pm 1 ,
\end{align*}
where the Dirichlet boundary conditions for~$F$ are used for the Neumann
Laplacian harmonicity of~$H$ on the boundary. As a consequence,
fixing the additive constant in the harmonic conjugate so that~$H$ it is
zero-average for definiteness, we find
\begin{align*}
H (\zs^\meshsz) =
  - (\dnabla^{\nu})_\zsbis \Green_{\ddomain}^\Neumann (\zs^\meshsz, \zsbis_\medial^\meshsz) .
\end{align*}
We can therefore write the double derivative of the discrete Neumann Green's
function in the following form
\begin{align*}
\meshsz^{-2} \, (\dnabla^{\mu})_\zs (\dnabla^{\nu})_\zsbis
    \Green_{\ddomain}^\Neumann (\zs_\medial^\meshsz, \zsbis_\medial^\meshsz)
= \; & - \meshsz^{-2} \, \dnabla^{\mu} H (\zs_\medial^\meshsz) \\
= \; & - \meshsz^{-2} \, \dnabla^{-\ii \mu} F (\zs_\medial^\meshsz) \\
= \; & - \meshsz^{-2} \, (\dnabla^{-\ii \mu})_\zs (\dnabla^{-\ii \nu})_\zsbis
    \Green_{\ddomain}^\Dirichlet (\zs_\medial^\meshsz, \zsbis_\medial^\meshsz) \\
= \; & - \nabla^{-\ii \mu}_z \nabla^{-\ii \nu}_w \Green_{\domain}^\Dirichlet \big( z , w \big) + o(1) ,
\end{align*}
where the last step uses~\eqref{eq: Dirichlet Greens double derivative scaling limit}.
This shows that the double derivative has a scaling limit, given in terms of
different ($90^\circ$ rotated) directional derivatives of the Dirichlet Green's
function.
It remains to observe the relationship
\begin{align*}
- \nabla^{-\ii \mu}_z \nabla^{-\ii \nu}_w \Green_{\domain}^\Dirichlet \big( z , w \big)
= \nabla^{\mu}_z \nabla^{\nu}_w \Green_{\domain}^\Neumann \big( z , w \big) .
\end{align*}
of directional double derivatives of continuum Green's functions, which can
for example be directly verified from the defining
expressions of~$\Green_{\domain}^\DorN$.
\end{proof}

\subsection{Discrete contour integration and discrete monomial functions}
\label{subsec: discrete integration}\label{subsec: discrete monomials}
Two final notions of discrete complex analysis are employed in constructing
the representations of Heisenberg and Virasoro algebras on the space of local fields
of the discrete model as in~\cite{HKV}: discrete analogues of
contour integration and of Laurent monomials.
These notions are only used as such on the square grid of unit-mesh~$\meshsz=1$.
As a notational distinction, we typically denote vertices of the unit-mesh square grid
and the associated diamond and medial lattices by
$\zu, \zubis \in \SqLat$,
$\zu_\diamond, \zubis_\diamond \in \SqLatDiamond$, and
$\zu_\medial, \zubis_\medial \in \SqLatMedial$.
For the appropriate notion of discrete contour integration, we need yet one more
lattice,
\begin{align*}
\SqLatCorner := \; & \set{ r  + \ii s \; \Big| \; r,s \in \half\bZ + \frac{1}{4}}
    & \text{\term{(corner lattice)}}
\end{align*}
whose vertices are referred to as corners.
A \term{corner path} is a finite sequence $\gamma=(c_0,\ldots,c_\ell)$ of consecutively nearest corners, i.e.,
$c_j\in\ZCorner$ for $0 \le j \le \ell$ such that ${|c_j-c_{j-1}|=\frac{1}{2}}$ for $0 < j \le \ell$.
Then, given two functions $f \colon \SqLatDiamond \to \C$ and $g \colon \SqLatMedial \to \C$ we define
\begin{align}\label{eq: discrete integration}
\dcint{\gamma} f(\zu_\diamond)g(\zu_\medial)\dd \zu :=
	\sum_{j=1}^\ell \, (c_j - c_{j-1})
        f(\zu_j^\diamond) g(\zu_j^\medial)\,, &\ \ \ \ \textnormal{and}
	\\ \nonumber
\dcint{\gamma} f(\zu_\diamond)g(\zu_\medial)\dd{\cconj{\zu}} :=
	\sum_{j=1}^\ell \, \cconj{(c_j - c_{j-1})}
        f(\zu_j^\diamond)g(\zu_j^\medial)\,, &
\end{align}
where $\zu_j^\diamond \in \SqLatDiamond$ and $\zu_j^\medial \in \SqLatMedial$ are the unique diamond vertex
and medial vertex that have both $c_j$ and $c_{j-1}$ among their nearest corner vertices~---
see Figure~\ref{fig: discrete contour}.
Note, furthermore, that bilinear discrete integration may be defined by the same formulas also when one
of the functions $f, g$ is complex-valued and the other one takes values in a complex vector space.

\begin{figure}[h!]
\centering
\begin{overpic}[scale=0.776, tics=10]{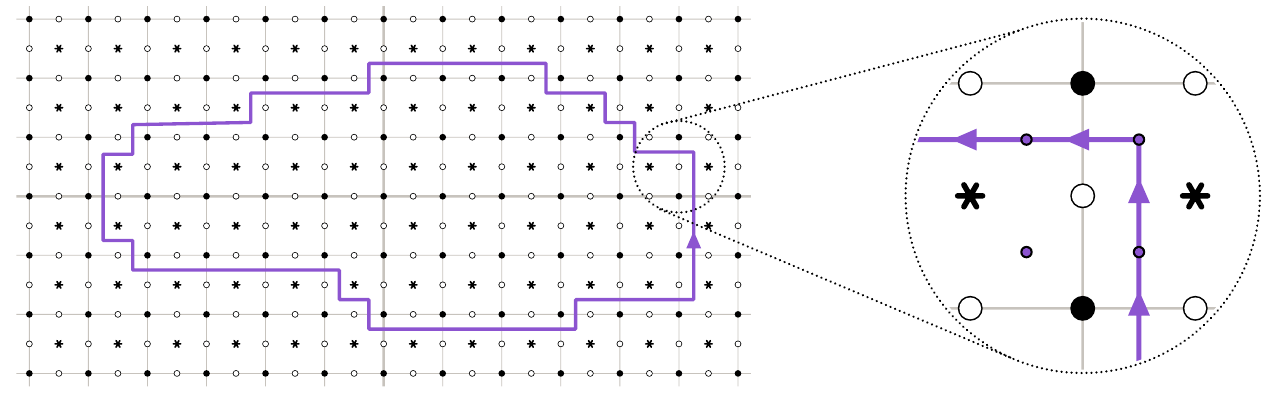}
	\put(95.5,13){$\zu^{\diamond}_j$}
	\put(85.5,13){$\zu^{\medial}_j$}
	\put(90.5,18.5){$c_j$}
	\put(90.5,9.3){$c_{j-1}$}
	\put(60,29){$\Z^2$}
\end{overpic}
\centering
\caption{A discrete contour on the corner lattice. Each step $(c_{j-1},c_j)$ of it
separates a medial vertex $\zu_j^\medial \in \SqLatMedial$ from a diamond vertex
$\zu_j^\diamond \in \SqLatDiamond$.}
\label{fig: discrete contour}
\end{figure}

A corner path $\gamma=(c_0,\ldots,c_\ell)$ is said to be closed if~$c_\ell = c_0$. If, moreover, the corners
$c_1, \ldots, c_\ell$ are distinct, we call $\gamma$ a \term{corner contour}. A corner contour~$\gamma$
is essentially a Jordan curve on the lattice; it surrounds a bounded region.
The set ${\interior{\gamma} \subset \SqLatDiamond \cup \SqLatMedial}$ of surrounded diamond and
medial vertices is called the \term{interior} of~$\gamma$, and we also denote
${\interior_\diamond \gamma := \interior \gamma \cap \SqLatDiamond}$ and
${\interior_\medial \gamma := \interior \gamma \cap \SqLatMedial}$
for the two types of surrounded vertices separately~--- see Figure~\ref{fig: discrete contour}.
A corner contour $\gamma=(c_0,\ldots,c_\ell)$ is said to be \term{positively oriented}, if the
corners $c_1, \ldots, c_\ell$ appear counterclockwise along~$\gamma$ seen as a Jordan curve.
The symbol~$\dcoint{\gamma}$ is used
for discrete integrations~\eqref{eq: discrete integration} along positively oriented
corner contours~$\gamma$.

For calculations with discrete integration, the most important properties are
the following analogues of Stokes' formula, contour deformation, and integration
by parts.

Let $\gamma$ be a positively-oriented corner contour.
For any two functions $f \colon \SqLatDiamond \to \C$ and $g \colon \SqLatMedial \to \C$ we have
the \term{discrete Stokes' formulas}
\begin{align}\label{eq: discrete Stokes}
\dcoint{\gamma} f(\zu_\diamond)g(\zu_\medial)\dd{\zu}
= \; & \phantom{-} \ii \, \sum_{\zubis_\diamond\in\interior_\diamond\gamma} f(\zubis_\diamond)\gdeebar g(\zubis_\diamond)
	  + \ii \, \sum_{\zubis_\medial\in\interior_\medial\gamma} \gdeebar f(\zubis_\medial) g(\zubis_\medial) \, , 
\\
\dcoint{\gamma} f(\zu_\diamond)g(\zu_\medial)\dd{\cconj \zu}
= \; & - \ii \, \sum_{\zubis_\diamond\in\interior_\diamond\gamma} f(\zubis_\diamond) \gdee g(\zubis_\diamond)
	  - \ii \, \sum_{\zubis_\medial\in\interior_\medial\gamma} \gdee f(\zubis_\medial) g(\zubis_\medial) \, . &
\end{align}
In particular
if the two functions $f$ and $g$ are discrete holomorphic on the symmetric differences
of the appropriate subsets of the interior of two positively-oriented corner contours $\gamma_1$ and $\gamma_2$
(namely, $\gdeebar f \equiv 0$ on ${(\interior_\medial \gamma_1 \setminus \interior_\medial \gamma_2)
\cup (\interior_\medial \gamma_2 \setminus \interior_\medial \gamma_1)}$
and $\gdeebar g \equiv 0$ on  ${(\interior_\diamond \gamma_1 \setminus \interior_\diamond \gamma_2)
\cup (\interior_\diamond \gamma_2 \setminus \interior_\diamond \gamma_1)}$),
then one has the \term{contour deformation} equality
\begin{align*}
\dcoint{\gamma_1} f(\zu_\diamond)g(\zu_\medial)\dd{\zu} = \dcoint{\gamma_2} f(\zu_\diamond)g(\zu_\medial)\dd{\zu} .
\end{align*}
A similar contour deformation equality holds for discrete antiholomorphic functions and the integrals
$\dcoint{\gamma_i} \cdots \dd{\cconj{\zu}}$.

The \term{discrete integration by parts} equalities
\begin{align}\nonumber
\dcoint{\gamma} f(\zu_\diamond) \, \gdee h(\zu_\medial) \, \dd{\zu}
= \; & - \dcoint{\gamma} h(\zu_\diamond) \, \gdee f(\zu_\medial) \, \dd{\zu} \,
\mspace{30mu} \text{ and } \\ \label{eq: integration by parts}
\dcoint{\gamma} f(\zu_\diamond) \, \gdeebar h(\zu_\medial) \, \dd{\cconj{\zu}}
= \; & - \dcoint{\gamma} h(\zu_\diamond) \, \gdeebar f(\zu_\medial) \, \dd{\cconj{\zu}}
\end{align}
hold whenever ${f, h \colon \SqLatDiamond \to \C}$ are two discrete holomorphic functions
on a discrete neighborhood of a corner contour~$\gamma$.

\begin{figure}
	\centering
	\subfigure[The values of $\zu^{[1]}$ are simply the complex points where the diamond and
	medial vertices~$\zu$ are embedded.]{
		\includegraphics[scale=2]{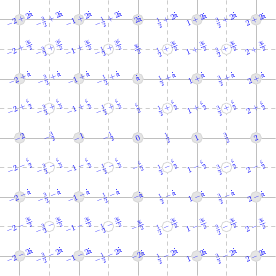}} \\
	\vspace{.1cm}
	\subfigure[The values of $\zu^{[3]}$ vanish in a finite neighborhood of the origin.
	By property~3 of Proposition~\ref{prop: monomials},
	such vanishing neighborhoods grow with the order of the monomial.]{
		\includegraphics[scale=2]{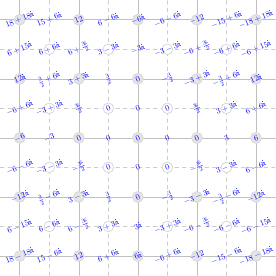}}
	\caption{The values of the positive Laurent monomials $\zu^{[1]}$ and $\zu^{[3]}$.}
	\label{fig: positive monomial values}
\end{figure}

The final necessary ingredient is
discrete analogues of the Laurent monomials
${z \mapsto z^n}$ for~$n \in \bZ$.
Such discrete Laurent monomials were constructed in~\cite{HKV}, but
we crucially need a minor modification here~--- without this,
the exact correspondence between the CFT local fields
and the lattice model local fields simply does not work.
We state the result of the modified construction here. The modification is in
the exact coefficients used in property~5 below.
We still refer to the original proof, as it remains
in all essential ways similar.

\begin{figure}
	\centering
\subfigure[The values of $\zu^{[-1]}$.]{
	\includegraphics[scale=2]{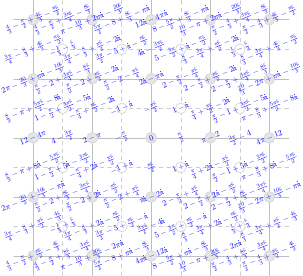}} \\
\vspace{.1cm}
\subfigure[The values of $\zu^{[-3]}$.]{
	\centering
	\includegraphics[scale=2]{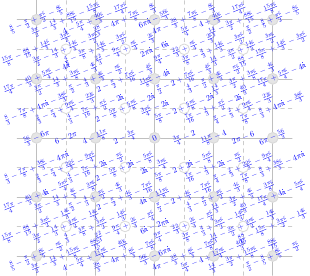}}
	\caption{The values of the negative Laurent monomials $\zu^{[-1]}$ and $\zu^{[-3]}$.}
	\label{fig: negative monomial values}
\end{figure}

\begin{prop}[{\cite[Proposition~2.1]{HKV}}]\label{prop: monomials}
There exists a unique family of $\C$-valued functions $\{\zu\mapsto \zu^{[n]}\}_{n\in\Z}$ on
$\ZDiamond\cup\ZMedial$ that satisfies the following properties:
\begin{enumerate}
\item For all $n\in\Z$, the function $\zu \mapsto \zu^{[n]}$ has the same square-grid symmetries
as the Laurent monomial $z \mapsto z^n$,
i.e., $(\ii \zu)^{[n]}=\ii^n \zu^{[n]}$ and $\cconj{\zu}^{[n]} = \cconj{\zu^{[n]}}$ for all
$\zu \in \ZDiamond\cup\ZMedial$.
\item For all $\zu\in\ZDiamond\cup\ZMedial$, $\zu^{[0]}=1$ and, for all $n\in\Z$, $\gdee \zu^{[n]} = n\,\zu^{[n-1]}$.
\item For each $\zu\in\ZDiamond\cup\ZMedial$, there exists an $N\in\N$ such that $\zu^{[n]}=0$ for all $n\geq N$.
\item For $n<0$, we have $\zu^{[n]} \to 0$ as $\vert \zu\vert\rightarrow \infty$.
\item The first negative-power monomial has the following explicit
failure of discrete holomorphicity near the origin
\begin{align*}
    \frac{1}{2\pi}\gdeebar \zu^{[-1]}
    =
    \frac{1}{2}\delta_{\zu,0}
    +
    \frac{1}{4}\sum_{\vert \zubis\vert=\frac{1}{2}}\delta_{\zu,\zubis}
    +
    \frac{1}{8}\sum_{\zubis=\frac{\pm1\pm\ii}{2}}\delta_{\zu,\zubis}\,.
\end{align*}
\item \label{property: poles}
For any $n\geq 0$ and all $\zu\in\ZDiamond\cup\ZMedial$ we have $\gdeebar \zu^{[n]}=0$. 
For any $n<0$, we have $\gdeebar \zu^{[n]}=0$ except at finitely many
points\footnote{This finite set of failure of discrete holomorphicity
is made more explicit in the considerations below the proposition.}
~$\zu \in \ZDiamond\cup\ZMedial$.
\item For any $n,m\in\Z$, we have the discrete residue formula
\begin{align*}
    \frac{1}{2\pi\ii}
    \dcoint{\gamma} \zu^{[n]}_\diamond \zu^{[m]}_\medial \, \dd \zu 
    \; = \; \delta_{n+m+1}\,,
\end{align*}
for any large enough positively-oriented corner contour $\gamma$ that encircles the origin.
\item For any $n \in \bZ$, as~${|\zu| \to \infty}$, the discrete monomial has the asymptotics
\begin{align}\label{eq: discrete monomial asymptotics}
\zu^{[n]} = \zu^n + \oo(|\zu|^n) \, .
\end{align}
\end{enumerate}
\end{prop}

Example values of some positive monomials are illustrated in Figure~\ref{fig: positive monomial values}.
These values remain identical to~\cite{HKV}. Our modification only affects the values of negative monomials,
such as those illustrated in Figure~\ref{fig: negative monomial values}.

According to property~6 above,
for each $k>0$, there is a finite set
\begin{align}\label{eq: monomial pole support set}
\Poles{-k} := \set{ \zu \in \ZDiamond \cup \ZMedial \; \Big| \; \gdeebar \zu^{[k]} \ne 0 }
\end{align}
where the discrete holomorphicity of $\zu \mapsto \zu^{[-k]}$ fails.
Let us also denote $\Poles{-k}_\medial\coloneqq\Poles{-k}\cap\ZMedial$.
The values of~$\gdeebar \zu^{[-3]}$ in Figure~\ref{fig: monomial poles}(b) display
the shape of~$\Poles{-3}$.

In our analysis, we will need the exact growth of~$\Poles{-k}_\medial$ with~$k$.
For that purpose, define discrete balls in $\ZDiamond\cup\ZMedial$ as
\begin{align}\label{eq: lattice ball}
\Ball(r) := \set{ \zu \in \ZDiamond \cup \ZMedial \; \Big| \; \| \zu \|_1\leq r } \, ,
\end{align}
where $\|\zu\|_1 := \re(\zu)+\im(\zu)$,
and denote by $\Ball_\medial(r) := \Ball(r) \cap \ZMedial$ the set of its medial vertices.
We record the following facts for later reference.

\begin{figure}
	\centering
\subfigure[The values of $\gdeebar \zu^{[-1]}$ are prescribed in property~5 of
Proposition~\ref{prop: monomials}, and they are the single crucial modification to the monomial definition
that we made compared to~\cite{HKV}.]{
	\centering
	\includegraphics[scale=2]{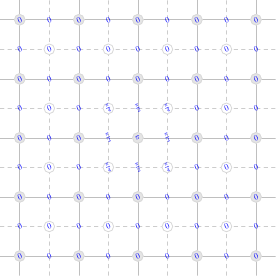}} \\
\vspace{.1cm}
\subfigure[The values of $\gdeebar \zu^{[-3]}$ illustrate the shape
of the finite set $\Poles{-3} \subset \ZDiamond \cup \ZMedial$ of failure of
discrete holomorphicity of $\zu^{[-3]}$.]{
	\centering
	\includegraphics[scale=2]{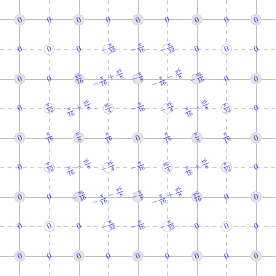}}
	\caption{The values of $\gdeebar \zu^{[-1]}$ and $\gdeebar \zu^{[-3]}$.}
	\label{fig: monomial poles}
\end{figure}

\begin{rmk}\label{rmk: pole extension}
From property~5 of Proposition \ref{prop: monomials} it follows
that ${\Poles{-1}\subset \Ball(1)}$ and that ${\Poles{-1}_\medial\subset\Ball_\medial(\frac{1}{2})}$;
see also Figure~\ref{fig: monomial poles}(a).
For $k>0$, we have
\begin{align}\label{eq: }
\Poles{-k} \subset \Ball\bigg(\frac{k+1}{2}\bigg)
    \mspace{30mu}
    \text{ and }
    \mspace{30mu}
\Poles{-k}_\medial \subset
    \Ball_\medial\bigg(\left\lfloor\frac{k}{2}\right\rfloor+\frac{1}{2}\bigg) ,
\end{align}
by virtue of the formula $\zu^{[-k]}= \frac{(-1)^{k-1}}{(k-1)!} \, \gdee^{k-1} \zu^{[-1]} $
from property 2 and the form~\eqref{eq: discrete Wirtinger derivatives} of the finite difference
operator~$\gdee$.
\hfill$\diamond$
\end{rmk}

\section{Discrete Gaussian Free Field and its local fields}
	\label{sec:currents}
	In this section we define the discrete Gaussian free field (DGFF), its local fields, and
the current modes that provide the representation of two commuting copies of the Heisenberg
algebra on the space of these local fields. This essentially amounts to recalling results
from~\cite{HKV}~--- but two differences are worth pointing out.
First of all, it is necessary to use the slightly modified discrete Laurent monomials of
Section~\ref{subsec: discrete monomials}, to get the commutation of the two chiralities.\footnote{
In \cite{HKV} only one chiral Heisenberg (and Virasoro) algebra action was written down.
It is straightforward to define another chiral copy analogously.
But the two chiral copies would not commute if one were to use the discrete Laurent monomials exactly
as defined originally.}
Second, we allow for both Dirichlet and Neumann boundary conditions for the DGFF,
and we correspondingly focus on local fields built from the discrete gradient of the DGFF,
as this is exactly what corresponds to the (full) Fock space of local fields of the free boson CFT
(see Section~\ref{sec:scal_lim}).

Our probabilistic models, the DGFF with Dirichlet and Neumann boundary conditions, are
defined in Section~\ref{sub: DGFF def}.
In Section~\ref{sub: currents} we introduce
the discrete holomophic and antiholomorphic currents of the DGFF,
which are the quantities that enable the subsequent 
discrete complex analysis approach.
Section~\ref{sub: field polynomials} details what we mean by using abstract
polynomials to specify fields, and how these become
concrete random variables in the model in any fixed domain and either Dirichlet
or Neumann boundary conditions.
The final notion of local fields is derived from such field polynomials
in Section~\ref{sub: corr equiv}, by
forming equivalence classes of field polynomials which have indistinguishable
correlations at macroscopic distances.
The construction of the fundamental algebraic structure of local fields
is given in Section~\ref{sub: current modes}: there are two commuting
Heisenberg algebra actions on the space of local fields given by
Laurent-monomial weighted discrete contour integrals of the discrete currents
of the DGFF.
Finally, the two commuting Virasoro actions on the space of local fields of the
DGFF are obtained by Sugawara constructions.
Section~\ref{sub: Sugawara} recalls how the Sugawara construction
applies in this discrete setup, and draws attention in particular to the grading
of local fields by their scaling dimensions,
i.e., eigenvalues of~$\dvirL{0} + \dvirBarL{0}$.

\subsection{DGFF with Dirichlet and Neumann boundary conditions}\label{sub: DGFF def}
Let $\ddomain$ be a discrete domain on the $\meshsz$-mesh square
grid as in Section~\ref{subsec: discrete domains}.

The discrete Gaussian free field with Dirichlet boundary conditions on~$\ddomain$
could be defined as the centered Gaussian
process~$(\dgffptwise(\zs))_{\zs \in \ddomain}$ indexed by the vertices $\zs$ of the discrete domain,
whose covariance is
\begin{align*}
\EX^\Dirichlet_{\ddomain} \big[ \dgffptwise(\zs) \, \dgffptwise(\zsbis) \big]
    = 4 \pi \, \Green^\Dirichlet_{\ddomain}(\zs,\zsbis) 
\end{align*}
where $\Green^\Dirichlet_{\ddomain} \colon \ddomain \times \ddomain \to [0,\infty)$
is the discrete Dirichlet Green's function~\eqref{eq: defining properties of discrete Dirichlet Green function}.
Equivalently, the probability density of the vector~$(\dgffptwise(\zs))_{\zs \in \ddomain \setminus \bdry \ddomain}$ is
\begin{align*}
\frac{1}{\mathcal{Z}_{\ddomain}} \;
    \exp \bigg( -\frac{1}{8\pi} \sum_{\substack{\set{\zs,\zsbis} \subset \ddomain \\ |\zs-\zsbis| = \meshsz}}
        \big( \dgffptwise(\zs) - \dgffptwise(\zsbis) \big)^2 \bigg) \; \prod_{\zs \in \ddomain \setminus \bdry \ddomain} \ud \dgffptwise(\zs) ,
\end{align*}
where $\mathcal{Z}_{\ddomain}$ is a normalization constant and
$\dgffptwise(\zs)$ is interpreted as $0$ if $\zs \in \bdry \ddomain$.
The quadratic form in the exponential above is a constant multiple of the discrete Dirichlet energy of~$\dgffptwise$.

Imagining the DGFF concretely as a pointwise defined random field is convenient,
but for the purposes of this article we make two adjustments.
First, for a closer parallel with the continuum Gaussian free field, it is more appropriate to consider the
``mollified'' values $\dgff{f} := \meshsz^2 \sum_{\zs \in \ddomain} \dgffptwise(\zs) f(\zs)$,
and view the DGFF as a process indexed by ``test functions''~$f$.
Furthermore, the principal object for us is the gradient of the field,
so we define the
\term{discrete Gaussian free field with Dirichlet boundary conditions} on~$\ddomain$
(Dirichlet DGFF for short) to be the centered Gaussian process
\begin{align*}
(\dgff{f})_{f \in \ddistribMC{\ddomain}} ,
\end{align*}
indexed by zero-average test functions~$f \in \ddistribMC{\ddomain}$
as defined in~\eqref{eq: discr zero avg fun}, whose covariance is given by
\begin{align*}
\EX^\Dirichlet_{\ddomain} \big[ \dgff{f_1} \dgff{f_2} \big]
    = 4 \pi \meshsz^4 \, \sum_{\zs_1, \zs_2 \in \ddomain} f_1(\zs_1) \, \Green^\Dirichlet_{\ddomain}(\zs_1,\zs_2) \, f_2(\zs_2) .
\end{align*}
Discrete derivatives are recovered by suitable zero-average mollifiers, so they
are well-defined random variables which we for convenience still denote by
\begin{align*}
\pdee \dgffptwise (\zs) = \meshsz^{-2} \, \dgff{- \pdee \delta_\zs}, \quad
\pdeebar \dgffptwise (\zs) = \meshsz^{-2} \, \dgff{- \pdeebar \delta_\zs}, \quad
\gLapl \dgffptwise (\zs) = \meshsz^{-2} \, \dgff{\gLapl \delta_\zs}, \quad \text{etc.}
\end{align*}
In principle, the pointwise values $(\dgffptwise(\zs))_{\zs \in \ddomain}$
and the test-function indexed field \linebreak[4]
${(\dgff{f})_{f \in \ddistribMC{\ddomain}}}$
contain exactly the same information; $(\dgff{f})_{f \in \ddistribMC{\ddomain}}$
straightforwardly determines the discrete derivatives, and one can
``integrate the discrete derivatives'' to recover values of~$\dgffptwise$ anywhere in~$\ddomain$
starting from the boundary~$\bdry \ddomain$ where the values are zero.
However, it is more appropriate to think of $\dgff{\cdot}$ as a description of the
gradient of the free field only,
since we will only ever probe the free field $(\dgff{f})_{f \in \ddistribMC{\ddomain}}$ with
finitely many finitely supported test functions~$f$
(see also Section~\ref{subsec: GFF def} for the analogue in the continuum).

For Neumann boundary conditions, there is an inherent ambiguity about an additive constant in
the free field, and in this case only the gradient of the field is meaningful.
We define the
\term{discrete Gaussian free field with Neumann boundary conditions} on~$\ddomain$
(Neumann DGFF for short) to be the centered Gaussian process
\begin{align*}
(\dgff{f})_{f \in \ddistribMC{\ddomain}} ,
\end{align*}
indexed by zero-average test functions~$f \in \ddistribMC{\ddomain}$, whose covariance is given by
\begin{align*}
\EX^\Neumann_{\ddomain} \big[ \dgff{f_1} \dgff{f_2} \big]
= - 4 \pi \meshsz^2 \sum_{\zs \in \ddomain} f_1(\zs) \; \big( (\dNLapl)^{-1} f_2 \big) (\zs) .
\end{align*}
Note that if $\Green^\Neumann_{\ddomain}$ is any choice of a discrete Neumann Green's function, this
covariance may be alternatively written in a form where $\Green^\Neumann_{\ddomain}$ is an approximate integral kernel,
\begin{align*}
\EX^\Neumann_{\ddomain} \big[ \dgff{f_1} \dgff{f_2} \big]
    = 4 \pi \meshsz^4 \, \sum_{\zs_1, \zs_2 \in \ddomain} f_1(\zs_1) \, \Green^\Neumann_{\ddomain}(\zs_1,\zs_2) \, f_2(\zs_2) .
\end{align*}
Although the pointwise values of the Neumann DGFF are not defined,
the discrete derivatives are recovered by zero-average mollifiers, and
by a mild abuse of notation we
still denote the corresponding random variables by
$\pdee \dgffptwise (\zs)$, $\pdeebar \dgffptwise (\zs)$, $\gLapl \dgffptwise (\zs)$ as above.

Wick's formula applies to any centered Gaussians,
in particular to both the Dirichlet and Neumann DGFF. Concretely,
for any $f_1, \ldots, f_n \in \ddistribMC{\ddomain}$, \term{Wick's formula} gives
\begin{align}
    \label{eq: Wick formula for DGFF}
\EX_{\ddomain}^{\DorN} \Big[ \prod_{i=1}^{n} \dgff{f_i} \Big]
= \; & \sum_{P \in \Pair(n)} \prod_{\set{i,j} \in P} 
    \EX_{\ddomain}^{\DorN} \big[ \dgff{f_i}, \dgff{f_j} \big] \\
\nonumber
= \; & \meshsz^{2n} \sum_{\zs_1, \ldots, \zs_n \in \ddomain}
    f_1(\zs_1) \cdots f_n(\zs_n)
    \bigg( \sum_{P \in \Pair(n)} \prod_{\set{i,j} \in P} 4 \pi \, \Green^\DorN_{\ddomain}(\zs_i,\zs_j) \bigg) ,
\end{align}
where the sums are over the set~$\Pair(n)$ of \term{pairings}~$P$ of the index
set~$\set{1,2,\ldots,n}$, i.e., partitions of the index set into subsets of size two each
(for $n$ odd, there are no such pairings and the empty sum is zero).
It is convenient to keep in mind the more concise formal version
\begin{align*}
\text{``}\, \EX_{\ddomain}^{\DorN} \big[ \dgffptwise (\zs_1) \cdots \dgffptwise (\zs_n) \big]
= \; & \sum_{P \in \Pair(n)} \prod_{\set{i,j} \in P} 4 \pi \, \Green^\DorN_{\ddomain}(\zs_i,\zs_j) \,\text{''}
\end{align*}
of~\eqref{eq: Wick formula for DGFF}, which is meaningful for the pointwise defined
Dirichlet DGFF and can also be used for calculations with the Neumann DGFF if one ensures
that only zero-average linear combinations are considered.

\begin{rmk}\label{rmk: lap 0 corr}
Let $\ddomain$ be a discrete domain and let $\zs \in \ddomain \setminus \bdry \ddomain$ be an interior point.
Recall that the discrete Laplacian of the DGFF at~$\zs$ is the random variable
$\gLapl \dgffptwise (\zs) = \meshsz^{-2} \, \dgff{\gLapl \delta_\zs}$.
For any $f \in \ddistribMC{\ddomain}$ the covariance of $\gLapl \dgffptwise (\zs)$ and
$\dgff{f}$ simplifies to the following frequently useful formula
\begin{align}\label{eq: 2pt function of the Laplacianf of DGFF}
\phantom{\Big\vert}\EX_{\ddomain}^{\DorN} \big[ \gLapl \dgffptwise(\zs) \; \dgff{f} \big]
    = \; & - 4\pi \meshsz^2 \, f(\zs) \\ \nonumber
\quad \text{ or formally } \phantom{\Big\vert} 
\quad \text{``}\,
\EX_{\ddomain}^{\DorN} \big[ \gLapl \dgffptwise(\zs) \, \dgffptwise(\zsbis) \big]
    = \; &- 4\pi \, \delta_{\zs,\zsbis} \,\text{''} .
\end{align}
Taking furthermore into account Wick's formula~\eqref{eq: Wick formula for DGFF},
a particularly simple
consequence of~\eqref{eq: 2pt function of the Laplacianf of DGFF}
is that the correlation functions of~$\gLapl \dgffptwise(\zs)$ of the form
$\EX_{\ddomain}^{\DorN} \big[ \gLapl \dgffptwise(\zs) \, \prod_{i=1}^{n} \dgff{f_i} \big]$
vanish when none of the supports of $f_1, \ldots, f_n$ contain~$\zs$.
\hfill $\diamond$
\end{rmk}

\subsection{Discrete holomorphic and antiholomorphic currents}\label{sub: currents}
Using the primal lattice discrete holomorphic and antiholomorphic derivatives
$\pdee$ and $\pdeebar$ defined in~\eqref{eq: primal lattice dee}
and~\eqref{eq: primal lattice deebar},
we define the \term{discrete holomorphic current} of the DGFF (either Dirichlet
or Neumann) at an edge midpoint~$\zs$ of~$\ddomain$ as
\begin{align}\label{eq: discrete holomorphic current}
\HolCurr(\zs) \, := \, \ii \, \pdee \dgffptwise(\zs)
  \, = \, \frac{\ii \meshsz}{2} \sum_{\substack{\zsbis \in \ddomain \\ |\zsbis-\zs| = \meshsz/2}} \frac{\dgffptwise(\zsbis)}{\zsbis-\zs} \, ,
\end{align}
where, as usual, the zero-average linear combination of DGFF values on the right should properly
be interpreted as~$\dgff{\cdots}$ with the suitable zero-average function inserted.
Similarly, we define the \term{discrete antiholomorphic current} at~$\zs$ as
\begin{align}\label{eq: discrete antiholomorphic current}
\AntiHolCurr(\zs) \, := \, -\ii \, \pdeebar \dgffptwise(\zs)
  \, = \, \frac{-\ii \meshsz}{2} \sum_{\substack{\zsbis \in \ddomain \\ |\zsbis-\zs| = \meshsz/2}} \frac{\dgffptwise(\zsbis)}{\cconj{\zsbis}-\cconj{\zs}} \, .
\end{align}

The following discrete holomorphicity/antiholomorphicity of correlations of the currents
is a consequence of the factorizations~\eqref{eq: factorizations of Laplacian with primal graph derivatives}
of the Laplacian and the vanishing of correlations of the discrete Laplacian of the free field
(see Remark~\ref{rmk: lap 0 corr}). 
\begin{lemma}[{\cite[Lemma~3.5]{HKV}\footnote{In the article \cite{HKV}
there is a factor two mistake in this statement and the proof has an error which is simply
fixed by correctly using the factorization~\eqref{eq: factorizations of Laplacian with primal graph derivatives}.}}]
\label{lem: discrete holomorphicity properties of currents}
Let $\ddomain \subset \SqLatMesh$ be a discrete domain and let
\linebreak[4]
${f \in \ddistribMC{\ddomain}}$ be a zero-average function.
Then the functions
$\zs \mapsto \EX_{\ddomain}^{\DorN} \big[ \, \HolCurr(\zs) \, \dgff{f} \big]$ and
$\zs \mapsto \EX_{\ddomain}^{\DorN} \big[ \, \AntiHolCurr(\zs) \, \dgff{f} \big]$ are, respectively,
discrete holomorphic and discrete antiholomorphic on the set
\begin{align*}
\big(\ddomain \cap \meshsz\ZDiamond\big) \setminus \big( \bdry \ddomain \; \cup \, \set{ \zsbis \in \ddomain \; | \; f(\zsbis) \ne 0 } \big)
\end{align*}
of interior diamond vertices excluding the support of~$f$.
\end{lemma}

\subsection{Field polynomials}\label{sub: field polynomials}
Recall the rough idea of a local field in a field theory:
a quantity that is determined by the values of
the basic fields of the theory in a microscopic neighborhood of
its point of insertion, in a manner that does not depend on the
domain, boundary conditions, or other details.
In the lattice model context,
an abstract local field is meant to encode a rule to construct
concrete random variables from the basic degrees of freedom in a
finite set of lattice sites around any point of any discrete domain~---
see Figure~\ref{fig: local fields intuition} again for an illustration.
%

We define a \term{field polynomial} of the (gradient of the) DGFF to be a
polynomial in indeterminates $\field(\zu)$ indexed by the points $\zu\in\Z^2$ of the
unit-mesh square grid, i.e., an element of the polynomial ring
\begin{align}\label{eq: space of field polynomials}
\dLocFi \coloneqq \C \big[ \field(\zu) \, \colon \zu\in\Z^2 \big] \, .
\end{align}
Any such polynomial is ``local'' in the sense that there are only 
finitely many terms in the polynomial. We define the \term{support} of a field
polynomial $F \in \dLocFi$ to be the minimally chosen (finite) subset of those~$\zu \in \Z^2$
such that $\field(\zu)$ appears in~$F$,
\begin{align}\label{eq: support}
\PolySupp F \, := \; \bigcap \set{ S \subset \Z^2 \; \Big| \; F \in \C[\field(\zu) \,\colon \zu \in S]} .
\end{align}
When a discrete domain~$\ddomain \subset \SqLatMesh$, a choice of boundary conditions,
and a point~$\zs \in \ddomain$ are given,
we define an \term{evaluation} of field polynomials
\begin{align*}
\ev_\zs^{\ddomain} \colon \dLocFi \to \set{\text{random variables for the DGFF in $\ddomain$}} .
\end{align*}
The evaluation of general field polynomials will be determined by linear extension of
an evaluation of monomials~$\field(\zu_1) \cdots \field(\zu_n)$.
When $\zs + \meshsz \zu_1 , \ldots, \zs + \meshsz \zu_n \in \ddomain$, 
such a monomial is evaluated to
\begin{align*}
\ev_\zs^{\ddomain} \colon
	\field(\zu_1) \cdots \field(\zu_n)
	\; \mapsto \;
	\Big( \dgffptwise(\zs + \meshsz \zu_1)-\dgffptwise(\zs) \Big)
        \cdots \Big( \dgffptwise(\zs + \meshsz \zu_n)-\dgffptwise(\zs) \Big) \, .
\end{align*}
This is the only case we actually care about, because when $(\ddomain)_{\meshsz > 0}$ are
discrete approximations to a continuum domain, i.e., an open set~$\domain \subset \bC$, then
we indeed have $\zs + \meshsz \zu_j \in \ddomain$
for any small enough lattice mesh~$\meshsz > 0$.
For completeness of the definition,
monomials $\field(\zu_1) \cdots \field(\zu_n)$
such that $\zs + \meshsz \zu_j \notin \ddomain$ for some~$j$ are evaluated to zero
(this is somewhat arbitrary, but the exact choice in this irrelevant case does not matter).

\begin{ex}\label{eq: currents as local fields}
For any edge midpoint $\zu_\medial\in\ZMedial$ of the unit-mesh square lattice,
the linear combinations
\begin{align*}
\holcurrfield(\zu_\medial) := \; & \phantom{-} \ii \, \pdee\field(\zu_\medial)
    = \begin{cases}
      \field(\zs + \frac{1}{2}) - \field(\zs - \frac{1}{2}) & \text{ if $\zu_\medial$ is on a horizontal edge} \\
      -\ii \, \field(\zs + \frac{\ii}{2}) + \ii \, \field(\zs - \frac{\ii}{2}) & \text{ if $\zu_\medial$ is on a vertical edge}
      \end{cases} \\
\antiholcurrfield(\zu_\medial) := \; & -\ii \, \pdeebar\field(\zu_\medial)
    = \begin{cases}
      \field(\zs + \frac{1}{2}) - \field(\zs - \frac{1}{2}) & \text{ if $\zu_\medial$ is on a horizontal edge} \\
      \ii \, \field(\zs + \frac{\ii}{2}) - \ii \, \field(\zs - \frac{\ii}{2}) & \text{ if $\zu_\medial$ is on a vertical edge}
      \end{cases}
\end{align*}
are field polynomials. When evaluated at a point~$\zs \in \ddomain$ of a discrete domain,
they give rise to the random variables
\begin{align*}
\ev^{\ddomain}_\zs \big(\holcurrfield(\zu_\medial)\big)=\HolCurr(\zs + \meshsz \zu_\medial)
\qquad \text{ and } \qquad
\ev^{\ddomain}_\zs \big(\antiholcurrfield(\zu_\medial)\big)=\AntiHolCurr(\zs + \meshsz \zu_\medial)
\end{align*}
which are the values of the discrete holomorphic and discrete antiholomorphic currents,
\eqref{eq: discrete holomorphic current} and~\eqref{eq: discrete antiholomorphic current},
on an edge at a fixed finite number of lattice steps away from the point~$\zs$.
In particular the field polynomials $\holcurrfield(\frac{\ii}{2})$, $\holcurrfield(\frac{1}{2})$,
$\holcurrfield(\frac{-\ii}{2})$ and $\holcurrfield(\frac{-1}{2})$ evaluate to the discrete
holomorphic currents on the edge to the north, east, south, and west of~$\zs$, respectively.
Trusting that no confusion arises, the field polynomial valued functions
$\holcurrfield \colon \ZMedial \to \dLocFi$ and
$\antiholcurrfield \colon \ZMedial \to \dLocFi$ on the unit-mesh medial lattice
will still be referred to as the discrete holomorphic and antiholomorphic currents.
\hfill$\diamond$
\end{ex}

\subsection{Local fields}\label{sub: corr equiv}
It can happen that two different field polynomials produce random variables
which have indistinguishable correlation functions with anything at a macroscopic distance
away from their insertion.
Such field polynomials do not
then really represent different observable quantities in a field theory, so we want to identify
them. Following~\cite{HKV}, we now define precisely the equivalence relation, and then define
local fields as the equivalence classes.
These local fields are going to be the main object of our interest.

We say that a field polynomial $F\in\dLocFi$ is a \term{null field} if
\begin{align}
\nonumber
\text{for } \quad
    & \text{all discrete domains~$\ddomain \subset \SqLatMesh$}, \\
\nonumber
    & \text{both choices of boundary conditions (Dirichlet or Neumann)}  \\
\nonumber
    & \text{any point } \zs \in \ddomain  \\
\nonumber
    & \text{any  $n \in \bN$ and all test functions} f_1 , \ldots, f_n \in \ddistribMC{\ddomain}  \\
\nonumber
\text{such that } \quad
    & \zs + \meshsz \, \PolySupp F \; \subset \; \ddomain \setminus \bigcup_{j=1}^n \FunSupp f_j \\ 
\label{eq: null field condition}
\text{we have } \quad
    & \EX_{\ddomain}^{\DorN} \big[ ( \ev^{\ddomain}_\zs F ) \, \dgff{f_1} \cdots \dgff{f_n} \big] = 0 \,.
\end{align}
We let $\dNuFi\subset \dLocFi$ denote the set of null fields.

\begin{ex}\label{ex: trivialer null field}
Any monomial $\field(\zu_1)\cdots\field(\zu_n) \in \dLocFi$ with $\zu_j = 0$ for some $j$ is
trivially null, since in its evaluation as a random variable, the factor
$\field(\zu_j)$ becomes
$\dgffptwise(\zs+\meshsz \zu_j) - \dgffptwise(\zs) = \dgffptwise(\zs) - \dgffptwise(\zs) = 0$.
\hfill$\diamond$
\end{ex}

\begin{ex}\label{ex: Laplacian nulls}
The discrete Laplacians of the basic field~$\field$ are null fields:
it follows easily from~\eqref{eq: 2pt function of the Laplacianf of DGFF}
and Wick's formula~\eqref{eq: Wick formula for DGFF} that
for any~$\zu \in \Z^2$ we have
\begin{align}\label{eq: Laplacian nulls}
\gLapl\field(\zu) = \sum_{\substack{\zubis \in \Z^2 \\ \zubis \sim \zu}}
        \big( \field(\zubis) - \field(\zu) \big) \in \dNuFi .
\end{align}
As a consequence  of the
factorization~\eqref{eq: factorizations of Laplacian with primal graph derivatives},
we get that for any $\zu_\diamond \in \ZDiamond$
\begin{align}\label{eq: current derivative nulls}
\gdeebar \holcurrfield(\zu_\diamond) \in \dNuFi
\qquad \text{ and } \qquad
\gdee \antiholcurrfield(\zu_\diamond) \in \dNuFi .
\end{align}
This observation is closely related to
Lemma~\ref{lem: discrete holomorphicity properties of currents}; it is a
discrete holomorphicity (resp. antiholomorphicity) property of the current,
now viewed as a field polynomial.
\hfill$\diamond$
\end{ex}
The null fields~\eqref{eq: Laplacian nulls}
are closely related to
the equations of motion of the theory: the minimizers of the discrete Dirichlet energy are
discrete harmonic functions.
It is natural to anticipate that the null fields $\gLapl \field(\zu)$ will
play a particularly important role.

We record one more explicit form of null fields that will appear in later calculations.

\begin{ex}\label{ex: quadratic nulls}
For any $\zu,\zubis\in\Z^2$, the field polynomial
\begin{align*}
\big( \gLapl\field(\zu) \big) \field(\zubis)+4\pi(\delta_{\zu,\zubis}-\delta_{\zu,0})
\end{align*} 
is null again by~\eqref{eq: 2pt function of the Laplacianf of DGFF} and Wick's
formula~\eqref{eq: Wick formula for DGFF}.
\hfill$\diamond$
\end{ex}

The space of \term{local fields} of (the gradient of) the DGFF 
is now defined as the quotient
\begin{align}\label{eq: space of correlation equivalence classes of fields}
\dFields := \dLocFi / \dNuFi \, .
\end{align}
In other words, we view two field polynomials $F_1, F_2 \in \dLocFi$ as correlation equivalent
if their difference $F_1-F_2$ is null. Since working with concrete field polynomial representatives
is still often convenient, we write $F+\Null$ for the local field that is the equivalence class of
$F\in\dLocFi$ in $\dFields$.

Note that while $\dNuFi \subset \dLocFi$ is evidently a vector subspace,
Examples~\ref{ex: quadratic nulls} and~\ref{ex: Laplacian nulls} show that it is not an
ideal in~$\dLocFi = \C[\field(\zu) \, \colon \zu\in\Z^2]$ with the usual 
polynomial ring structure.
In particular, the space of local
fields~\eqref{eq: space of correlation equivalence classes of fields}
does not inherit any obvious multiplication from the polynomial ring.

Different representatives of the same equivalence class in $\dFields$ may of course
have different supports~\eqref{eq: support}. A meaningful and useful notion, however, is
the \term{minimal radius of support} of a local field $F+\Null$, defined as the smallest~$r$
such that a representative with a support in a ball~\eqref{eq: lattice ball}
of $r$ lattice units exists, i.e.,
\begin{align*}
\RadSupp{F + \Null} := \min \set{r \in \Znn \; \Big| \; \exists \tilde{F}\in\dLocFi :
                \tilde{F}-F \in \dNuFi
                \, \text{ and } \,
                \PolySupp \tilde{F}\subset \Ball(r)
                } .
\end{align*}
Since exactly determining the minimal radius of support for interesting fields is
not straightforward (requires controlling all choices of representatives),
we postpone examples for later.

Understanding the space~\eqref{eq: space of correlation equivalence classes of fields}
of local fields
is at the heart of the present work.
Let us pause to comment on why this is nontrivial.
The space ${\dFields = \dLocFi / \dNuFi}$ is formed as a quotient by null fields,
which are defined by a philosophically motivated
but very implicit condition. 
No apparent tractable procedure exists to decide whether two polynomials
$F_1, F_2 \in \dLocFi$ differ by a null field:
by definition this would involve inspecting correlation
functions in all discrete domains, with all boundary conditions,
at all points, and with all possible other fields.
With enough algebraic structure on the space~$\dFields$
and some concretely verifiable dimension bounds,
we will, however, ultimately arrive at a fully explicit description
of the quotient~$\dFields = \dLocFi / \dNuFi$.

\begin{figure}[h!]
\centering
\begin{overpic}[scale=0.776, tics=10]{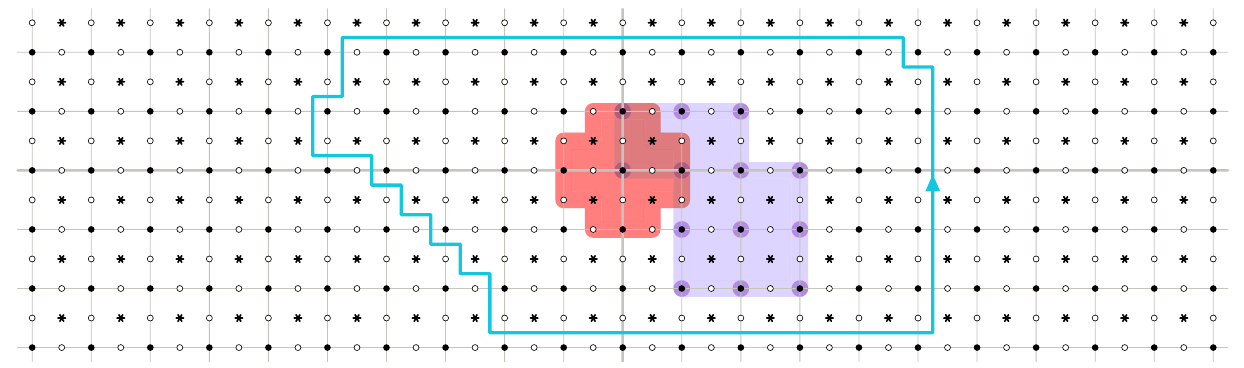}
\end{overpic}
\caption{A visualization of the definition \eqref{eq: discrete current modes} for $k=-2$: In purple the support of
a representative $F$, in red the set $\Poles{-2}$ of nonholomorphicity of the monomial $\zu\longmapsto\zu^{[-2]}$,
and in light blue a choice of positively oriented corner contour $\gamma$;
all of them laid on the infinite square grid $\Z^2$ and its sublattices $\ZDiamond$ and $\ZMedial$.
\label{fig: def currents}}
\end{figure}

\subsection{Current modes}\label{sub: current modes}
On the space~\eqref{eq: space of correlation equivalence classes of fields}
of correlation equivalence classes of local fields
there are operators closely analogous to the Laurent modes of the currents as
operators in conformal field theory.

The \term{holomorphic current modes} $\dHolCurrMode{k}$ and
\term{antiholomorphic current modes} $\dAntHolCurrMode{k}$, for $k\in\Z$,
are the linear operators on the space of fields $\dFields$ defined by
\begin{align}\label{eq: discrete current modes}
\dHolCurrMode{k}(F+\Null)
	:= \; & \phantom{-}\;\, \frac{1}{2\pi\ii}
        \dcoint{\gamma} \, \zu_\diamond^{[k]} \, \holcurrfield(\zu_\medial) \, F \; \dd{\zu} \, + \, \Null\,
	\\ \nonumber
\dAntHolCurrMode{k}(F+\Null)
	:= \; & -\frac{1}{2\pi\ii}
        \dcoint{\gamma} \, \cconj{\zu}_\diamond^{[k]} \, \antiholcurrfield(\zu_\medial) \, F \; \dd{\cconj \zu} \, + \, \Null\, ,
\end{align}
where $\holcurrfield \colon \ZMedial \to \dLocFi$ and $\antiholcurrfield  \colon \ZMedial \to \dLocFi$
are the holomorphic and antiholomorphic currents
from Example~\ref{eq: currents as local fields}, and
the discrete integration in the sense~\eqref{eq: discrete integration} is performed over
any positively oriented corner contour~$\gamma$ that surrounds both the support of the
field polynomial~$F \in \dLocFi$
and the set~\eqref{eq: monomial pole support set} of nonholomorphicity of the
Laurent monomial~$\zu \mapsto \zu^{[k]}$; i.e., $\gamma$ is required to
satisfy $\Poles{k} \cup \PolySupp F \, \subset \, \interior\gamma$.
The setup is illustrated in Figure~\ref{fig: def currents}.
To see that the operators $\dHolCurrMode{k}, \dAntHolCurrMode{k} \colon \dFields \to \dFields$
are indeed well-defined by~\eqref{eq: discrete current modes},
one checks that the discrete integral on the right hand side is in the same correlation equivalence
class for any representative~$F$ and for
any allowed choice of $\gamma$, see~\cite[Lemma 4.2]{HKV}.

\begin{ex}\label{ex: id field}\label{ex: id primary}
The \term{identity field} 
is defined as 
\begin{align*}
\idField := 1 + \Null \in \dFields \, .
\end{align*}
We note that
\begin{align*}
\idField \ne 0 \, , 
\end{align*}
or equivalently $1 \notin \dNuFi$,
because the random variable~$\ev_\zs^{\ddomain} (1) = 1$ (the evaluation of $1 \in \dLocFi$)
has a nonvanishing expected value.
This may sound like
an utterly trivial observation, but let us emphasize that $1$ is the
\emph{only} field polynomial we need to explicitly verify not to be null;
current modes allow us to generate new local fields starting from~$\idField$,
and algebraic considerations will imply that we get infinitely many linearly independent ones.
Let us first consider the action of a single current mode on the identity field.
For $k\in\Z$, using Stokes' formula~\eqref{eq: discrete Stokes} and the
definition $\holcurrfield=\ii\pdee\field$, we find
\begin{align*}
    \dHolCurrMode{k}\idField
    = &\ 
    \frac{1}{2\pi\ii}
    \dcoint{\gamma}
    \zu_\diamond^{[k]}
    \holcurrfield(\zu_\medial)\,
    \dd{\zu}
    +
    \Null
    \\
    = &\ 
    \frac{\ii}{2\pi}
    \sum_{\zu_\medial\in\interior_\medial\gamma}
    \deebar \zu_\medial^{[k]}\pdee\field(\zu_\medial)
    +
    \frac{\ii}{2\pi}
    \sum_{\zu_\diamond\in\interior_\diamond\gamma}
    \zu_\diamond^{[k]}\deebar\pdee\field(\zu_\diamond)
    +
    \Null\,.
\end{align*}
The second term is null for any $k\in\Z$, by the factorization of the
Laplacian~\eqref{eq: factorizations of Laplacian with primal graph derivatives}
and the null fields in Example~\ref{ex: Laplacian nulls}.
For $k\geq 0$, the first term is identically $0$ by the discrete holomorphicity
of the monomials --- see Proposition \ref{prop: monomials}.
We conclude
\begin{align}\label{eq: current primary property of id}
\dHolCurrMode{k} \idField = 0 \; \in \dFields
\qquad \text{ for all } k\in\Znn .
\end{align}
A formula that is also valid for~$k < 0$ is obtained by keeping the first term,
\begin{align}\label{eq: representative for Jk}
\dHolCurrMode{k} \idField = \dRep{k} + \Null
\qquad \text{ where } \qquad
\dRep{k} := \frac{\ii}{2\pi}
    \sum_{\zu_\medial\in\Poles{k}_\medial}\deebar \zu_\medial^{[k]}\pdee\field(\zu_\medial) \, .
\end{align}
Note that the formula for the representative $\dRep{k} \in \dLocFi$ here is manifestly independent
of the choice of the contour~$\gamma$. The support of the representative is also
transparently related to the nonholomorphicity set~\eqref{eq: monomial pole support set} of
the Laurent monomial.
Specifically, using Remark~\ref{rmk: pole extension} we obtain that
\begin{align*}
\PolySupp \dRep{-k} \subset
    \Ball_\primary \bigg(\left\lfloor\frac{k}{2}\right\rfloor+1\bigg)
\quad \text{ and thus } \quad
\RadSupp{\dHolCurrMode{-k}\idField}
    \leq \left\lfloor\frac{k}{2}\right\rfloor + 1 \, ,
    \quad \text{ for } k \in \Zpos.
\end{align*}
Similarly, $\dAntHolCurrMode{-k}\idField=\dRepBar{-k}+\text{\small Null}$ with $\dRepBar{k} \coloneqq
\frac{-\ii}{2\pi}
\sum_{\zu_\medial\in\Poles{k}_\medial}\dee \cconj{\zu}_\medial^{[k]}\pdeebar\field(\zu_\medial)$ and the same arguments yield $\RadSupp{\dAntHolCurrMode{-k}\idField}
\leq
\left\lfloor k/2\right\rfloor + 1$.
\hfill$\diamond$
\end{ex}

The next proposition states that the space of correlation-equivalent local fields $\dFields$ can be equipped
with a representation of two commuting copies of~$\Hei$; one for the holomorphic and one for the antiholomorphic
chirality. The construction of the holomorphic representation was the key content in~\cite{HKV}, and the
antiholomorphic one essentially repeats the same~--- but it is here that the small convention differences
become important: with the exact conventions of~\cite{HKV}, the two chiralities would fail to commute!
With our slightly modified Laurent monomials, the desired commutation property is recovered.
\begin{prop}\label{prop: comm relations}
The current modes satisfy, for all $k,\ell\in\Z$,
\begin{align*}
[\dHolCurrMode{k},\dHolCurrMode{\ell}]
    = [\dAntHolCurrMode{k},\dAntHolCurrMode{\ell}]
    = k \, \delta_{k+\ell} \, \id_\dFields
    \mspace{30mu}
    \text{ and }
    \mspace{40mu}
    [\dHolCurrMode{k},\dAntHolCurrMode{\ell}] = 0 \, .
\end{align*}
\end{prop}

The proof is presented after an auxiliary result regarding discrete integration.
In that auxiliary result, we use the function
\begin{align*}
\LatticeSign \colon \ZDiamond\cup\ZMedial \to \{+1,-1\}
\qquad \text{given by} \qquad 
\LatticeSign(\zu) := (-1)^{2 \, \im(\zu)} ,
\end{align*}
i.e., the function which takes the value~$+1$ on primal vertices and on horizontal
edges
and that takes the value~$-1$ on dual vertices and vertical edges.

\begin{lemma}\label{lemma: comm integrals}
For any $\ell,k\in\Z$, we have
\begin{align*}
\dcoint{\gamma} \cconj{\zu}_\diamond^{[\ell]} \; \LatticeSign(\zu_\medial) \, \zu_\medial^{[k]} \; \dd{\cconj \zu}
    = 0 \, ,
\end{align*}
where $\gamma$ is any positively oriented corner contour surrounding the
support of the discrete poles of both discrete monomials, i.e.,
$\Poles{k}\cup\Poles{\ell} \subset \interior \gamma$. 
\end{lemma}
\begin{proof}
Note, first of all, that, for any function $f$ on $\ZDiamond\cup\ZMedial$, we have
\[ \gdee(\LatticeSign\cdot f) = \LatticeSign\cdot\gdeebar f
\qquad \text{ and } \qquad
\gdeebar(\LatticeSign\cdot f) = \LatticeSign\cdot\gdee f 
\] where the
dot $\cdot$ stands for pointwise multiplication of functions.
By Stokes' formula~\eqref{eq: discrete Stokes} combined with this observation,
the integral can be written as
\begin{align*}
\dcoint{\gamma} \cconj{\zu}_\diamond^{[\ell]} \; \LatticeSign(\zu_\medial) \, \zu_\medial^{[k]} \; \dd{\cconj \zu}
& \; =  - \ii \, \sum_{\zubis_\diamond \in \interior_\diamond \gamma} \cconj{\zubis}_\diamond^{[\ell]}
                        \, \LatticeSign(\zubis_\diamond) \, \big( \gdeebar \zubis_\diamond^{[k]} \big)
        - \ii \, \sum_{\zubis_\medial \in \interior_\medial \gamma} \big( \gdee \cconj{\zubis}_\medial^{[\ell]} \big)
                        \, \LatticeSign(\zubis_\medial) \, \zubis_\medial^{[k]} \\
& \; =  - \ii \, \sum_{\zubis_\diamond \in \Poles{k}_\diamond} \cconj{\zubis}_\diamond^{[\ell]}
                        \, \LatticeSign(\zubis_\diamond) \, \big( \gdeebar \zubis_\diamond^{[k]} \big)
        - \ii \, \sum_{\zubis_\medial \in \Poles{\ell}_\medial} \cconj{\big( \gdeebar \zubis_\medial^{[\ell]} \big)}
                        \, \LatticeSign(\zubis_\medial) \, \zubis_\medial^{[k]} .
\end{align*}
If $k,\ell \ge 0$ then there are no discrete poles that contribute, $\Poles{\ell} = \Poles{k} = \emptyset$,
so the integral vanishes as asserted.

If $k + \ell \le -2$, then using the fact that the integral does not depend on $\gamma$, we can argue as
follows. Taking $\gamma$ to be a symmetric square path at distance~$r$ from the origin,
the integrand is $\OO(r^{k+\ell})$ and the length of the integration contour is $\OO(r)$, so the integral
is $\OO(r^{1+k+\ell})$, and taking $r \to \infty$ shows that it must vanish.

It remains to consider the case $k < 0$ and $\ell \ge -1-k \ge 0$, and the case
$\ell < 0$ and $k \ge -1-\ell \ge 0$. By repeated integration by parts,
these can be reduced to cases when the negative exponent is~$-1$.

For example if~$\ell < 0$ then
with~\eqref{eq: integration by parts} we can rewrite the integral as
\begin{align*}
\dcoint{\gamma} \cconj{\zu}_\diamond^{[\ell]} \; \LatticeSign(\zu_\medial) \, \zu_\medial^{[k]} \; \dd{\cconj \zu}
& \; = \phantom{-} \frac{1}{\ell+1} \dcoint{\gamma} \big( \gdeebar \cconj{\zu}_\diamond^{[\ell+1]} \big)
                                \; \LatticeSign(\zu_\medial) \, \zu_\medial^{[k]} \; \dd{\cconj \zu} \\
& \; = - \, \frac{1}{\ell+1} \dcoint{\gamma} \cconj{\zu}_\medial^{[\ell+1]}
                                \; \LatticeSign(\zu_\diamond) \, \big( \gdee \zu_\diamond^{[k]} \big) \; \dd{\cconj \zu} \\
& \; = - \, \frac{k}{\ell+1} \dcoint{\gamma} \cconj{\zu}_\medial^{[\ell+1]}
                                \; \LatticeSign(\zu_\diamond) \, \zu_\diamond^{[k-1]} \; \dd{\cconj \zu} .
\end{align*}
Applying this recursively, the integral is seen to be proportional to either
\begin{align*}
\dcoint{\gamma} \cconj{\zu}_\medial^{[-1]}\; 
    \LatticeSign(\zu_\diamond) \zu_\diamond^{[k+\ell+1]} \; \dd{\cconj \zu} 
\qquad \text{ or } \qquad
\dcoint{\gamma} \cconj{\zu}_\diamond^{[-1]} \; 
    \LatticeSign(\zu_\medial) \zu_\medial^{[k+\ell+1]} \; \dd{\cconj \zu} \, ,
\end{align*}
depending on the parity of $\ell < 0$.
These integrals are then evaluated by Stokes' formula~\eqref{eq: discrete Stokes}, as above.
Note that we have
$\gdeebar \zubis^{[k+\ell+1]} \equiv 0$ since $k+\ell+1 \ge 0$, so
one of the terms in Stokes' formula does not contribute.
The integrals above thus become, up to multiplicative constants,
\begin{align*}
\sum_{\vert \zubis_\diamond\vert \leq \frac{1}{\sqrt{2}}} \, \cconj{ \big( \gdeebar \zubis_\diamond^{[-1]} \big) }
        \; \LatticeSign(\zubis_\diamond) \, \zubis_\diamond^{[k + \ell + 1]} 
\qquad \text{ or } \qquad
\sum_{\vert \zubis_\medial\vert = \frac{1}{2}} \, \cconj{ \big( \gdeebar \zubis_\medial^{[-1]} \big) }
        \; \LatticeSign(\zubis_\medial) \, \zubis_\medial^{[k + \ell + 1]} .
\end{align*}
For $k + \ell + 1 \ge 3$, the monomial $\zubis^{[k + \ell + 1]}$ vanishes on the support of
$\gdeebar \zubis^{[-1]}$, so all terms in these sums are zero, and the original integral vanishes again
as asserted. Only in the cases $k+\ell+1 \in \set{0, 1, 2}$ the sums above have nonzero terms.
In the case $k+\ell+1 = 1$ the summands are odd, and they therefore cancel.
In the case $k+\ell+1 = 0$ there are equal contributions with both signs of~$\LatticeSign$,
and they therefore cancel.
In the case $k+\ell+1 = 2$ on the medial lattice the relevant values of the
monomial~$\zubis_\medial^{[2]}$ are again zero, and on the diamond lattice
the value at the origin is zero and the other four values change sign under
$90^\circ$-rotations, by virtue of symmetries of the monomial~$\zubis_\diamond^{[2]}$,
leading to cancellations again.

Similarly if $k<0$, with repeated integrations by parts one reduces to two cases with
a first order pole, both of which are evaluated by Stokes' formula, and both of which can be
explicitly seen to vanish~--- case by case according to the value of~$k+\ell+1 \ge 0$.
\end{proof}

\begin{figure}[h!]
\centering
\begin{overpic}[scale=0.776, tics=10]{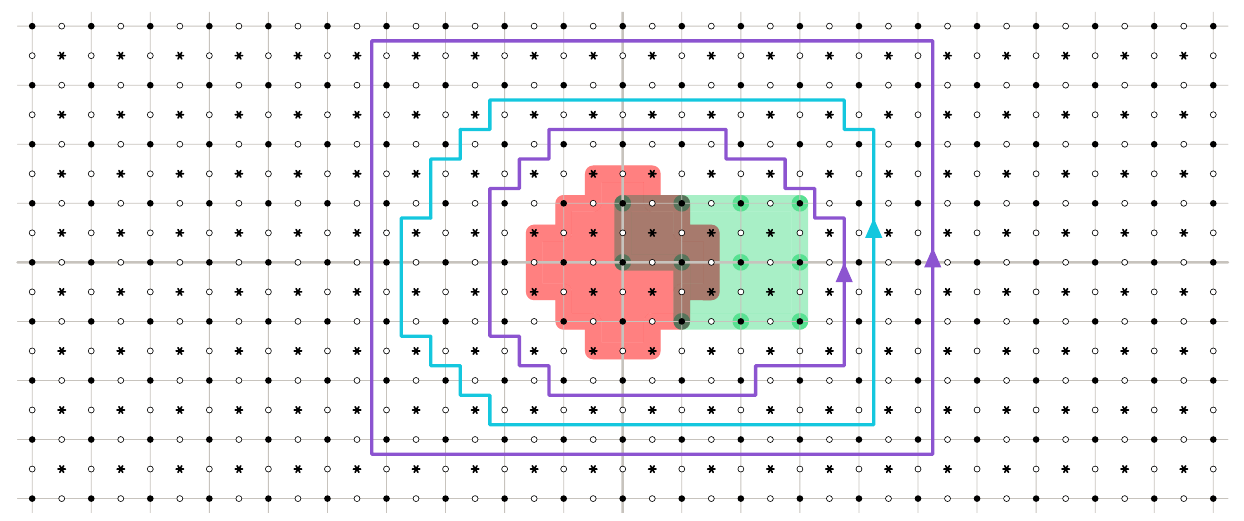}
\end{overpic}
\centering
\caption{An example of the sets and contours involved in the computation
of $[\dHolCurrMode{k},\dAntHolCurrMode{\ell}]$. In red the set
$\Poles{k}\cup\Poles{\ell}$ of nonholomorphicity of the relevant monomials,
in green the support of a representative $F$, in purple a choice of corner
contours $\gamma_-$ (inner) and $\gamma_+$ (outer), and in light blue a
choice of corner contour $\gamma$;
all of them laid on the infinite square grid $\Z^2$ and its
sublattices $\ZDiamond$ and $\ZMedial$.
\label{fig: nested integrals}}
\end{figure}

\noindent
\begin{proof}[Proof of Proposition~\ref{prop: comm relations}]
For the proof of $[\dHolCurrMode{k},\dHolCurrMode{\ell}] =
[\dAntHolCurrMode{k},\dAntHolCurrMode{\ell}] = k \, \delta_{k+\ell} \, \id_\dFields$
we refer the reader to \cite[Proposition 4.5]{HKV}; the minor differences in our
conventions do not affect the essence of this part of the proof. It remains to prove
that $[\dHolCurrMode{k},\dAntHolCurrMode{\ell}] = 0$.
Let $F\in\dLocFi$ be a local field.
Take three disjoint corner contours $\gamma_-$, $\gamma$, and $\gamma_+$ satisfying
\begin{align*}
\PolySupp F \, \cup \, \Poles{k} \, \cup \, \Poles{\ell}
    \; \subset \; \interior\gamma_-
    \; \subset \; \interior\gamma
    \; \subset \; \interior\gamma_+ \, ,
\end{align*}
i.e., the contours must be nested and all surround
the support of $F$ as well as the sets~\eqref{eq: monomial pole support set}
of failure of the discrete (anti)holomorphicity of the monomials of order~$k$ and~$\ell$;
see Figure~\ref{fig: nested integrals}.
Then, by the discrete Stokes' formula~\eqref{eq: discrete Stokes} and the
factorization~\eqref{eq: factorizations of Laplacian with primal graph derivatives}
of the discrete Laplacian we calculate
\begin{align*}
[\dHolCurrMode{k},\dAntHolCurrMode{\ell}] \big( F + \Null\big) 
= &\ \bigg( \frac{1}{2\pi\ii} \dcoint{\gamma_+} \zu_\diamond^{[k]} \holcurrfield(\zu_\medial) \dd{\zu} \bigg)
    \bigg( \frac{-1}{2\pi\ii} \dcoint{\gamma} \cconj{\zubis}_\diamond^{[\ell]}
        \antiholcurrfield(\zubis_\medial) \dd{\cconj \zubis} \bigg) \,F \\
& \qquad - \bigg( \frac{-1}{2\pi\ii} \dcoint{\gamma} \cconj{\zubis}_\diamond^{[\ell]}
        \antiholcurrfield(\zubis_\medial) \dd{\cconj \zubis} \bigg)
    \bigg( \frac{1}{2\pi\ii} \dcoint{\gamma_-} \zu_\diamond^{[k]} \holcurrfield(\zu_\medial) \dd{\zu} \bigg) \,F
    + \Null \\
= &\  \frac{1}{4\pi^2} \dcoint{\gamma}\dd{\cconj \zubis}\,\cconj{\zubis}_\diamond^{[\ell]}
    \bigg[ \dcoint{\gamma_+}-\dcoint{\gamma_-} \bigg] \dd{\zu}\,\zu_\diamond^{[k]}  
        \pdeebar\field(\zubis_\medial) \pdee\field(\zu_\medial) \,F
    + \Null \\
= &\  \frac{\ii}{4\pi^2} \dcoint{\gamma}\dd{\cconj \zubis}\,\cconj{\zubis}_\diamond^{[\ell]}
    \sum_{\zu_\diamond \in \interior_\diamond\gamma_+ \, \setminus \; \interior_\diamond\gamma_-}
        \zu_\diamond^{[k]} \pdeebar\field(\zubis_\medial) \deebar\pdee\field(\zu_\diamond) \,F
    + \Null  \\
= &\  \frac{\ii}{8\pi^2} \dcoint{\gamma}\dd{\cconj \zubis}\,\cconj{\zubis}_\diamond^{[\ell]}
    \sum_{\zu_\primary \in \interior_\primary\gamma_+ \, \setminus \; \interior_\primary\gamma_-}
        \zu_\primary^{[k]} \pdeebar\field(\zubis_\medial) \gLapl \field(\zu_\primary) \,F
    + \Null \, .
\end{align*}
Now recall the quadratic null fields of
Example~\ref{ex: quadratic nulls}, which yield in particular that
\begin{align*}
\pdeebar\field(\zubis_\medial) \gLapl \field(\zu_\primary)
    \; = \; 4 \pi (\pdeebar)_{\zubis_\medial} \delta_{\zu_\primary , \zubis_\medial} + \Null \, .
\end{align*}
We may therefore simplify the earlier calculation to
\begin{align*}
[\dHolCurrMode{k},\dAntHolCurrMode{\ell}] \big( F + \Null\big)
= &\  \frac{\ii}{2\pi} \dcoint{\gamma}\dd{\cconj \zubis}\,\cconj{\zubis}_\diamond^{[\ell]}
    \sum_{\underset{\zu_\primary\notin\interior_\primary\gamma_-}{\zu_\primary\in\interior_\primary\gamma_+}}
        \zu_\primary^{[k]} (\pdeebar)_{\zubis_\medial} (\delta_{\zu_\primary,\zubis_\medial}) \,F
    + \Null \\
= &\  \frac{\ii}{2\pi} \dcoint{\gamma}\dd{\cconj \zubis}\,\cconj{\zubis}_\diamond^{[\ell]}
    \pdeebar \zubis_\medial^{[k]} \,F
    + \Null\,.
\end{align*}
Now note that if $f$ is a function on $\ZDiamond \cup \ZMedial$, then at any edge
midpoint $\zubis_\medial \in \ZMedial$ we have
\begin{align*}
\pdeebar f (\zubis_\medial) = \gdeebar f (\zubis_\medial) + \LatticeSign (\zubis_\medial) \, \gdee f (\zubis_\medial) .
\end{align*}
In particular if $f$ is discrete holomorphic we have
$\pdeebar f (\zubis_\medial) = \LatticeSign (\zubis_\medial) \, \gdee f (\zubis_\medial)$.
The function $\zubis_\medial \mapsto \zubis_\medial^{[k]}$ is discrete holomorphic on $\gamma$,
because the contour~$\gamma$ surrounds the poles of the monomials by assumption,
so on $\gamma$ we get
$\pdeebar \zubis_\medial^{[k]} = k \, \LatticeSign (\zubis_\medial) \, \zubis_\medial^{[k-1]}$,
using also the derivative property of the discrete monomials from Proposition~\ref{prop: monomials}.
The earlier calculation thus simplifies to
\begin{align*}
[\dHolCurrMode{k},\dAntHolCurrMode{\ell}] \big( F + \Null\big)
= &\  \frac{\ii \, k}{2\pi} \dcoint{\gamma}\dd{\cconj \zubis}\,\cconj{\zubis}_\diamond^{[\ell]}
     \, \LatticeSign (\zubis_\medial) \, \zubis_\medial^{[k-1]} \,F
    + \Null\, ,
\end{align*}
which vanishes by Lemma~\ref{lemma: comm integrals}.
\end{proof}

By Proposition~\ref{prop: comm relations},
the space~$\dFields$ of local fields
carries representations of two commuting Heisenberg algebras.
Let us again denote correspondingly by $\UHei$ and $\AntUHei$ the two
commuting associative algebras~\eqref{eq: Heisenberg UEA},
whose representations on~$\dFields$ are determined
by formulas~\eqref{eq: discrete current modes} for their generators, i.e.,
$\HeiJ{k} \mapsto \dHolCurrMode{k}$ and 
$\AntHeiJ{k} \mapsto \dAntHolCurrMode{k}$, respectively.

Recall that our first main goal is to put
the space~$\dFields$ of
local fields
of the (gradient of the) discrete Gaussian Free Field
into a one-to-one correspondence with the two-chiral Fock space
$\FullFock$, which serves as the space of local fields for the
bosonic CFT of the (gradient of the) continuum Gaussian Free Field.
The Fock space~$\FullFock$ carries representations of two
commuting copies of the Heisenberg algebra by construction,
and we want our correspondence to respect this structure, i.e.,
to be a map of representations of $\UHei \tens \AntUHei$.
By purely algebraic arguments, one direction of the desired one-to-one
correspondence now becomes very easy: an isomorphic copy of the
Fock space is found inside the space~$\dFields$ of local fields
of the DGFF as follows.
\begin{coro}\label{cor: easy inclusion}
The $\UHei \tens \AntUHei$ subrepresentation
in $\dFields$ generated by the identity field~$\idField$
(Example~\ref{ex: id field}) is isomorphic to the
full Fock space,
\begin{align*}
\FullFock \, \isom \, (\UHei \tens \AntUHei) \idField \, \subset \, \dFields .
\end{align*}
\end{coro}
\begin{proof}
In~\eqref{eq: current primary property of id}
we saw that $\dHolCurrMode{k} \idField = 0$
and $\dAntHolCurrMode{k} \idField = 0$ for all $k \ge 0$,
and we also noted that ${\idField \ne 0 \in \dFields}$. The asserted isomorphism
therefore follows immediately from Lemma~\ref{lem: abstract nonsense}.
\end{proof}


\subsection{Sugawara construction and homogeneous local fields}\label{sub: Sugawara}
In Section~\ref{ssec: Sugawara basics}, we recalled how the Fock space representation of the Heisenberg algebra can be rendered a Virasoro representation via the Sugawara construction.
A key observation in~\cite{HKV} was that the same construction can be applied in
the space~$\dFields$ of correlation equivalence classes of local fields
of the discrete GFF, thanks to the following truncation property.
\begin{lemma}\label{lemma: truncation}
	For any $F\in\dLocFi$, there exists $K\in\Z_{>0}$ such that 
	\begin{align*}
		\dHolCurrMode{k}\big(\,F + \Null\,\big)
		=
		\dAntHolCurrMode{k}\big(\,F + \Null\,\big)
		=\,
		0\,+\,\Null
	\end{align*}
	for all $k\geq K$.
\end{lemma}
\begin{proof}
See \cite[Lemma~4.4]{HKV} for the holomorphic sector.
The argument is identical for the antiholomorphic sector.
\end{proof}
The important conclusion about
the space~$\dFields$ of correlation equivalence classes of local fields
is the following.
\begin{coro}
The formulas
\begin{align*}
\dvirL{n} = \frac{1}{2} \Big( \sum_{k \ge 0} \dHolCurrMode{n-k} \, \dHolCurrMode{k}
            + \sum_{k < 0} \dHolCurrMode{k} \, \dHolCurrMode{n-k} \Big)
\quad \text{ and } \quad
\dvirBarL{n} = \frac{1}{2} \Big( \sum_{k \ge 0} \dAntHolCurrMode{n-k} \, \dAntHolCurrMode{k}
            + \sum_{k < 0} \dAntHolCurrMode{k} \, \dAntHolCurrMode{n-k} \Big)
\end{align*}
equip the space~$\dFields$ with two commuting Virasoro representations
with central charge~$c=1$.
\end{coro}
\begin{proof}
The Virasoro commutation relations with $c=1$ for both $\dvirL{n}$
and $\dvirBarL{n}$ are shown as in
\cite[Theorem~4.10]{HKV}, using the 
truncation property of 
Lemma~\ref{lemma: truncation}.
The mutual commutation, $[\dvirL{n}, \dvirBarL{m}]=0$,
follows from the mutual commutation of the
corresponding Heisenberg modes,
$[\dHolCurrMode{k}, \dAntHolCurrMode{\ell}]=0$, proven in
Proposition~\ref{prop: comm relations}.
\end{proof}

A particular role is played by
the holomorphic and antiholomorphic Virasoro generators with index $n=0$,
\begin{align*}
\dvirL{0}
= \frac{1}{2} \dHolCurrMode{0}\dHolCurrMode{0}
    + \sum_{k=1}^\infty\dHolCurrMode{-k}\dHolCurrMode{k} 
\qquad \text{ and } \qquad
\dvirBarL{0}
= \frac{1}{2} \dAntHolCurrMode{0} \dAntHolCurrMode{0}
    + \sum_{k=1}^\infty \dAntHolCurrMode{-k} \dAntHolCurrMode{k}\, .
\end{align*}
In CFT the sum $\dvirL{0} + \dvirBarL{0}$
is the Hamiltonian (energy) operator, which in radial quantization serves
as the infinitesimal generator of scalings.
The difference $\dvirL{0} - \dvirBarL{0}$
is the spin operator, which serves as the infinitesimal
generator of rotations. For determining the needed renormalization of
fields in the scaling limit, the eigenvalues of $\dvirL{0} + \dvirBarL{0}$
will be crucial~--- these are called the scaling dimensions of the fields.
The pair of eigenvalues for both $\dvirL{0}$ and $\dvirBarL{0}$ carries
the information on both scaling dimension and spin;
we define
the space of \term{homogeneous} local fields of \term{conformal dimensions}
$\Delta,\barDelta\in\C$ as the joint eigenspace
\begin{align}\label{eq: homogeneous local fields}
\dHomFi{\Delta}{\bar\Delta}
\, := \, \ker\big(\dvirL{0}-\Delta\big) \, \cap \, \ker\big(\dvirBarL{0} - \barDelta\big)
    \, \subset \, \dFields \, .
\end{align}
While it was very easy to see that the Fock space~$\FullFock$
has a grading~\eqref{eq: Fock bigrading} by conformal dimensions,
at this stage we have not yet established the same conclusion about the
space~$\dFields$: diagonalizability of
$\dvirL{0}$ and $\dvirBarL{0}$ and finite-dimensionality of the
joint eigenspaces still need to be proven in order for the homogeneous
components~\eqref{eq: homogeneous local fields} to be complete and usable
decomposition of correlation equivalence classes of local fields of the
discrete GFF. We can, however, already give some examples of homogeneous fields,
because by Corollary~\ref{cor: easy inclusion}, the space $\dFields$ contains a
subspace $(\UHei \tens \AntUHei) \idField$ isomorphic to the Fock space.
\begin{ex}\label{ex: L0 eigenvalues}
The basis vectors~\eqref{eq: full Fock basis} of the Fock space
are eigenvectors, and correspondingly we have fields
\begin{align}\label{eq: scaling basis fields}
\dHolCurrMode{-k_m} \cdots \dHolCurrMode{-k_2} \, \dHolCurrMode{-k_1} \, 
\dAntHolCurrMode{-k'_{m'}} \cdots \dAntHolCurrMode{-k'_2} \, \dAntHolCurrMode{-k'_1} \idField
    \; \in \, \dHomFi{\Delta}{\bar{\Delta}} & \\ \nonumber
\text{ with conformal dimensions } \;
\Delta = \sum_{i=1}^m k_i
\; \text{ and } \;
\bar{\Delta} = \sum_{j=1}^{m'} k'_j. &
\end{align}
The most obvious special case is the identity field $\idField$:
the eigenvalue properties $\dvirL{0} \idField = 0$
and $\dvirBarL{0} \idField = 0$ in fact also follow easily
from~\eqref{eq: current primary property of id}
and we indeed have $\idField \in \dHomFi{0}{0}$.
\hfill$\diamond$
\end{ex}

\section{Linear local fields of the DGFF}
	\label{sec:linear}
	Our first main goal is to fully work out the structure of the
space~$\dFields := \dLocFi/\dNuFi$
of local fields of the DGFF.
We seek to show that it has the same structure as the space of local fields
of a CFT, i.e., that it is isomorphic to the Fock space~$\FullFock$.
A priori, the difficulty stems from the fact that 
$\dFields = \dLocFi/\dNuFi$ involves a quotient by null fields,
which are defined by an implicit condition that cannot be decided
by a straightforward method.

Recall, however, that one inclusion,
$\FullFock \isom (\UHei \tens \AntUHei) \idField \subset \dFields$,
was already obtained in Corollary~\ref{cor: easy inclusion}
by virtue of the Heisenberg algebra actions of
Proposition~\ref{prop: comm relations}. Establishing the remaining
opposite inclusion
$\dFields \, \subset \, (\UHei \tens \AntUHei) \idField$
now amounts to showing that
the whole space $\dFields$ is exhausted by linear combinations of
those fields that can be obtained from the
identity field~$\idField$ by repeated actions of the Heisenberg
generators~\eqref{eq: discrete current modes}.
A natural strategy for doing that is to exhibit concrete
upper bounds for the dimensions of some suitably chosen subspaces of~$\dFields$,
and showing that the upper bounds are already saturated within the
subspace $(\UHei \tens \AntUHei) \idField \subset \dFields$.
We will carry out such a strategy in two steps in this section and the next.
The present section achieves dimension upper bounds for \emph{linear} local
fields, i.e., those corresponding to homogeneous field polynomials of
degree one. The task in Section~\ref{sec:higher} will then be to reduce the
case of higher degree fields to such linear factors.

To achieve useful dimension bounds for linear local fields, we must still
refine to further subspaces: the space of all linear local fields
(even modulo null fields) remains infinite dimensional, so counting arguments
without refinement would be doomed. What turns out to work is to
construct a filtration of linear local fields by the finite-dimensional subspaces
with at most a given radius of support.

We start in Section~\ref{subsec: def linear local fields} by
defining linear local fields and stating the 
result (Theorem~\ref{thm: basis lin loc fields}) which gives an explicit basis
for them. Then, in
Section~\ref{subsec: filtration by radius}, we introduce the filtration
with finite-dimensional subspaces in which dimension counting is to be performed.
Here we also already record the dimension lower bounds, which follow from the
earlier observation $(\UHei \tens \AntUHei) \idField \subset \dFields$
and some observations about the radii of supports of some explicit
linear field polynomials. The main task of proving the matching upper bounds
for the dimensions is done in Section~\ref{subsec: dimension upper bounds},
and once this is done, we give the proof of the basis
theorem for linear local fields (Theorem~\ref{thm: basis lin loc fields}).
We conclude in Section~\ref{subsec: homogeneous linear local fields} with
simple remarks on what the basis theorem says about homogeneous
linear local fields.

\subsection{Linear local fields}
\label{subsec: def linear local fields}

By definition~\eqref{eq: space of field polynomials},
the space of field polynomials
\begin{align*}
\dLocFi = \C[\field(\zu) \, \colon \zu\in\Z^2] 
\end{align*}
is the free commutative (polynomial) algebra generated by
the symbols~$\field(\zu)$.
The space of \term{linear field polynomials} is now defined to be the subspace
\emph{spanned} by these symbols,
\begin{align}\label{eq: space of linear field polynomials}
\dLocLinFi := \spn_\C \set{ \,\field(\zu)\,\colon \zu\in\Z^2\, }
    \, \subset \, \dLocFi\, .
\end{align}
The space of \term{linear local fields} is then again
defined by identifying field polynomials which differ by a null field,
\begin{align}\label{eq: space of correlation equivalent linear local fields}
\dLinFields := \dLocLinFi / \dNuFi \, \subset \, \dFields \, .
\end{align}

\begin{ex}\label{ex: corr-equiv lin loc fi}
Recall from Example \ref{ex: id primary} that, for $k \in \Zpos$, we can write
$\dHolCurrMode{-k}\idField = \dRep{-k} + \Null$ and
$\dAntHolCurrMode{-k}\idField = \dRepBar{-k} + \Null$
with linear field polynomials $\dRep{-k}, \dRepBar{-k} \in \dLocLinFi$
given explicitly as in~\eqref{eq: representative for Jk}.
Therefore, 
\begin{align*}
\dHolCurrMode{-k}\idField, \, \dAntHolCurrMode{-k}\idField \, \in \, \dLinFields
\end{align*}
are linear local fields.
\hfill~$\diamond$
\end{ex}

Corollary~\ref{cor: easy inclusion},
together with the fact that~\eqref{eq: full Fock basis} is a basis of the Fock space,
implies that the linear local fields given in Example~\ref{ex: corr-equiv lin loc fi}
are linearly independent.
The main goal of this section is to prove that
they in fact form a basis.
\begin{thm}\label{thm: basis lin loc fields}
The set
\begin{align}\label{eq: basis of lin loc fields}
\set{\dHolCurrMode{-k}\idField \; \big| \; k\in\Zpos} \cup
\set{\dAntHolCurrMode{-k}\idField \; \big| \; k\in\Zpos}
\end{align}
is a basis of the space $\dLinFields$ of linear local fields.
\end{thm}

We moreover obtain a complete characterization of linear null fields,
which will also be used in Section~\ref{sec:higher}. Among linear
null fields, the explicit ones given in Examples~\ref{ex: Laplacian nulls}
and~\ref{ex: trivialer null field} are all there is.

\begin{coro}\label{coro: linear nulls}
The set
\begin{align*}
\set{ \gLapl \field(\zu) \; \big| \; \zu\in\Z^2 } \cup \set{ \field(0) }
\end{align*}
spans the subspace $\dLocLinFi\cap\dNuFi$ of null linear local fields.
\end{coro}

\subsection{Filtration by radius}
\label{subsec: filtration by radius}
We now present the filtration which enables a dimension counting
argument that is the key to proving Theorem~\ref{thm: basis lin loc fields}.

Consider the \term{discrete balls} of radii~$r \in \Zpos$
\begin{align}\label{eq: discrete ball}
\Ball_\primary(r)
\, := \, \set{ \zu \in \bZ^2 \; \Big| \; \|\zu\|_1 \le r } ,
\end{align}
with respect to the Manhattan norm $\norm{z}_1 := \re(z)+\im(z)$.
In the space~\eqref{eq: space of linear field polynomials} of linear field polynomials,
the subspace of those fields whose support is in~$\Ball_\primary(r)$ is
\begin{align}\label{eq: space of linear local fields of given radius}
\dLocLinFiRad{r} := \spn_\C \set{\,\field(\zu)\,\colon \zu \in \Ball_\primary(r)\,}
    \, \subset \, \dLocLinFi \, .
\end{align}
These subspaces are evidently finite-dimensional since the subset $\Ball_\primary(r) \subset \bZ$
is finite,
\begin{align*}
\dmn \big( \dLocLinFiRad{r} \big) \le |\Ball_\primary(r)| = 2 r^2 + 2 r + 1
\qquad \text{ for } r \in \Zpos .
\end{align*}
Denote the corresponding subspace of the
quotient~\eqref{eq: space of correlation equivalent linear local fields} by
\begin{align}\label{eq: space of corr equiv linear local fields of given radius}
\dLinFieldsRad{r} := \dLocLinFiRad{r}/\dNuFi \subset \dLinFields \, .
\end{align}
These subspaces form a filtration of~$\dLinFields$,
\begin{align*}
\dLinFieldsRad{1} \subset \dLinFieldsRad{2}
    \subset \dLinFieldsRad{3}
    \subset \cdots \subset \dLinFields
\qquad \text{ and } \qquad
\sum_{r \in \Zpos} \dLinFieldsRad{r} \, = \, \dLinFields \, ,
\end{align*}
because  the discrete balls form an increasing sequence of (finite) subsets
that exhaust the square grid~$\bZ^2$,
\begin{align*}
\Ball_\primary(1) \subset \Ball_\primary(2)
    \subset \Ball_\primary(3)
    \subset \cdots \subset \bZ^2
\qquad \text{ and } \qquad
\bigcup_{r \in \Zpos} \Ball_\primary(r) = \bZ^2 \, .
\end{align*}

\begin{ex}\label{ex: modes in filtration}
Recall from Example \ref{ex: id primary} that, for $k\in\Zpos$, the
linear field polynomials $\dRep{-k}$ and $\dRepBar{-k}$ are supported
in $\Ball_\primary(r_k)$ with $r_k \coloneqq \lfloor k/2 \rfloor + 1$.
By Example \ref{ex: corr-equiv lin loc fi}, we therefore see that
$\dHolCurrMode{-k}\idField, \dAntHolCurrMode{-k}\idField \in \dLinFieldsRad{r_k}$. 
\hfill~$\diamond$
\end{ex}

Let us give one slightly more subtle example in the form of a lemma.
In this example the precise form (slightly different from \cite{HKV}) of our monomials
defined in Section~\ref{subsec: discrete monomials} again becomes important.

\begin{lemma}\label{lemma: -2r and -2r bar}
For $r\in\Zpos$, we have
\begin{align*}
\dHolCurrMode{-2r}\idField - \dAntHolCurrMode{-2r}\idField
    \, \in \, \dLinFieldsRad{r}\,.
\end{align*}
\end{lemma}
\begin{proof}
Take a corner contour $\gamma$ sufficiently large for the action of $\dHolCurrMode{-2r}$.
Using integration by parts~\eqref{eq: integration by parts} $2r-1$ times,
the properties of the discrete monomials from Proposition~\ref{prop: monomials},
Stokes' formula~\eqref{eq: discrete Stokes},
and the null fields of Example~\ref{ex: Laplacian nulls},
we calculate
\begin{align*} 
\dHolCurrMode{-2r}\idField
= \; & \frac{1}{2\pi} \dcoint{\gamma} \zu_{\diamond}^{[-2r]} \pdee \field (\zu_\medial) \dd{\zu} \,+\, \Null \\
= \; & \frac{1}{2\pi(2r-1)!}
    \dcoint{\gamma} \zu_{\medial}^{[-1]} \gdee^{2r-1}\pdee \field (\zu_\diamond) \dd{\zu} \,+\, \Null \\
= \; & \frac{\ii}{2\pi(2r-1)!} \sum_{\zu_\diamond\in\interior_\diamond\gamma}
    \gdeebar \zu_{\diamond}^{[-1]} \gdee^{2r-1}\pdee \field (\zu_\diamond) \,+\, \Null \\
= \; & \frac{\ii}{(2r-1)!} \Bigg[ \frac{1}{2}\,\gdee^{2r-1}\pdee \field (0)
        + \frac{1}{8} \sum_{\zubis=\frac{\pm 1\pm \ii}{2}} \gdee^{2r-1}\pdee \field (\zubis) \Bigg]
    \,+\, \Null\,.
\end{align*}
Similarly we get
\begin{align*} 
\dAntHolCurrMode{-2r}\idField
= \; & \frac{-\ii}{(2r-1)!} \Bigg[ \frac{1}{2}\,\gdeebar^{2r-1}\pdeebar \field (0)
        + \frac{1}{8} \sum_{\zubis=\frac{\pm 1\pm \ii}{2}} \gdeebar^{2r-1}\pdeebar \field (\zubis) \Bigg]
    \,+\, \Null\,.
\end{align*}
We now make some observations about the compositions~$\gdee^{2r-1}\pdee$
and~$\gdeebar^{2r-1}\pdeebar$ of the finite difference operators appearing in the above formulas.
The expression $\gdee^{2r-1}\pdee \field (\zubis)$ is a linear combination of~$\field (\zu)$,
with contributions from paths from~$\zubis \in \ZDiamond$ to the primary vertex~$\zu \in \ZPrimary$ which use $2r-1$
arbitrary half lattice unit steps and then one half lattice unit step along a half-edge leading to~$\zu$,
with coefficients that are products of the weights in the finite difference
operators~\eqref{eq: discrete Wirtinger derivatives}, \eqref{eq: primal lattice dee}
and~\eqref{eq: primal lattice deebar}
corresponding to the steps.
If~$\zubis \in \ZPrimary$ is a primal vertex, then necessarily an even number of vertical
half-steps is used by the path, and an even number of weight factors are imaginary, so the
coefficient is real. By contrast, if~$\zubis \in \ZDual$ is a dual vertex, then an odd number of
vertical half-steps is used, and the coefficient is imaginary. Moreover, the only difference in
the weights between~$\gdee^{2r-1}\pdee$ and $\gdeebar^{2r-1}\pdeebar$ is a different sign for the
vertical steps, so for $\zubis \in \ZPrimary$ we have
$\gdeebar^{2r-1}\pdeebar \field (\zubis) = \gdee^{2r-1}\pdee \field (\zubis)$,
and for $\zubis \in \ZDual$ we have
$\gdeebar^{2r-1}\pdeebar \field (\zubis) = - \gdee^{2r-1}\pdee \field (\zubis)$.

The above considerations of coefficients in compositions of finite difference operators
show that in~$\dHolCurrMode{-2r}\idField - \dAntHolCurrMode{-2r}\idField$
there is a
cancellation of the terms corresponding to $\zubis=\frac{\pm 1\pm \ii}{2} \in \ZDual$.
More precisely, $\dHolCurrMode{-2r}\idField - \dAntHolCurrMode{-2r}\idField$ has
a representative
\begin{align*}
\frac{\ii}{(2r-1)!} \gdee^{2r-1}\pdee \field (0) \, \in \, \dLocLinFiRad{r} \, ,
\end{align*}
where the sufficiency of the radius of support~$r$ is a consequence of using 
exactly $2r$ half-steps starting from the origin~$0 \in \ZPrimary$. This proves
$\dHolCurrMode{-2r}\idField - \dAntHolCurrMode{-2r}\idField \in \dLinFieldsRad{r}$.
\end{proof}

From the explicit examples above, we get lower bounds for the dimension of the
subspaces in the filtration.
\begin{lemma}\label{lem: dimension lower bounds in filtration}
For every $r \in \Zpos$, we have
\begin{align*}
\dmn \big( \dLinFieldsRad{r} \big) \ge 4 r - 1 .
\end{align*}
\end{lemma}
\begin{proof}
Fix $r \in \Zpos$. For $1 \le k \le 2r - 1$ we have
$\dHolCurrMode{-k}\idField, \dAntHolCurrMode{-k}\idField \in \dLinFieldsRad{r}$
by Example~\ref{ex: modes in filtration}. This gives $2(2r-1) = 4r - 2$ fields in~$\dLinFieldsRad{r}$.
Lemma~\ref{lemma: -2r and -2r bar} gives one more,
$\dHolCurrMode{-2r} \idField - \dAntHolCurrMode{-2r} \idField \in \dLinFieldsRad{r}$.
Corollary~\ref{cor: easy inclusion}
implies that these $4r-1$
fields are linearly independent.
\end{proof}

\subsection{Dimension upper bounds}
\label{subsec: dimension upper bounds}

We now need dimension upper bounds for the subspaces in the
filtration, which match the lower bounds of
Lemma~\ref{lem: dimension lower bounds in filtration}.
Starting from the obvious spanning set
$\set{\field(\zu) + \Null \; \big| \; \zu \in \Ball_\primary(r)}$
for~$\dLinFieldsRad{r}$,
which has $|\Ball_\primary(r)| = 2r^2 + 2r + 1 \gg 4 r - 1$ elements,
we use a method based on discrete harmonic measures and divergence theorem,
which yields a smaller set of canonical representatives supported on
the boundary of the lattice ball.
Besides concretely yielding the desired dimension upper bound,
these could be seen conceptually as a discrete version of the state-field
correspondence in CFT.\footnote{In, e.g., Segal's axiomatization of CFT, states
live on circular boundary components of a bordered Riemann surface domain.
Thus representing an arbitrary local field inside a
lattice ball in terms of only the field values on the boundary of the ball
corresponds to forming an incoming ``state'' on a boundary
component of the complement of the ball in the discrete domain.
For a complete analogue of the state-field correspondence,
explicit representatives of linear local fields provided in this section
should further be combined with the results of Section~\ref{sec:higher}
on higher degree local fields to represent general states.}
From a practical point of view, 
the explicit representatives also naturally
enable symbolic computer computations in the
space of local fields.\footnote{A priori,
it would not be straightforward to perform computation in the
abstract quotient $\dFields = \dLocFi/\dNuFi$.}

In order to describe the harmonic measures literally
in terms of the discrete Laplacian operators
and Green's functions introduced in Section~\ref{subsec: difference operators},
we start with choosing an appropriate discrete domain, which
is just a slight modification of the discrete ball~$\Ball_\primary(r)$.
\begin{figure}[tb]
\centering
\subfigure[The discrete domain $\DomBall{r} \subset \bZ^2$ with ${r=4}$.
The interior vertices~$\DomBallInt{r} \subset \DomBall{r}$
are colored gray.
The subset $\DomCirc{r} \subset \bdry \DomBall{r}$ of boundary vertices
where harmonic measures are supported are colored white: the other
boundary vertices are black.]
{
  \includegraphics[width=.4\textwidth]{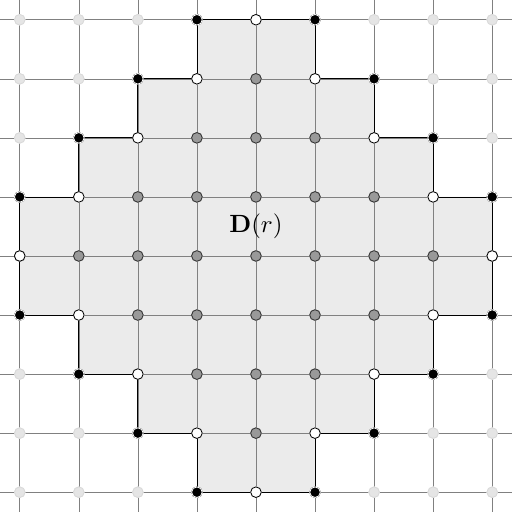}
  \label{sfig: discrete ball}
}
\hspace{1.0cm}
\subfigure[The discrete harmonic measure ${\zubis \mapsto H_\bu(\zubis)}$
at a given boundary point~$\bu \in \DomCirc{r}$ represents
the weight with which the field $\field(\zubis)$ at $\zubis \in \DomBallInt{r}$
contributes to the
coefficient of~$\field(\bu)$ in the linear local field
representative which has its support on~$\DomCirc{r}$.
In this figure $r=30$ and $\bu = 19-11\,\ii$ and the magnitudes of the values
of $H_\bu(\cdot)$ are indicated by colors.] 
{
  \includegraphics[width=.4\textwidth]{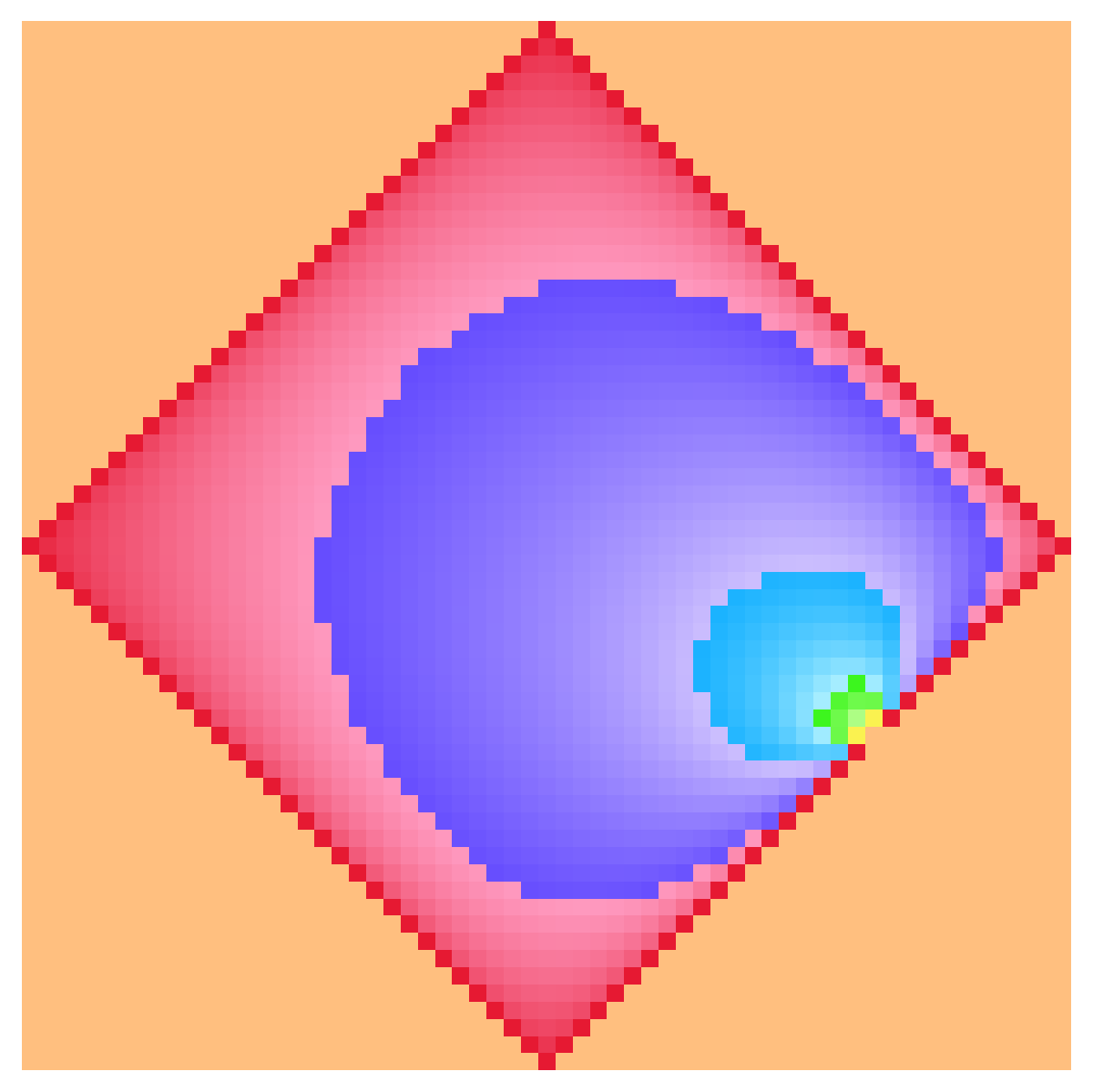}
  \label{sfig: discrete poisson kernel}
}
	\caption{The discrete ball-domains and Poisson kernels in them.}
	\label{fig: poisson kernel}
\end{figure}
For $r \in \Zpos$, consider the discrete domain~$\DomBall{r} \subset \bZ^2$
as illustrated in Figure~\ref{sec:linear}.\ref{sfig: discrete ball},
with interior
\begin{align}\label{eq: discrete ball interior}
\DomBallInt{r} = \set{ \zu \in \bZ^2 \; \Big| \; \|\zu\|_1 < r } .
\end{align}
The boundary $\bdry \DomBall{r}$ contains the subset
\begin{align}\label{eq: discrete circle}
\DomCirc{r} = \set{ \zu \in \bZ^2 \; \Big| \; \|\zu\|_1 = r } ,
\end{align}
which will carry all harmonic measure seen from the interior
points, and which is for that reason going to be the important part of the
boundary for our purposes. In Figure~\ref{sec:linear}.\ref{sfig: discrete ball},
the subset $\DomCirc{r} \subset \bdry \DomBall{r}$
and its complement $\bdry \DomBall{r} \setminus \DomCirc{r}$
are colored white and black, respectively.
Note that the discrete ball is the union of the interior and the subset which
supports the harmonic measures,
\begin{align*}
\Ball_\primary(r) = \DomBallInt{r} \cup \DomCirc{r} \, .
\end{align*}

A general graph-theoretic fact about Neumann Laplacians gives
the following ``divergenece theorem''.
\begin{lemma}\label{lem: divergence thm}
Let $\gLapl^{\Neumann; \DomBall{r}}$ denote the discrete Neumann
Laplacian in the domain~$\DomBall{r}$.
Then for any
two functions $f,g$ on $\DomBall{r}$ we have the equality
\begin{align*}
\sum_{\zu \in \DomBall{r}} f(\zu) \, \big( (\gLapl^{\Neumann; \DomBall{r}} g) (\zu) \big)
     = \sum_{\zu \in \DomBall{r}} \big( (\gLapl^{\Neumann; \DomBall{r}} f) (\zu) \big) \, g(\zu).
\end{align*}
\end{lemma}
\begin{proof}
Separate the contribution to each edge
in both sums and observe that both sides are equal to
$-\sum \big( f(\zubis) - g(\zu) \big)^2$, where the sum is over the edges
$\set{ \zu, \zubis}$ of the graph~$\DomBall{r}$.
\end{proof}

Let us now define the harmonic measures in a manner that facilitates
applying the divergence theorem.
Fix an interior vertex $\zubis \in \DomBallInt{r}$ and
consider the Dirichlet Green's function $\zu \mapsto \Green_{\DomBall{r}}^{\Dirichlet}(\zu, \zubis)$
in $\DomBall{r}$. Note that
$\gLapl^{\Neumann; \DomBall{r}} \Green_{\DomBall{r}}^{\Dirichlet} (\,\cdot\,, \zubis)$
vanishes outside~${\bdry \DomBall{r} \cup \set{\zubis}}$, because
the Neumann Laplacian $\gLapl^{\Neumann; \DomBall{r}}$ and the
Dirichlet Laplacian $\gLapl^{\Dirichlet; \DomBall{r}}$ agree at all interior
vertices and
$\gLapl^{\Dirichlet; \DomBall{r}} \Green_{\DomBall{r}}^{\Dirichlet} (\,\cdot\,, \zubis) = - \delta_{\zubis}(\,\cdot\,)$.
In fact, $\gLapl^{\Neumann; \DomBall{r}} \Green_{\DomBall{r}}^{\Dirichlet} (\,\cdot\,, \zubis)$ also vanishes on $\bdry \DomBall{r} \setminus \DomCirc{r}$, i.e.,
on the black boundary points in Figure~\ref{sec:linear}.\ref{sfig: discrete ball},
because all neighbors of such a boundary point are themselves boundary points and
the Dirichlet Green's function vanishes on the boundary.
We therefore define the harmonic measure
$H_{\bu}(\zubis)$ of $\bu \in \DomCirc{r}$ seen from
$\zubis \in \DomBallInt{r}$ by the formula
\begin{align*}
H_{\bu}(\zubis) = \gLapl^{\Neumann; \DomBall{r}} \Green_{\DomBall{r}}^{\Dirichlet} (\bu, \zubis),
\end{align*}
so that
the Neumann Laplacian of the Dirichlet Green's function
takes the following form.
\begin{lemma}\label{lem: Neumann Laplacian of Dirichlet Green}
For $\zubis \in \DomBallInt{r}$ we have
\begin{align*}
\gLapl^{\Neumann; \DomBall{r}} \Green_{\DomBall{r}}^{\Dirichlet} (\cdot, \zubis)
\; = \; - \delta_{\zubis}(\cdot) + \sum_{\bu \in \DomCirc{r}} H_\bu(\zubis) \, \delta_{\bu} (\cdot) .
\end{align*}
\end{lemma}
\begin{proof}
This follows by construction of $H_\bu(\zubis)$ and the observations preceding it.
\end{proof}
Observe also that the Neumann Laplacian always yields a zero-average
function, so the harmonic measure terms above must compensate the
term~$-\delta_\zubis(\cdot)$,
i.e., we have
\begin{align}\label{eq: harmonic measures sum up to one}
\sum_{\bu \in \DomCirc{r}} H_\bu(\zubis) \; = \; 1 \, .
\end{align}

With Lemma~\ref{lem: Neumann Laplacian of Dirichlet Green},
we are ready to give the (almost) canonical representatives for linear
local fields supported on the discrete
ball $\Ball_\primary(r) = \DomBallInt{r} \cup \DomCirc{r}$.
\begin{prop}\label{prop: canonical representatives}
For any $\zubis \in \DomBallInt{r}$, we have
\begin{align*}
\field(\zubis) + \Null = \sum_{\bu \in \DomCirc{r}} H_\bu(\zubis) \field(\bu) + \Null .
\end{align*}
\end{prop}
\begin{proof}
Use the
divergence theorem, Lemma~\ref{lem: divergence thm}, with 
$f = \Green_{\DomBall{r}}^{\Dirichlet} (\cdot, \zubis)$ and $g = \field$ to get
\begin{align*}
\sum_{\zu \in \DomBallInt{r}} \Green_{\DomBall{r}}^{\Dirichlet} (\zu, \zubis) \, \big( (\gLapl^{\Neumann; \DomBall{r}} \field) (\zu) \big)
\; = \; - \field(\zubis) + \sum_{\bu \in \DomCirc{r}} H_{\bu}(\zubis) \, \field(\bu) ,
\end{align*}
where on the left-hand-side we used the vanishing of the Dirichlet Green's function~$f$
on the boundary to omit the boundary terms, and on the right-hand-side we used the
formula of Lemma~\ref{lem: Neumann Laplacian of Dirichlet Green} for the Neumann
Laplacian of the Dirichlet Green's function~$f$ to omit all terms except $\zu = \zubis$
and $\zu \in \DomCirc{r}$. On the left hand side, since~$\zu$ is an interior point
of~$\DomBall{r}$, the Neumann Laplacian coincides with the square grid Laplacian, so
$(\gLapl^{\Neumann; \DomBall{r}} \field) (\zu) = \gLapl \field (\zu)$,
which is null by Example~\ref{ex: Laplacian nulls}. As a linear combination of null
fields, the entire left hand side is therefore null.
This proves the asserted equality modulo null fields, 
since the right hand side is exactly the difference of the nonnull terms on the
two sides of the asserted equality.
\end{proof}

We thus obtain representatives which are supported just on the boundary.
\begin{lemma}\label{lem: spanning set in filtration}
For every $r \in \Zpos$, we have
\begin{align*}
\dLinFieldsRad{r}
\; = \; \spn_\bC \set{ \field(\bu) + \Null \; \Big| \; \bu \in \DomCirc{r}} \, .
\end{align*}
\end{lemma}
\begin{proof}
By definition, the space~$\dLocLinFiRad{r}$ is spanned by
$\field(\zubis)$ for $\zubis \in \Ball_\primary(r) = \DomBallInt{r} \cup \DomCirc{r}$,
and $\dLinFieldsRad{r}$ is correspondingly spanned by 
the equivalence classes modulo null fields. For any $\zubis \in \DomBallInt{r}$,
Proposition~\ref{prop: canonical representatives} expresses $\field(\zubis) + \Null$
as a linear combination of fields of the form $\field(\bu) + \Null$
with $\bu \in \DomCirc{r}$, so these suffice to span~$\dLinFieldsRad{r}$.
\end{proof}

This also leads to the desired dimension upper bounds.
\begin{lemma}\label{lem: dimension upper bounds in filtration}
For every $r \in \Zpos$, we have
\begin{align*}
\dmn \big( \dLinFieldsRad{r} \big) = 4 r - 1 .
\end{align*}
A basis of $\dLinFieldsRad{r}$ is given by
$\dHolCurrMode{-k} \idField$ and $\dAntHolCurrMode{-k} \idField$, for $1 \le k \le 2r - 1$,
and $\dHolCurrMode{-2r} \idField -\dAntHolCurrMode{-2r} \idField$.
\end{lemma}
\begin{proof}
Recall that by Lemma~\ref{lem: dimension lower bounds in filtration}, we
have $\dmn (\dLinFieldsRad{r}) \ge 4 r - 1$ and the given elements 
are linearly independent in~$\dLinFieldsRad{r}$. It therefore
only remains to show that $\dmn (\dLinFieldsRad{r}) \le 4 r - 1$.
The spanning set given in Lemma~\ref{lem: spanning set in filtration}
is indexed by~$\DomCirc{r}$, which has $|\DomCirc{r}| = 4 r$ points;
one more than the asserted dimension.
However, applying Proposition~\ref{prop: canonical representatives} with $\zubis=0$
we see that
\begin{align*}
\field(0) + \Null \; = \; \sum_{\bu \in \DomCirc{r}} H_{\bu}(0) \, \field(\bu) + \Null .
\end{align*}
Since the left-hand-side $\field(0)$
is null by Example~\ref{ex: trivialer null field},
the right-hand-side gives one linear relation in the spanning set, provided that the
coefficients~$H_{\bu}(0)$ are not all zero.
Equation~\eqref{eq: harmonic measures sum up to one} shows that they are not.
\end{proof}

We are now ready to finish the proofs of the two main statements
of the section.

\begin{proof}[Proof of Theorem~\ref{thm: basis lin loc fields}]
We want to show that the fields
$\dHolCurrMode{-k} \idField$ and $\dAntHolCurrMode{-k} \idField$, for $k \in \Zpos$,
form a basis of the space~$\dLinFields$ of linear local fields.
We already observed, based on Corollary~\ref{cor: easy inclusion}, that these
fields are linearly independent. It remains to show that they span.
Using the filtration by radius, we see that any field $F \in \dLinFields$ belongs
to $\dLinFieldsRad{r}$ for some $r \in \Zpos$. Then by
Lemma~\ref{lem: dimension upper bounds in filtration},
$F$ is indeed a linear combination of the basis elements of
$\dLinFieldsRad{r}$, which are manifestly already linear combinations
of elements of the asserted basis~\eqref{eq: basis of lin loc fields} of~$\dLinFields$.
\end{proof}

\begin{proof}[Proof of Corollary~\ref{coro: linear nulls}]
By the filtrations by radii, it is again enough to focus on~$\dLocLinFiRad{r}$
for a fixed~$r \in \Zpos$.
The field polynomials $\field(\zu)$ with $\zu \in \Ball_\primary(r)$
are a basis of~$\dLocLinFiRad{r}$ by construction, so the dimension is
$\dmn \big( \dLocLinFiRad{r} \big) = |\Ball_\primary(r)| = 2r^2 + 2 r + 1$.
In Lemma~\ref{lem: dimension upper bounds in filtration} we found
that the dimension of the quotient
$\dLinFieldsRad{r} = \dLocLinFiRad{r} / \big( \dNuFi \cap \dLocLinFiRad{r} \big)$
equals~$4r-1$, so
the dimension of $\dNuFi \cap \dLocLinFiRad{r}$ is the difference
\begin{align*}
\dmn \Big( \dNuFi \cap \dLocLinFiRad{r} \Big) =  2r^2 + 2 r + 1 - (4 r - 1) 
= \; & 2 r^2 - 2 r + 2 \\
= \; & 2(r-1)^2 + 2(r-1) + 2 \\
= \; & |\DomBallInt{r}| + 1 \, .
\end{align*}
The known null fields
$\gLapl \field(\zu)$, for $\zu \in \DomBallInt{r}$, and $\field(0)$
are all linearly independent in $\dNuFi \cap \dLocLinFiRad{r}$
(by standard considerations with the discrete Laplacians)
and their number matches the above dimension of $\dNuFi \cap \dLocLinFiRad{r}$
calculated above, so they form a basis.
\end{proof}

\subsection{Homogeneous linear local fields}
\label{subsec: homogeneous linear local fields}

Recall that homogeneous local fields, i.e., eigenvectors of~$\dvirL{0}$ and~$\dvirBarL{0}$,
are important for the reason that their scaling limits are to be formed with definite renormalization
related to the eigenvalues. With our complete understanding of the space of linear local fields
from Theorem~\ref{thm: basis lin loc fields}, let us see the consequences about homogeneous
linear local fields.

In the same fashion as in~\eqref{eq: homogeneous local fields}, define
the \term{homogeneous linear local fields} of {conformal dimensions}
$\Delta,\barDelta\in\C$ as the subspace
\begin{align*}
\dHomLinFi{\Delta}{\bar\Delta} :=
\ker\Big(\, \dvirL{0}\big\vert_{\dLinFields}-\Delta\,\id_{\dLinFields} \,\Big)
\cap \ker\Big(\, \dvirBarL{0}\big\vert_{\dLinFields}-\barDelta\,\id_{\dLinFields} \,\Big)
\subset \dLinFields\,.
\end{align*}

\begin{prop}
	The Virasoro modes $\dvirL{0}$ and $\dvirBarL{0}$ are simultaneously diagonalizable on $\dLinFields$,
	i.e., we have the vector space direct sum
	$\dLinFields = \bigoplus_{\Delta, \bar{\Delta}} \dHomLinFi{\Delta}{\bar\Delta}$.
	Moreover, the nontrivial joint eigenspaces are one-dimensional; explicitly
	\begin{align*}
	\dHomLinFi{\Delta}{\bar\Delta} \, = \, \begin{cases} 
	\ \C \,\dHolCurrMode{-k} \idField &
	\mspace{30mu} \text{ if } \text(\Delta,\barDelta) = (k,0) \text{ with } k\in\Zpos \\
	\ \C \, \dAntHolCurrMode{-k} \idField &
	\mspace{30mu} \text{ if } (\Delta,\barDelta) = (0,k) \text{ with } k\in\Zpos \\
	\mspace{12mu}\ \set{0} &
	\mspace{30mu} \text{ otherwise .}
	\end{cases}
	\end{align*}
\end{prop}
\begin{proof}
	Theorem \ref{thm: basis lin loc fields} says that the fields $\dHolCurrMode{-k}\idField$ and $\dAntHolCurrMode{-k}\idField$ for $k\in\Zpos$ constitute a basis of $\dLinFields$,
	and for the basis vectors we have
	\begin{align*}
	\dvirL{0}
	\big(\dHolCurrMode{-k}\idField\big)
	=
	k\,\dHolCurrMode{-k}\idField\,,
	\mspace{20mu}
	\dvirBarL{0}
	\big(\dAntHolCurrMode{-k}\idField\big)
	=
	k\,\dAntHolCurrMode{-k}\idField
	\mspace{20mu}
	\text{ and }
	\mspace{20mu}
	\dvirBarL{0}
	\big(\dHolCurrMode{-k}\idField\big)
	=
	\dvirL{0}
	\big(\dAntHolCurrMode{-k}\idField\big)
	=
	0+\Null\,
	\end{align*}
	by the same calculation as in the Fock space~$\FullFock$.
\end{proof}

\begin{rmk}\label{rmk: currents and currents}
	Exactly two Virasoro primary states among the ones listed in~\eqref{eq: primary vacuum}~--~\eqref{eq: primary nine} correspond to
	linear local fields: the counterparts of
	\eqref{eq: primary holom current} and~\eqref{eq: primary antiholom current}
	are
	\begin{align*}
	& \dHolCurrMode{-1}\idField
	& \quad\text{ with a representative } \quad
	& \dRep{-1} = \phantom{-} \ii \, \left( \frac{1}{2} \, \frac{\field(1) - \field(-1)}{2}
	- \frac{\ii}{2} \, \frac{ \field(\ii) - \field(-\ii)}{2} \right) \\
	& \dAntHolCurrMode{-1}\idField
	& \quad\text{ with a representative } \quad
	& \dRepBar{-1} = -\ii \left( \frac{1}{2} \, \frac{\field(1) - \field(-1)}{2}
	+ \frac{\ii}{2} \, \frac{ \field(\ii) - \field(-\ii)}{2} \right),
	\end{align*}
	where the representatives where obtained from~\eqref{eq: representative for Jk}.
	Note that although \eqref{eq: primary holom current} and~\eqref{eq: primary antiholom current}
	are as CFT fields interpreted as the holomorphic and antiholomorphic currents,
	in our discrete current mode constructions~\eqref{eq: discrete current modes}
	we needed slightly different discretizations
	of the holomorphic and antiholomorphic currents.
	The representatives of~$\dHolCurrMode{-1}\idField$ and~$\dAntHolCurrMode{-1}\idField$
	above are, respectively, the averages of~$\holcurrfield(\zu_\medial)$ and~$\antiholcurrfield(\zu_\medial)$
	over the four midpoints~$\zu_\medial \in \set{\frac{\pm 1}{2}, \frac{\pm \ii}{2}}$ of edges
	adjacent to the origin.
	\hfill~$\diamond$
\end{rmk}

\section{Local fields of higher degree}
	\label{sec:higher}
	In the previous section, we fully characterized linear local fields of the DGFF, so the remaining task
for proving the isomorphism $\dFields \isom \FullFock$ is
to characterize local fields of higher degree. By definition~\eqref{eq: space of field polynomials},
the space $\dLocFi$ of field polynomials is a polynomial algebra, i.e., a symmetric tensor algebra,
which is naturally graded
\begin{align*}
\dLocFi \, = \, \bigoplus_{d \in \Znn} \dLocFiDeg{d}
\end{align*}
by the degrees of the polynomials, i.e., so that the degree~$d$ component is
the $d$th symmetric tensor power of
the space~\eqref{eq: space of linear field polynomials}
of linear field polynomials,
\begin{align*}
\dLocFiDeg{d}
     \, = \, \spn_\bC \set{ \prod_{i=1}^d \field(\zu_i)
            \; \bigg| \; \zu_1, \ldots, \zu_d \in \Z^2 } \,.
\end{align*}
The space~\eqref{eq: space of correlation equivalence classes of fields} of local fields
inherits a grading
\begin{align*}
\dFields \, = \, \bigoplus_{d \in \Znn} \dFieldsDeg{d}
\end{align*}
where $\dFieldsDeg{d} = \dLocFiDeg{d} / \dNuFi$ is the quotient of the space of degree~$d$
field polynomials by null fields. Recall, however, that while 
$\dLocFi = \bigoplus_{d \in \Znn} \dLocFiDeg{d}$ is a graded algebra, the space
$\dFields = \dLocFi / \dNuFi$ of local fields does not even inherit
a multiplication from the polynomials: the null fields do not form an ideal, let alone a graded ideal.
The space
$\dFields = \bigoplus_{d \in \Znn} \dFieldsDeg{d}$
of local fields is merely a
graded vector space.

In order to fully characterize higher degree fields and to finish the proof of
isomorphism $\dFields \isom \FullFock$, we will introduce normal ordering in~$\dFields$,
which partially rectifies the lack of an algebra structure and allows us to reduce the analysis
of higher degree local fields to that of the linear local fields.
Let us also note that the normal ordering will be
combinatorially similar with Wick products with respect to a Gaussian
measure, but it is not a Wick product for the DGFF measure in any discrete domain
(with any boundary conditions). Indeed, we are operating at the level of
abstract local fields, which can be evaluated in an arbitrary discrete domain and with
either Dirichlet or Neumann boundary conditions.
Following the idea from CFT that only a local
coordinate specification should affect the normal ordering procedure,
we use infinite square lattice quantities in the normal ordering contractions.

The normal ordering is defined in Section~\ref{sub: normal ordering}.
Besides essentially combinatorial lemmas, we prove
Lemma~\ref{lem: normal ordering with null fields}, which
expresses a crucial property that could be interpreted as null fields
being an ideal with respect to normal ordering.
In Section~\ref{sub: main result A} we are then ready to state and prove
our main result, Theorem~\ref{thm: main theorem about Fock space structure},
by combining this lemma with the full classification
of linear local fields from the previous section.
The proof also involves making a connection between the normal
ordering and the algebraic structure provided by the current modes, and we illustrate the
concrete computational nature of this procedure by giving a few
examples of higher degree fields in
Section~\ref{sub: examples of higher degree}.

\subsection{Normal ordering}\label{sub: normal ordering}

Let $\Green$ denote the Green's function on the infinite square lattice~$\Z^2$~\cite{LL-random_walk},
i.e., the unique function
\begin{align}\label{eq: square grid Green function}
\Green \colon \Z^2 \to \R
\quad \text{ satisfying } \quad
    \begin{cases}
    \gLapl \Green(\zu)=-\delta_0(\zu) \\
    \Green(0) = 0 \\
    \Green(\zu) = -\frac{1}{2\pi}\log |\zu| + C + \OO \big( |\zu|^{-2} \big)
	\text{ as } |\zu| \to \infty ,
    \end{cases}
\end{align}
where $C=-\frac{1}{2\pi}(\gamma+\frac{3}{2}\log 2)$ and $\gamma$ is the Euler--Mascheroni constant.
The infinite lattice Green's function is symmetric in the sense that $\Green(-\zu)=\Green(\zu)$
for any $\zu \in \Z^2$.

If $\Green(\zu)$ could be informally thought of as~``$\EX[\phi(\zu_0+\zu)\phi(\zu_0)]$''
(for any~$\zu_0 \in \Z^2$ by virtue of translation invariance),
then an informal counterpart two-point function appropriate for the
evaluations~$\field(\zu) \mapsto \phi(\zu) - \phi(0)$ would be
``$\EX\big[\big(\phi(\zu) - \phi(0)\big)\big(\phi(\zubis) - \phi(0)\big)\big]$''.
This motivates the definition of
\begin{align}\label{eq: full plane grad Green}
\GreenGrad \colon \Z^2 \times \Z^2 \to \R
\qquad \text{ as } \qquad
\GreenGrad(\zu,\zubis) := \Green(\zu-\zubis) - \Green(\zu) - \Green(\zubis) ,
\end{align}
which has the symmetricity property~$\GreenGrad(\zu, \zubis) = \GreenGrad(\zubis,\zu)$.
We use~$\GreenGrad$ in normal ordering contractions as follows.
We define \term{contractions} of linear field polynomials
\begin{align*}
\dLocLinFi \tens \dLocLinFi \to \; & \bC \\
L_1 \tens L_2 \mapsto \; & \wick{\c L_1 \c L_2}
\end{align*}
by bilinear extension from the contractions of the linear basis monomials given by the formula
\begin{align*}
\wick{\c \field(\zu) \c \field(\zubis)}
\coloneqq
4\pi\GreenGrad(\zu,\zubis)\,
\qquad \text {for $\zu,\zubis\in\Z^2$.}
\end{align*}
The contractions are symmetric,
$\wick{\c L_1 \c L_2} = \wick{\c L_2 \c L_1}$, and thus well-defined on~$\dLocFiDeg{2}$.
\begin{ex}\label{ex: Wick cont of Lap}
For $\zu,\zubis\in\Z^2$, we have $\wick{\big(\gLapl\c\field(\zu) \big)\c\field(\zubis)}
  = -4 \pi \delta_{\zubis}(\zu) + 4 \pi \delta_0(\zu)$.
\hfill
$\diamond$
\end{ex}

The \term{normal ordering} is a linear operation
\begin{align*}
\no{\cdots} \colon \dLocFi \to \dLocFi
\end{align*}
on the space $\dLocFi = \bigoplus_{d \in \Znn} \dLocFiDeg{d}$ of field polynomials,
defined on the degree~$d$ subspace~$\dLocFiDeg{d}$ by
\begin{align}\label{eq: normal ordering def}
\no{ \prod_{i=1}^d L_i } :=
	\sum_{P' \in \ParPair(d)} (-1)^{|P'|} \prod_{j \notin \bigcup P'} L_j \prod_{\{k,\ell\} \in P'} \wick{\c L_k \c L_\ell} ,
\end{align}
where $\ParPair(d)$ denotes the set of \term{partial pairings}
of the $d$ indices,
i.e., collections~$P'$ of disjoint two-element subsets of $\set{1,\ldots,d}$. One
can alternatively view~\eqref{eq: normal ordering def} as a sum of partitions~$\varpi$ of
the index set~$\set{1,\ldots,d}$
to subsets of size at most two; then the number of pairs reads $|P'| = d-|\varpi|$, and the first
and second products above range over the parts of size one and two, respectively.

Let us record two combinatorial formulas for later calculations.
\begin{lemma}\label{lem: explicit triangularity of normal ordering}
\label{lemma: NormOrd_Wick}
Let $L_1 , \ldots , L_d \in \dLocLinFi$ be linear field polynomials. Then we have
\begin{align}\label{eq: explicit triangularity of normal ordering}
\no{ \prod_{i=1}^d L_i }
\, = \; & \prod_{i=1}^d L_i + R
\qquad \text{where} \qquad
R \, \in \, \bigoplus_{d' \le d-2} \dLocFiDeg{{d'}}
\end{align}
and
\begin{align}\label{eq: partial normal ordered product}
L_1 \,\no{ \prod_{i=2}^d L_i }
\, = \; & \no{ L_1 \, \prod_{i=2}^d L_i } \, + \;
 	\sum_{i=2}^d \wick{\c L_1 \c L_i} \;
 	\no{\prod_{\substack{2 \le j \le d \\ j\neq i\,}} L_j}
\end{align}
In particular, \eqref{eq: explicit triangularity of normal ordering}
shows that $\no{\cdots} \colon \dLocFi \to \dLocFi$ is bijective.
\end{lemma}
\begin{proof}
Formula~\eqref{eq: explicit triangularity of normal ordering} is
obtained by separating in~\eqref{eq: normal ordering def} the
summand corresponding to the empty partial pairing~$P' = \emptyset$,
which gives the term $\prod_{i=1}^d L_i$, and the rest,
which include at least one pair and therefore give terms of degrees $d' = d - 2|P'| \le d-2$.
This ``upper triangularity''
formula
allows one to uniquely solve $\prod_{i=1}^d L_i$ in terms of normally ordered products,
inductively on the degree~$d$, so it shows bijectivity of $\no{\cdots} \colon \dLocFi \to \dLocFi$.

Formula~\eqref{eq: partial normal ordered product} is obtained by
grouping the terms in~$\no{ \prod_{i=1}^d L_i }$ to those where the index~$1$ has no pair
in the partial pairing~$P'$ (these produce the left hand side) and those where
$1$ has a pair~$i$
(these produce, up to a sign, the $i$th summand on the right hand side).
\end{proof}

\begin{ex}\label{ex: base case NormOrd null}
For $\zu,\zubis\in\Z^2$, with the formula from Example~\ref{ex: Wick cont of Lap}, we get
\begin{align*}
\no{\big(\gLapl\field(\zu) \big)\field(\zubis)}
= & \ \big(\gLapl \field(\zu) \big) \field(\zubis) - \wick{\big(\gLapl\c\field(\zu) \big)\c\field(\zubis)}
	\\
= & \ \big(\gLapl \field(\zu) \big) \field(\zubis) + 4 \pi \delta_{\zubis}(\zu) - 4 \pi \delta_0(\zu) \,.
\end{align*}
Observe that this field polynomial is null by Example~\ref{ex: quadratic nulls}.
\hfill
$\diamond$
\end{ex}

The observation in the example above generalizes.

\begin{lemma}\label{lem: normal ordering with null fields}
Let $L_1, \ldots,L_{d} \in \dLocLinFi$ be linear field polynomials.
If $L_j \in \dNuFi$ for some $j$, then also
$\no{ \prod_{i=1}^{d} L_i } \in \dNuFi $.
\end{lemma}

We give the proof shortly, but we first note the important consequence.
Lemma~\ref{lem: normal ordering with null fields} implies that the normally ordered product
of $d$ linear field polynomials
${\no{\cdots} \colon (\dLocLinFi)^{\tens d} \to \dLocFi}$
factorizes through the quotient by null fields in each tensorand, giving rise to
well-defined normally ordered products
$\noQuo{\cdots} \colon (\dLinFields)^{\tens d} \to \dFields$
of $d$ local fields by
\begin{align}\label{eq: quotient normal ordering definition}
\noQuo{(L_1 + \Null) \tens \cdots \tens (L_d + \Null)} := \no{\prod_{i=1}^d L_i} + \Null
\end{align}
(we write the tensor products explicitly on the left hand side
in order to avoid a false impression of a ring structure on $\dFields$).
Extending linearly to all degrees~$d$ and noting symmetricity in the arguments,
we get a normal ordering operation
\begin{align*}
\noQuo{\cdots} \colon \dFields \to \dFields
\end{align*}
on the space of local fields of the DGFF.
Moreover, it follows from
Lemma~\ref{lem: explicit triangularity of normal ordering} that this
normal ordering is surjective.

\begin{proof}[Proof of Lemma~\ref{lem: normal ordering with null fields}]
Without loss of generality, assume $L_1 \in \dNuFi$.
By definition of $\dNuFi$, we must check the nullity condition~\eqref{eq: null field condition} for $F = \no{ \prod_{i=1}^{d} L_i }$.
Let $\ddomain \subset \SqLatMesh$ be a discrete domain,
$(f_a)_{a \in A}$ a finite collection of zero-average test functions $f_a \in \ddistribMC{\ddomain}$,
and $\zs \in \ddomain$ a point, such that the supports are meaningful and distinct as required in the
condition~\eqref{eq: null field condition}. We must show the vanishing of
\begin{align*}
& \EX_{\ddomain}^{\DorN} \Big[ \Big( \ev^{\ddomain}_\zs \no{ \prod_{i=1}^{d} L_i } \Big)
		\, \prod_{a \in A} \dgff{f_a} \Big] \\
= \; & \sum_{P'} (-1)^{|P'|} \; \bigg( \prod_{\{k,\ell\} \in P'}  \wick{\c L_k \c L_\ell} \bigg)
	\; \EX_{\ddomain}^{\DorN} \Big[
        \prod_{j \notin P'} \big( \ev^{\ddomain}_\zs L_j \big) \, \prod_{a \in A} \dgff{f_a} \Big] \, .
\end{align*}
Use Wick's formula to calculate the expected value on the second line as a sum over
pairings~$P$ of $A \, \sqcup \, (\{1,\ldots,d\}\setminus P')$, schematically (avoiding cumbersome
notation that would be needed to explicitly spell out three possible forms of the pairs in~$P$)
\begin{align*}
& \EX_{\ddomain}^{\DorN} \Big[ \Big( \ev^{\ddomain}_\zs \no{ \prod_{i=1}^{d} L_i } \Big)
		\, \prod_{a \in A} \dgff{f_a} \Big] \\
= \; & \sum_{P'}  (-1)^{|P'|} \; \bigg( \prod_{\{k,\ell\} \in P'}  \wick{\c L_k \c L_\ell} \bigg)
	\sum_{P} \bigg( \prod_{{\textrm{pairs} \in P}} \EX_{\ddomain}^{\DorN} \Big[ \textrm{(pair product)} \Big] \bigg).
\end{align*}
Note that for each~$(P',P)$,
the index~$1$ of the linear null field~$L_1$ appears either in one pair in~$P'$ (leading to a normal ordering
contraction factor) or in one of the pairs in~$P$ (leading to a suitable DGFF two-point function factor),
and it is in fact useful to consider three separate cases of which exactly one occurs:
\begin{itemize}
\item $\set{1,a} \in P$ for some $a \in A$
\item $\set{1,m} \in P$ for some $m \in \set{2,\ldots,d}$
\item $\set{1,m} \in P'$ for some $m \in \set{2,\ldots,d}$.
\end{itemize}
In the first case, there is a factor
\begin{align*}
\EX_{\ddomain}^{\DorN} \Big[ \big( \ev^{\ddomain}_\zs L_1 \big) \, \dgff{f_a} \Big] ,
\end{align*}
which vanishes since $L_1$ is a null field and the supports of $L_1$ and $f_a$ satisfy the disjointness
and well-definedness conditions. Therefore only the last two cases contribute to our quantity of interest.
Moreover, the contributions of those can be naturally matched pairwise.
Namely, fixing $m \in \set{2,\ldots,d}$, in the second
case we may remove the pair $\set{1,m}$ from $P$ to get a pairing~$\hat{P}$ of
$(\set{1,\ldots,d}\setminus \set{1,m}) \sqcup A$, and in the third case
we may remove the pair $\set{1,m}$ from $P'$ to get a partial pairing~$\hat{P}'$ of
$\set{1,\ldots,d}\setminus \set{1,m}$, and then the sum of all contributions with the fixed $m$
combine to take the form
\begin{align*}
& \sum_{\hat{P}'} \sum_{\hat{P}} \bigg( (-1)^{|\hat{P}'| + 1} \wick{\c L_1 \c L_m} + (-1)^{|\hat{P}'|} \;
	\EX_{\ddomain}^{\DorN} \Big[ \big( \ev^{\ddomain}_\zs L_1 \big) \, \big( \ev^{\ddomain}_\zs L_m \big) \Big] \bigg)
	A(\hat{P}') \, B(\hat{P}) \; ,
\end{align*}
where $A(\hat{P}')$ and $B(\hat{P})$ are products of common contraction factors and common DGFF two-point function
factors, respectively (writing them explicitly is possible but will not be needed). We will show the vanishing of
the first factor, which up to a sign equals
\begin{align}\label{eq: the key vanishing factor for higher order null fields}
- \wick{\c L_1 \c L_m}
+ \EX_{\ddomain}^{\DorN} \Big[ \big( \ev^{\ddomain}_\zs L_1 \big) \, \big( \ev^{\ddomain}_\zs L_m \big) \Big] \; ,
\end{align}
and this will conclude the proof.

The key to prove the vanishing of~\eqref{eq: the key vanishing factor for higher order null fields}
is the complete classification of linear null fields from the previous section, Corollary~\ref{coro: linear nulls}.
The classification shows that $L_1 \in \dLocLinFi \cap \dNuFi$ is a linear combination of field
polynomials of the forms $\gLapl \field(\zu)$, for $\zu \in \Z^2$, and $\field(0)$. By linearity in~$L_1$,
we may assume~$L_1$ is exactly one of these. The latter case is trivial, so we will
assume $L_1 = \gLapl \field(\zu)$ for some $\zu \in \Z^2$.
By linearity in~$L_m$ we may also assume~$L_m = \field(\zubis)$ for some $\zubis \in \Z^2$.
In this case~\eqref{eq: the key vanishing factor for higher order null fields} reads
\begin{align*}
& - \wick{\big( \gLapl \c \field(\zu) \big) \c \field(\zubis)}
  + \EX_{\ddomain}^{\DorN} \Big[ \big( \gLapl \dgffptwise(\zs + \meshsz \zu) \big) \,
		\big( \dgffptwise(\zs + \meshsz \zubis) - \dgffptwise(\zs) \big) \Big] \; .
\end{align*}
The first term here was calculated in Example~\ref{ex: Wick cont of Lap}; it
is~$4 \pi \delta_{\zubis}(\zu) - 4 \pi \delta_{0}(\zu)$.
Directly by the defining covariance of the DGFF, the second term equals
\begin{align*}
\gLapl \Green_{\ddomain}^{\DorN}(\zsbis , \zs + \meshsz \zubis) \big|_{\zsbis = \zs + \meshsz \zu}
- \gLapl \Green_{\ddomain}^{\DorN}(\zsbis , \zs) \big|_{\zsbis = \zs + \meshsz \zu} .
\end{align*}
Since $\zs + \meshsz \zu \in \ddomain \setminus \bdry \ddomain$ by the
requirements~\eqref{eq: null field condition} of supports, this
second term simplifies to $- 4 \pi \delta_{\zubis}(\zu) + 4 \pi \delta_{0}(\zu)$.
Therefore it cancels the first term, completing the proof.
\end{proof}

Let us emphasize again that Lemma~\ref{lem: normal ordering with null fields} is key to us,
since it allows us to define normal ordering on the level of local fields of the DGFF, not just
field polynomials. The key ingredients enabling its proof, in turn, were the full classification of
linear null fields from Section~\ref{sec:linear}, and the fact that
the infinite lattice Green's functions used in our domain-independent
normal orderings had the exact same form of local singularities as the DGFF two-point functions in
any discrete domain.

\subsection{Local fields form a Fock space}\label{sub: main result A}

In this section we will prove that the space~$\dFields$ of local fields of
the DGFF is a Fock space. The local fields
$\dAntHolCurrMode{-k'_{m'}} \cdots \dAntHolCurrMode{-k'_1}
	\dHolCurrMode{-k_m} \cdots \dHolCurrMode{-k_1}\idField$
obtained by repeated current mode actions on the identity field
correspond to the Fock space basis vectors. In order to show that
these fields span~$\dFields$, we first want to express them in terms of
explicit normal ordered products of linear local fields.

\begin{lemma}\label{lem: basis elements as normal ordered products}
For any $k_1,\ldots,k_m, k'_1,\ldots, k'_{m'}\in\Zpos$,
we have the following equality in~$\dFields$:
\begin{align*}
\noQuo{\dAntHolCurrMode{- k'_{m'}}\idField \tens \cdots \tens \dAntHolCurrMode{- k'_1}\idField
	\tens \dHolCurrMode{-k_n}\idField \tens \cdots \tens \dHolCurrMode{-k_1}\idField}
= \dAntHolCurrMode{-k'_{m'}} \cdots \dAntHolCurrMode{-k'_1}
	\dHolCurrMode{-k_m} \cdots \dHolCurrMode{-k_1}\idField \, .
\end{align*}
\end{lemma}

This expression is the final ingredient of the proof of our first main result, so before giving
the proof of the lemma, let us show how we conclude with it.

\begin{thm}\label{thm: main theorem about Fock space structure}
The space~$\dFields$ of local fields of the discrete Gaussian free field
and the Fock space~$\FullFock$ are isomorphic,
\begin{align*}
\dFields \; \isom \; \FullFock \; ,
\end{align*}
as representations of two commuting Heisenberg algebras, and as representations of
two commuting Virasoro algebras.
\end{thm}
\begin{proof}
On both $\dFields$ and~$\FullFock$, the Virasoro representations are obtained from the
Heisenberg representations by the Sugawara construction, so it suffices to prove isomorphism
as representations of Heisenberg algebras.
In Corollary~\ref{cor: easy inclusion} we already established
$\FullFock \, \isom \, (\UHei \tens \AntUHei) \idField \, \subset \, \dFields$
as representations of
Heisenberg algebras, so to complete the proof we must
show $\dFields \subset (\UHei \tens \AntUHei) \idField$.

So let $F + \Null \in \dFields$ be an arbitrary local field, with $F \in \dLocFi$ denoting
a representative field polynomial. By the surjectivity of the normal ordering
that was observed in Lemma~\ref{lem: explicit triangularity of normal ordering},
we can find some $\tilde{F} \in \dLocFi$
such that $\no{\tilde{F}} = F$.
Split the field polynomial $\tilde{F}$
to homogeneous pieces, i.e., write $\tilde{F} = \sum_{d=0}^D \tilde{F}_d$ with
$\tilde{F}_d \in \dLocFiDeg{d}$. By linearity, it suffices to show
that~$\no{\tilde{F}_d} + \Null \in (\UHei \tens \AntUHei) \idField$ for each~$d$.
Factorize the homogeneous field polynomial $\tilde{F}_d$ into linear factors, i.e.,
write  $\tilde{F}_d = \prod_{i=1}^d L_i$ for some $L_1 , \ldots , L_d \in \dLocLinFi$,
and note that by definition~\eqref{eq: quotient normal ordering definition} we then have
\begin{align*}
\no{\tilde{F_d}} + \Null
\; = \; \no{L_1 \cdots L_d } + \Null
\; = \; \noQuo{(L_1 + \Null) \tens \cdots \tens (L_d + \Null)} .
\end{align*}
Theorem~\ref{thm: basis lin loc fields} gives a basis for linear local fields, and
it in particular guarantees that for each $i = 1, \ldots, d$, we can write $L_i + \Null$
as a linear combination of $\dHolCurrMode{-k} \idField$ and
$\dAntHolCurrMode{-k} \idField$, $k \in \Zpos$.
Lemma~\ref{lem: basis elements as normal ordered products} explicitly shows that
normal ordering applied to such linear factors is expressible in terms of
local fields of the form $\dAntHolCurrMode{-k'_{m'}} \cdots \dAntHolCurrMode{-k'_1}
	\dHolCurrMode{-k_m} \cdots \dHolCurrMode{-k_1}\idField$,
which are manifestly in $(\UHei \tens \AntUHei) \idField$.
We conclude that $F + \Null \in (\UHei \tens \AntUHei) \idField$ as desired.
\end{proof}

For the proof of Lemma~\ref{lem: basis elements as normal ordered products},
we still need one auxiliary result. In it, and in the proof of
Lemma~\ref{lem: basis elements as normal ordered products} itself,
we will use the explicit expressions given in Example~\ref{ex: id primary}
for the field polynomials $\dRep{-k}, \dRepBar{-k} \in \dLocLinFi$
which are representatives of the linear local fields
${\dHolCurrMode{-k} \idField, \dAntHolCurrMode{-k} \idField \in \dLinFields}$,
respectively. Recall that the explicit expressions are
\begin{align*}
\dRep{-k} := \frac{\ii}{2\pi}
	\sum_{\zu_\medial\in\Poles{-k}_\medial}\gdeebar \zu_\medial^{[-k]}\pdee\field(\zu_\medial)
\qquad \text{ and } \qquad
\dRepBar{-k} := \frac{-\ii}{2\pi}
	\sum_{\zu_\medial\in\Poles{-k}_\medial}\gdee \cconj{\zu}_\medial^{[-k]}\pdeebar\field(\zu_\medial) \; ,
\end{align*}
where as in~\eqref{eq: monomial pole support set},
we denote by $\Poles{-k}\subset\ZDiamond\cup\ZMedial$ the neighborhood
of the origin where the function $\zu \mapsto \zu^{[-k]}$ is not discrete holomorphic,
and $\Poles{-k}_\medial = \Poles{-k} \cap \ZMedial$.

\begin{lemma}\label{lemma: technical}
For any $k,\ell\in\Zpos$ and any positively oriented corner contour $\gamma$ satisfying
$\Poles{k}\cup\Poles{\ell} \subset \interior\gamma$, we have
\begin{align*}
\dcoint{\gamma} \zu_\diamond^{[-k]} \wick{ \big( \pdee\c\field(\zu_\medial) \big) \c {\dRep{-\ell}} } \,\dd{\zu}
\,=\, \dcoint{\gamma} \zu_\diamond^{[-k]} \wick{ \big( \pdee \c \field(\zu_\medial) \big) \c {\dRepBar{-\ell}} } \,\dd{\zu}
\,=\, & \ 0 \,, \\
\dcoint{\gamma} \cconj{\zu}_\diamond^{[-k]} \wick{ \big(\pdeebar \c \field(\zu_\medial) \big) \c {\dRep{-\ell}} } \,\dd{\cconj \zu}
\,=\,  \dcoint{\gamma} \cconj{\zu}_\diamond^{[-k]}
		\wick{ \big( \pdeebar\c \field(\zu_\medial) \big) \c { \dRepBar{-\ell}}} \,\dd{\cconj \zu}
\,=\, & \ 0\,.
\end{align*}
\end{lemma}
\begin{proof}
The proofs of vanishing of all four integrals are similar, so let us just look at the top left one.
Using the explicit expression of $\dRep{-\ell}$ recalled above and the definition of contractions, we find
\begin{align*}
\dcoint{\gamma} \zu_\diamond^{[-k]} \wick{ \big( \pdee \c \field(\zu_\medial) \big) \c {\dRep{-\ell}} } \,\dd{\zu}
\, =& \, \frac{\ii}{2\pi} \sum_{\zubis_\medial\in\Poles{\ell}_\medial}\gdeebar \zubis_\medial^{[-\ell]}
	\dcoint{\gamma} \zu_\diamond^{[-k]} \wick{ \big( \pdee\c\field(\zu_\medial) \big) \big( \pdee\c\field(\zubis_\medial) \big) } \,\dd{\zu} \\
= & \, \frac{\ii}{2\pi} \sum_{\zubis_\medial\in\Poles{\ell}_\medial}\gdeebar \zubis_\medial^{[-\ell]}
	\dcoint{\gamma} \zu_\diamond^{[-k]} (\pdee)_{\zu}(\pdee)_\zubis \GreenGrad(\zu_\medial , \zubis_\medial) \,\dd{\zu} \,,
\end{align*}
Note moreover, that the function $\zu_\medial \mapsto(\pdee)_{\zu}(\pdee)_\zubis
\GreenGrad(\zu_\medial,\zubis_\medial)$ is discrete holomorphic away from the
neighbors of $\zubis_\medial$
by the factorization~\eqref{eq: factorizations of Laplacian with primal graph derivatives} of the Laplacian
and the discrete harmonicity of $\GreenGrad(\cdot, \zubis)$ away from~$0$ and~$\zubis$.
The integrand in the last discrete contour integral has the asymptotic
behaviour $\OO (|\zu|^{-k-2})$ as $|\zu| \to \infty$.
By discrete holomorphicity of the integrand, the contour~$\gamma$ can be taken arbitrarily
large without changing the result. On a square contour at distance~$R$ from the origin,
the integral is $\OO(R^{-k-1})$, so taking~$R \to \infty$ shows that it vanishes for~$k > 0$.
\end{proof}

We are then ready to provide the proof that we postponed earlier.

\begin{proof}[Proof of Lemma~\ref{lem: basis elements as normal ordered products}]
We use the representative field polynomials
$\dRep{-k}, \dRepBar{-k} \in \dLocLinFi$ given in Example~\ref{ex: id primary} for the linear local fields
$\dHolCurrMode{-k}\idField = \dRep{-k} + \Null$ and $\dAntHolCurrMode{-k}\idField = \dRepBar{-k} + \Null$.
In particular, the left hand side of the asserted equality becomes
\begin{align*}
\noQuo{\dAntHolCurrMode{- k'_{m'}}\idField \tens \cdots \tens \dAntHolCurrMode{- k'_1}\idField
	\tens \dHolCurrMode{-k_m}\idField \tens \cdots \tens \dHolCurrMode{-k_1}\idField}
\; = \;
	\no{\prod_{i=1}^m \dRep{-k_{i}} \, \prod_{i'=1}^{m'} \dRepBar{-k'_{i'}} }
	\,+\, \Null\, .
\end{align*}
To compare this with the right hand side in the assertion,
\begin{align*}
\dAntHolCurrMode{-k'_{m'}} \cdots \dAntHolCurrMode{-k'_1}
	\dHolCurrMode{-k_m} \cdots \dHolCurrMode{-k_1}\idField \, ,
\end{align*}
we argue by induction on the degree~${m+m'}$. In order to simplify notation, we spell out the details only in
the case of~$m'=0$; the proof of the general case would just consist of splitting the
induction step into two similarly handled cases and would involve
decorating some of the symbols with bars and primes
($\dAntHolCurrMode{-k'_{i'}}$, $\dRepBar{-k'_{i'}}$, $\pdeebar$, etc.).

Focusing on $m'=0$, we seek to prove the equality of
\begin{align*}
\no{ \prod_{i=1}^m \dRep{-k_{i}} } + \Null
\qquad \text{ and } \qquad
\dHolCurrMode{-k_m} \cdots \dHolCurrMode{-k_1}\idField \, .
\end{align*}
Normal ordering does nothing to polynomials of degree at most one
(there are no nontrivial partial pairings of a set with fewer than two elements),
so the cases $m=0$ and $m=1$ are clear:
for $m=0$ the equality is just the definition of the
identity local field, and for $m=1$ it is our preferred representative choice
$\dRep{-k_1} + \Null = \dHolCurrMode{-k_1}\idField$.

For brevity, from here on, denote
\begin{align*}
F = \prod_{1 \le i \le m} \dRep{-k_{i}}
\qquad \text{ and } \qquad
\widehat{F}_j = \prod_{\substack{1 \le i \le m \\ i \ne j}} \dRep{-k_{i}} \; \text{ for } j = 1 , \ldots, m .
\end{align*}
Inductively, we assume
$
\no{ F } + \Null = \dHolCurrMode{-k_m} \cdots \dHolCurrMode{-k_1}\idField
$, 
and we must prove the equality of
\[ {\no{ F \dRep{-k_{m+1}} } + \Null
\qquad \textrm{ and } \qquad
\dHolCurrMode{-k_{m+1}} \dHolCurrMode{-k_m} \cdots \dHolCurrMode{-k_1}\idField}
\; = \; \dHolCurrMode{-k_{m+1}} \Big( \no{ F } + \Null \Big) . \]
We will use~\eqref{eq: partial normal ordered product} to rewrite both sides,
and we then compare the results.
On the left hand side, rewriting directly gives
\begin{align}\label{eq: lhs of the normal ordering ingredient without decorations}
\no{F \dRep{-k_{m+1}}}
\; = \; \dRep{-k_{m+1}} \, \no{ F }
	-\sum_{j=1}^m \wick{ \c {\dRep{-k_{m+1}}} \c {\dRep{-k_j}} } \, \no{\widehat{F}_j} \,.
\end{align}
On the right hand side we first take a suitable corner contour $\gamma$ for the action
of $\dHolCurrMode{-k_{n+1}}$ and unravel the definition of the action of the current mode
\begin{align*}
\dHolCurrMode{-k_{m+1}} \big( \no{ F } + \Null \big)
= \, \frac{1}{2\pi} \dcoint{\gamma} \zu_\diamond^{[-k_{m+1}]}
	\pdee\field(\zu_\medial)\, \no{ F } \, + \, \Null \, .
\end{align*}
Calculating the discrete contour integral using discrete Stokes' formula~\eqref{eq: discrete Stokes}
and recalling the explicit expression for $\dRep{-k_{m+1}}$, we get
\begin{align*}
\dHolCurrMode{-k_{m+1}} \big( \no{ F } + \Null \big)
= \; & \bigg( \dRep{-k_{m+1}} + \frac{\ii}{2\pi} \sum_{\zu_\diamond\in\interior\gamma}
		\zu_\diamond^{[-k_{n+1}]} \gdeebar^{}\pdee\field(\zu_\diamond) \bigg) \no{ F }
	\,+\, \Null .
\end{align*}
In the second term, we then recall the
factorization~\eqref{eq: factorizations of Laplacian with primal graph derivatives}
of the discrete Laplacian and
rewrite using~\eqref{eq: partial normal ordered product} to see that
\begin{align*}
& \frac{\ii}{2\pi} \sum_{\zu_\diamond\in\interior\gamma}
		\zu_\diamond^{[-k_{n+1}]} \gdeebar^{}\pdee\field(\zu_\diamond) \no{ F }
		\\
= \; &
	\frac{\ii}{4\pi} \sum_{\zu_\diamond\in\interior_\primary\gamma} \zu_\diamond^{[-k_{m+1}]}
	\bigg( \no{ \gLapl\field(\zu_\diamond) F }
		+ \sum_{j=1}^m \wick{ \big(\gLapl\c\field(\zu_\diamond)\big)  \c {\dRep{-k_j}} } \, \no{ \widehat{F}_j } \bigg) \, .
\end{align*}
The linear factor~$\gLapl\field(\zu_\diamond)$ is null in the fully normal ordered term, so
by Lemma~\ref{lem: normal ordering with null fields} the corresponding normal ordered product is also null.
Combining, we have found that
\begin{align}\nonumber
& \dHolCurrMode{-k_{m+1}} \big( \no{ F } + \Null \big) \\
\label{eq: rhs of the normal ordering ingredient without decorations}
= \; & \dRep{-k_{m+1}}\, \no{ F }
	+ \frac{\ii}{4\pi} \sum_{j=1}^m \sum_{\zu_\diamond\in\interior_\primary\gamma} \zu_\diamond^{[-k_{m+1}]}
		\wick{ \big(\gLapl\c\field(\zu_\diamond)\big)  \c {\dRep{-k_j}} } \, \no{ \widehat{F}_j }
		\, + \, \Null \, .
\end{align}
Comparing \eqref{eq: lhs of the normal ordering ingredient without decorations}
and~\eqref{eq: rhs of the normal ordering ingredient without decorations},
we see that it suffices to prove that for each~$j=1,\ldots,m$, the difference
\begin{align*}
\wick{ \c {\dRep{-k_{m+1}}} \c {\dRep{-k_j}} }
    + \frac{\ii}{4\pi} \sum_{\zu_\diamond\in\interior_\primary\gamma} \zu_\diamond^{[-k_{m+1}]}
		\wick{ \big(\gLapl\c\field(\zu_\diamond)\big)  \c {\dRep{-k_j}} }
\end{align*}
of the explicitly written terms is null.
Again writing out the definition of the representative~$\dRep{-k_{m+1}}$
and using the same factorization of the discrete Laplacian and Stokes' formula,
this difference becomes
\begin{align*}
& \frac{\ii}{2\pi} \sum_{\zu_\medial\in\Poles{k_{m+1}}_\medial} \gdeebar \zu_\medial^{[-k_{m+1}]}
		\wick{ \c {\pdee\field(\zu_\medial)} \c {\dRep{-k_j}} }
    + \frac{\ii}{2\pi} \sum_{\zu_\diamond\in\interior_\primary\gamma} \zu_\diamond^{[-k_{m+1}]}
		\wick{ \big(\gdeebar^{}\pdee\c\field(\zu_\diamond)\big)  \c {\dRep{-k_j}} } \\
= \; & \frac{1}{2\pi} \dcoint{\gamma} \zu_\medial^{[-k_{m+1}]}
		\wick{ \big( \pdee \c  \field(\zu_\medial) \big) \c {\dRep{-k_j}} } \, \dd \zu \, .
\end{align*}
This discrete contour integral indeed vanishes by Lemma~\ref{lemma: technical}, and
thus the proof of the induction step is complete, too.
\end{proof}

\subsection{Concrete examples of higher degree local fields}
\label{sub: examples of higher degree}

Lemma~\ref{lem: basis elements as normal ordered products} was necessary for our proof
of Theorem~\ref{thm: main theorem about Fock space structure}, but it can also
be applied directly to give concrete representatives of local fields.
Let us give two examples.

In our first example, we elaborate
on the details of the discretization of the gradient squared of the DGFF which
was discussed in the intruduction.
For this, we start by admitting an observation about our conformal field theory
(this fact would properly make sense only in Section~\ref{sec:scal_lim}):
the CFT field $\frac{1}{4} \, \|\nabla \gff\|^2$, suitably regularized,
is the Virasoro primary with conformal weights~$(1,1)$ seen
in~\eqref{eq: quadratic primary}, i.e.,
$\HeiJ{-1} \AntHeiJ{-1} \FockId$.
The local field of the DGFF corresponding to
it is
\begin{align*}
\dHolCurrMode{-1}\dAntHolCurrMode{-1} \idField
\; = \; \noQuo{\dHolCurrMode{-1} \idField \tens \dAntHolCurrMode{-1} \idField }
\; = \; \no{\dRep{-1} \dRepBar{-1} } + \Null \, .
\end{align*}
A concrete representative for this quadratic primary field is therefore
obtained with formulas of~Example~\ref{ex: id primary} and some simplification
\begin{align*}
\no{\dRep{-1} \dRepBar{-1} }
= \; & \dRep{-1} \dRepBar{-1} - \wick{ \c {\dRep{-1}} \c {\dRepBar{-1}} } \\
= \; & \frac{1}{4} \, \bigg( \Big( \frac{\field(1) - \field(-1)}{2} \Big)^2 + \Big( \frac{\field(\ii) - \field(-\ii)}{2} \Big)^2 \bigg)
    - (\pi-2) \,.
\end{align*}
This representative of the discrete local field
contains a clear discrete analogue of~$\frac{1}{4} \, \|\nabla \gff\|^2$ and an
explicit additive constant that a priori comes from normal ordering.
The constant is exactly what is needed to make the corresponding random variable
(obtained by evaluation) have asymptotically zero expected value in the scaling limit:
the two expected squares of gaussians are cancelled by the negative additive constant.
This subtraction is \emph{necessary} before it is meaningful to renormalize
by the diverging prefactor~$\meshsz^{-2}$, which corresponds to the scaling dimension~$2$ of this local field.

Giving concrete representives also for the other Virasoro primary fields
\eqref{eq: primary vacuum}~--~\eqref{eq: primary nine} etc.
would be possible, but the expressions become long and not as enlightening.

As the other explicit example, we choose the holomorphic stress tensor~$T = \VirL{-2} \FockId$ of the CFT,
which in terms of Heisenberg generators reads~$T = \frac{1}{2} \, \HeiJ{-1} \HeiJ{-1} \FockId$.
The local field of the DGFF corresponding to
the holomorphic stress tensor is
\begin{align*}
\frac{1}{2} \, \dHolCurrMode{-1}\dHolCurrMode{-1} \idField
\; = \; \frac{1}{2} \,\noQuo{\dHolCurrMode{-1} \idField \tens \dHolCurrMode{-1} \idField }
\; = \; \frac{1}{2} \, \no{\dRep{-1} \dRep{-1} } + \Null \, .
\end{align*}
A concrete representative is
\begin{align*}
\frac{1}{2} \, \no{\dRep{-1} \dRep{-1} }
= \; & \frac{1}{2} \, \dRep{-1} \dRep{-1} - \frac{1}{2} \, \wick{ \c {\dRep{-1}} \c {\dRep{-1}} } \\
= \; & -\frac{1}{2} \left( \frac{1}{2} \, \frac{\field(1) - \field(-1)}{2} - \frac{\ii}{2} \, \frac{ \field(\ii) - \field(-\ii)}{2} \right)^2 \,.
\end{align*}

\section{Scaling limits of local field correlations}
    \label{sec:scal_lim}
    The conformal field theory (CFT) of interest to us is known in physics as the massless free boson.
The constructive quantum field theory approach to it amounts to studying the continuum
probabilistic model called the Gaussian Free Field~(GFF)
\cite{Gawedzki-lectures_on_CFT, KM-GFF_and_CFT, PW-lecture_notes_on_GFF}.
Analogously to our conventions for the discrete GFF, we will consider the
gradient of the GFF with either Dirichlet or Neumann boundary conditions in
simply connected planar domains of general shape.
Changing the domain and/or boundary conditions results in a different
probablistic object.\footnote{In the continuum, conformal invariance relates the free
fields with the same boundary conditions in different simply connected domains, but nevertheless
defining and viewing these as different probabilistic models is a good perspective, since
in the presence of nontrivial conformal moduli (e.g., for multiply
connected domains or for domains with marked points) no such reduction by simple coordinate changes
could be used.
Moreover, Dirichlet and Neumann boundary condition free fields are rather evidently 
different probabilistic models.}
Nevertheless, our space of local fields of the CFT will always be the same:
the full Fock space~$\FullFock$ defined in Section~\ref{ssec: Heisenberg basics}.

Throughout this section, $\domain \subset \bC$ is taken to be a nonempty open simply
connected proper subset
of the complex plane.

Section~\ref{subsec: GFF def} contains the definition of
(the gradient of) the GFF, and a discussion of the sense in which
this basic field has pointwise defined correlation functions.
In Section~\ref{sub: current correlation kernel}, we then describe
correlation functions of something akin to linear local fields of the GFF,
which include, most notably, the holomorphic and antiholomorphic currents
$\HolCurr$ and~$\AntiHolCurr$.
From these correlation kernels of currents, a standard construction of
a bosonic CFT is given in Section~\ref{subsec: local fields of GFF}.
More precisely, using the currents and operator product expansions as building blocks,
we characterize and construct the $n$-point correlation functions of general Fock
space fields, i.e., for a general domain~$\domain$, boundary conditions
($\Dirichlet$ or~$\Neumann$), and number $n \in \N$ of points,
we describe a map
\begin{align*}
\CFTcorr{\domain}{\DorN}{\,\cdots\,} \; \colon \; (\FullFock)^{\tens n}
    \; \to \; \ContFun^\omega(\ConfigSp{n}{\domain},\bC) ,
\end{align*}
assigning in a symmetric multilinear way to an $n$-tuple of
Fock space field a complex-valued real-analytic correlation
function defined on the \term{configuration space}
\begin{align}\label{eq: configuration space}
\ConfigSp{n}{\domain} :=
    \domain^n \setminus
        \bigcup_{1 \le i < j \le n} \set{(z_1 , \ldots, z_n) \in \domain^n \; \big| \; z_i = z_j }
\end{align}
of (ordered) $n$-tuples of distinct points in~$\domain$.

The following main result will be precisely stated and proven in
Section~\ref{subsec: renormalized correlations}.
If the domain~$\domain$ is suitably approximated by discrete
domains~$\domain_\meshsz$ as $\meshsz \to 0$, then the correlations of local fields
of the DGFF on~$\domain_\meshsz$, renormalized by the lattice mesh~$\meshsz$ to the power
of the scaling dimensions of the fields, converge to the corresponding
correlation functions of Fock space fields in the bosonic CFT built from the GFF in~$\domain$.

\subsection{Gaussian Free Fields}
\label{subsec: GFF def}
Analogously to the discrete GFF,
one would like to think of the continuum GFF in~$\domain$ as a random real-valued
function~$\gff$ on~$\domain$ with a centered Gaussian distribution and covariance
$\EX[\gff(z) \gff(w)]$ equal to
$4 \pi$ times a continuum Green's function~$\Green_\domain(z,w)$.
However, due to the logarithmic singularity~\eqref{eq: Green's function log singularity}
of the Green's function when approaching~$z = w$, one cannot actually assign pointwise values to the GFF; rather 
the GFF in~$\domain$ should be constructed as a random generalized function, i.e., 
a random distribution. Moreover, in the case of Neumann boundary conditions, there is
an ambiguity of an additive constant (and in fact no unique choice of the Green's
function). We thus seek to define only the
gradient of the GFF
\[ \nabla \gff = \Big( \pderof{x}{\gff} , \; \pderof{y}{\gff} \Big) \]
as a two-component distribution-valued random variable.
Gradients are not arbitrary two-component vector fields ---they are necessarily
curl-free--- so a convenient equivalent
perspective is to view the field~$\gff$ itself as defined up to additive constants.
Let us furthermore note already that it will soon be convenient to
repackage the two components of the gradient in $\bC$-linear combinations:
the Wirtinger derivatives
\begin{align*}
\pder{z} = \frac{1}{2} \pder{x} - \frac{\ii}{2} \pder{y}
\qquad \text{ and } \qquad
\pder{\bar{z}} = \frac{1}{2} \pder{x} + \frac{\ii}{2} \pder{y}
\end{align*}
contain exactly the same information as the partial derivatives $\pder{x}$ and $\pder{y}$.
With suitable normalization constants,
$\dee \gff := \frac{\partial \gff}{\partial z}$ and
$\deebar \gff := \frac{\partial \gff}{\partial \bar{z}}$ will be the currents
in our CFT.

As ordinary \term{test functions}, we use compactly supported
$C^\infty$-smooth functions \linebreak[4] ${g \colon \domain \to \R}$.
The set of such test functions
is denoted by $\testfun{\domain}$,
and is equipped with the topology of uniform convergence
on compact subsets for all derivatives up to an arbitrary order.
The space of \term{distributions} is denoted by $\distrib{\domain}$; it is the dual
of the space of test functions, consisting of all continuous linear maps $\testfun{\domain} \to \R$.
The appropriate topology on $\distrib{\domain}$ is the weak-* topology. We denote the duality pairing
of a distribution $\phi \in \distrib{\domain}$ with a test function $g \in \testfun{\domain}$ by
$\langle\phi, g\rangle \in \R$.

The constant distribution $1' \in \distrib{\domain}$ is defined by
\begin{align*}
\langle 1', g \rangle = \int_\domain g(z) \, \ud^2 z
\qquad \text{ for } g \in \testfun{\domain}.
\end{align*}
The subspace of \term{zero-average test functions} is
\begin{align*}
\testfunZA{\domain} := \Kern{\;1'} = \set{f \in \testfun{\domain} \; \bigg| \; \int_\domain f(z) \, \ud^2 z = 0}
\subset \testfun{\domain} .
\end{align*}
Note that two distributions $\phi, \psi \in \distrib{\domain}$ differ by a multiple of the constant
distribution $1'$ if and only if $\langle \phi , f \rangle = \langle \psi , f \rangle$ for
all $f \in \testfunZA{\domain}$. The quotient space $\distrib{\domain} \, / \, \bR 1'$ of
\term{distributions up to additive constants} can therefore be identified with the space
$\distribMC{\domain}$ of continuous linear functionals on $\testfunZA{\domain}$.

We define the \term{GFF up to additive constants}, in a domain~$\domain$, with Dirichlet~($\Dirichlet$)
or Neumann~($\Neumann$) boundary conditions, as the $\distribMC{\domain}$-valued
random variable~$\gff$ whose characteristic function is
\begin{align}\label{eq: GFF characteristic function}
\EX^\DorN_\domain \big[ e^{\ii \langle \gff , f \rangle} \big]
    \; = \;\; & \exp \Big( - 2 \pi \iint_{\domain \times \domain}
        f(z) \, \Green^{\DorN}_\domain(z,w) \, f(w) \; \ud^2 w \, \ud^2 z \Big)
\quad \text{ for } f \in \testfunZA{\domain} .
\end{align}
Note that even in the Neumann case,
the double integral above does not depend on the choice of the Green's function.
This characteristic function uniquely determines the law of~$\gff$,
and an explicit construction of such a process~$\gff$
can be found, e.g., in~\cite{BP-GFF_and_LQG}.

Observe that for any $g \in \testfun{\domain}$, the partial derivatives are zero-average,
$\pderof{x}{g}, \pderof{y}{g} \in \testfunZA{\domain}$.
In particular the distributional derivatives $\pderof{x}{\gff}, \pderof{y}{\gff}$
of the GFF up to additive constants
are random ordinary distributions defined by
\begin{align*}
\big\langle \pderof{x}{\gff}, g \big\rangle = - \big\langle \gff, \pderof{x}{g} \big\rangle
\quad \text{ and } \quad
\big\langle \pderof{y}{\gff}, g \big\rangle = - \big\langle \gff, \pderof{y}{g} \big\rangle
\qquad \text{ for } g \in \testfun{\domain} ,
\end{align*}
and from the $\distribMC{\domain}$-valued random variable $\gff$
we thus obtain the \term{gradient of the Gaussian free field (GFF)}
${\nabla \gff = \big( \pderof{x}{\gff} , \; \pderof{y}{\gff} \big)}$ as a two-component
$\distrib{\domain}$-valued random variable.

From the characteristic function~\eqref{eq: GFF characteristic function}
one gets, in particular, the finite-dimensional marginals,
i.e., the joint distributions of $(\langle \gff , f_i \rangle)_{i=1}^n$ for any finite
number~$n$ of zero-average test functions $f_1 , \ldots , f_n \in \testfunZA{\domain}$.
These marginals are $n$-dimensional centered Gaussians with covariances
\begin{align*}
C_{ij} \, = \, 4 \pi
    \iint_{\domain \times \domain} f_i(z) \Green^{\DorN}_\domain(z,w) f_j(w) \, \ud^2 w \, \ud^2 z ,
\qquad i,j = 1 , \ldots , n .
\end{align*}
The directional derivatives~$\nabla^\mu$ of~\eqref{eq: continuum directional derivative}
are convenient for concisely writing down the covariances of the components of $\nabla \gff$.
For any test functions $g_1 , \ldots , g_n \in \testfun{\domain}$, the joint distribution
of $(\langle \nabla^{\mu_i} \gff , g_i \rangle)_{i=1}^n$ is, by construction,
the centered $n$-dimensional Gaussian with covariances
\begin{align*}
C^{\mu_i \, \mu_j}_{\; i \; j}
= \; & 4 \pi \iint_{\domain \times \domain}
        (\nabla^{\mu_i}_i g_i)(z_i) \; \Green^{\DorN}_\domain(z_i,z_j) \; (\nabla^{\mu_j}_j g_j)(z_j) \, \ud^2 z_i \, \ud^2 z_j \\
= \; & 4 \pi \iint_{\domain \times \domain}
        g_i(z_i) \, \big( \nabla^{\mu_i}_i \nabla^{\mu_j}_j \Green^{\DorN}_\domain(z_i,z_j) \big) \, g_j(z_j) \, \ud^2 z_i \, \ud^2 z_j 
\qquad i,j = 1 , \ldots , n ,
\end{align*}
where on the second line the double-derivative of the Green's function is in the
distributional sense (delta-like terms do appear on the diagonal $\{z_i=z_j\}$, but
if the supports of $g_i, g_j$ are disjoint, then even the second line only involves
ordinary integration of smooth functions).
Note that even in the Neumann case,
the double derivative of the Green's function on the last line above does not
depend on the precise choice of the Green's function.
As for any centered Gaussians, \term{Wick's formula} applies, and in this case gives

\begin{align}
    \label{eq: bona fide n point function of gff}
& \EX_\domain^{\DorN} \Big[ \prod_{i=1}^{n} \langle \nabla^{\mu_i} \gff , g_i \rangle \Big] \\
\nonumber
= \; & \sum_{P \in \Pair(n)} \prod_{\set{i,j} \in P} 4\pi \iint_{\domain \times \domain}
    g_i(z_i) \, \big( \nabla^{\mu_i}_i \nabla^{\mu_j}_j \Green_\domain^{\DorN}(z_i,z_j) \big) \,  g_j(z_j)
    \; \ud^2 z_i \, \ud^2 z_j \\
\nonumber
= \; & \idotsint_{\domain^n} g_1(z_1) \cdots g_n(z_n)
    \bigg( (4\pi)^{n/2} 
    \sum_{P \in \Pair(n)} \prod_{\set{i,j} \in P} \nabla^{\mu_i}_{\, i} \nabla^{\mu_j}_{\, j} \Green_\domain^{\DorN}(z_i, z_j) \bigg)
    \; \ud^2 z_1 \cdots \ud^2 z_n ,
\end{align}
where the sum is over the set~$\Pair(n)$ of
pairings~$P$ of the index set~$\set{1,2,\ldots,n}$.
In this sense of integral kernels (distributional in general, but literal
integration of smooth functions when the supports
of $g_1, \ldots, g_n$ are disjoint), we may thus interpret
the $n$-point \term{correlation function}
of the
components of the
gradient GFF~$\nabla \gff$ at $(z_1 , \ldots, z_n) \in \ConfigSp{n}{\domain}$
as
\begin{align}\label{eq: n point correlation of gff}
\text{\large{``}}
\EX_{\domain}^{\DorN} \big[ (\nabla^{\mu_1}\gff)(z_1) \cdots (\nabla^{\mu_n}\gff)(z_n) \big]
\text{\large{''}}
  := (4\pi)^{n/2} \sum_{P \in \Pair(n)} \prod_{\set{i,j} \in P} \nabla^{\mu_i}_{\, i} \nabla^{\mu_j}_{\, j} \Green^{\DorN}_\domain(z_i, z_j) .
\end{align}
One can also view~\eqref{eq: n point correlation of gff} as the limit
of~\eqref{eq: bona fide n point function of gff}
with mollifiers taken as approximate delta-functions
(``$g_i^{(\varepsilon)} \underset{\eps \to 0}{\longrightarrow} \delta(\cdot - z_i)$'')
at distinct points~$z_1, \ldots, z_n \in \domain$.

\subsection{Current correlation kernels of the GFF}
\label{sub: current correlation kernel}

Before addressing correlation functions of general local fields in the CFT, we slightly
generalize the above integral kernels for correlations of the gradient of the GFF,
and we transform them to the most convenient form for the purposes of constructing the
free boson CFT.
Namely, consider derivatives of arbitrary order of~$\gff(z)$~--- they serve as the natural
counterparts of linear local fields in the CFT. It is convenient to change basis from the
horizontal and vertical partial derivatives $\pderof{x}{\gff}$ and~$\pderof{y}{\gff}$ to the
holomorphic $\dee \gff = \frac{1}{2} \pderof{x}{\gff} - \frac{\ii}{2} \pderof{y}{\gff}$ and
antiholomorphic $\deebar \gff = \frac{1}{2} \pderof{x}{\gff} + \frac{\ii}{2} \pderof{y}{\gff}$
Wirtinger derivatives, and similarly in higher order.

In the same sense as~\eqref{eq: n point correlation of gff},
the correlation functions of the fields $\dee^{a}{\deebar} {}^{b} \gff$, for $a,b \in \N$ with $a+b>0$,
are defined by
\begin{align}
\nonumber
& \text{\large{``}}
\EX_{\domain}^{\DorN} \big[ (\dee^{a_1}{\deebar} {}^{b_1} \gff)(z_1) \cdots
        (\dee^{a_n}{\deebar} {}^{b_n} \gff)(z_n) \big]
\text{\large{''}} \\
\label{eq: n point correlation of linear local fields of gff}
:= \; & (4 \pi)^{n/2} \frac{\partial^{a_1}}{\partial z_1^{a_1}} \frac{\partial^{b_1}}{\partial \bar{z}_1^{b_1}}
        \cdots \frac{\partial^{a_n}}{\partial z_n^{a_n}} \frac{\partial^{b_n}}{\partial \bar{z}_n^{b_n}}
    \sum_{P \in \Pair(n)} \prod_{\set{i,j} \in P} \Green^{\DorN}_\domain(z_i, z_j) ;
\end{align}
these correlation functions $\ConfigSp{n}{\domain} \to \bC$ are the integral kernels
for correlations of the distributional Wirtinger derivatives of the GFF~$\gff$ mollified by test
functions, analogously to~\eqref{eq: bona fide n point function of gff}.
Observing the factorization of the Laplacian into Wirtinger derivatives,
$\Lapl = 4 \pder{z} \pder{\bar{z}}$, and noting that
\eqref{eq: n point correlation of linear local fields of gff} only contains sums of products
of the Green's functions,
we see that  the correlation~\eqref{eq: n point correlation of linear local fields of gff} is zero if
for some~$j$ we have both $a_j>0$ and $b_j>0$.
This is the continuum
counterpart of a null field property for $\Lapl \gff$ and its further derivatives.
In particular, the only derivative-fields with interesting correlations are the holomorphic
derivatives~$\dee^a \gff$ of order~$a \in \Zpos$ and the antiholomorphic derivatives ${\deebar} {}^b \gff$
of order~$b \in \Zpos$.

We define the holomorphic and antiholomorphic \term{currents} by $\HolCurr := \ii \dee \gff$
and $\AntiHolCurr := - \ii \deebar \gff$, respectively.
The integral kernel for $m$ holomorphic and $m'$ antiholomorphic currents
will play a key role in the construction of the CFT, and we therefore use the notation
\begin{align}\label{eq: integral kernel for holom and antiholom currents}
\GFFkernelWO{\domain}{\DorN} \in \ContFun^\omega \big(\ConfigSp{m+m'}{\domain} , \bC \big) 
\end{align}
for it.
This kernel has the informal interpretation
\begin{align*}
\GFFkernel{\domain}{\DorN}{z_1,\ldots,z_m}{w_1,\ldots,w_{m'}} = \; &
\text{\large{``}}
\EX_{\domain}^{\DorN} \big[ \HolCurr (z_1) \cdots \HolCurr (z_m)
         \AntiHolCurr (w_1) \cdots \AntiHolCurr (w_{m'}) \big]
\text{\large{''}} \; ;
\end{align*}
the explicit formula would be notationally cumbersome, but it is naturally just
$\ii^{m-m'}$ times~\eqref{eq: n point correlation of linear local fields of gff}
with $n = m+m'$,
$(a_i,b_i)=(1,0)$ for $i=1,\ldots,m$, and
$(a_{m+j},b_{m+j})=(0,1)$ for $j=1,\ldots,m'$, together with the relabeling
$w_{j}:=z_{m+j}$ for $j=1,\ldots,m'$.
Derivatives of any nonzero order of the GFF can be written in terms of the currents:
${\dee^a \gff = - \ii \dee^{a-1} \HolCurr}$ and ${{\deebar} {}^b \gff = \ii \dee^{b-1} \AntiHolCurr}$,
and the integral kernels for their correlations are obtained as further
derivatives of~$\GFFkernelWO{\domain}{\DorN}$.

The Green's functions~\eqref{eq: Green's function log singularity}
differentiated in both variables have the forms
\begin{align*}
\pder{z} \pder{w} \Green_\domain^{\DorN}(z, w)
= \; & \frac{-1/4\pi}{(z-w)^2} + \pder{z} \pder{w} g_\domain^\DorN(z,w) , \;\; &
\pder{z} \pder{\bar{w}} \Green_\domain^{\DorN}(z, w)
= \; & \pder{z} \pder{\bar{w}} g_\domain^\DorN(z,w) , \\
\pder{\bar{z}} \pder{\bar{w}} \Green_\domain^{\DorN}(z, w)
= \; & \frac{-1/4\pi}{(\bar{z}-\bar{w})^2} + \pder{\bar{z}} \pder{\bar{w}} g_\domain^\DorN(z,w) , \;\; &
\pder{\bar{z}} \pder{w} \Green_\domain^{\DorN}(z, w)
= \; & \pder{\bar{z}} \pder{w} g_\domain^\DorN(z,w) \,
\end{align*}
where $g_\domain^\DorN \colon \domain \times \domain \to \bR$ is real-analytic, and harmonic
in both variables separately.
By virtue of the factorization~$\Lapl = 4 \pder{z} \pder{\bar{z}}$ of the
Laplacian,
the double Wirtinger derivatives of~$g_\domain^\DorN$ are then
real-analytic functions $\domain \times \domain \to \bC$ which are
holomorphic or antiholomorphic separately in the two variables.
We obtain, in particular, the following explicit formula for the two-point function of the holomorphic current
\begin{align*}
\text{\large{``}}
\EX_{\domain}^{\DorN} \big[ \HolCurr (z_1) \HolCurr (z_2) \big]
\text{\large{''}}
\;=\;
\GFFkernel{\domain}{\DorN}{z_1,z_2}{}
\;=\;
\frac{1}{(z_1-z_2)^2}
    - 4\pi \underbrace{\pder{z_1}\pder{z_2} g^{\DorN}_\domain(z_1,z_2)}_{\text{regular as $z_1 \to z_2$}} .
\end{align*}
This shows that
$z_1 \mapsto \GFFkernel{\domain}{\DorN}{z_1,z_2}{} - \frac{1}{(z_1-z_2)^2}$
is holomorphic in the whole~$\domain$. Similarly, 
$w_1 \mapsto \GFFkernel{\domain}{\DorN}{}{w_1,w_2} - \frac{1}{(\cconj{w}_1-\cconj{w}_2)^2}$
is antiholomorphic in the whole~$\domain$. In the mixed two-point functions, there are no
singularities: just by themselves,
$z \mapsto \GFFkernel{\domain}{\DorN}{z}{w}$ is holomorphic and 
$w \mapsto \GFFkernel{\domain}{\DorN}{z}{w}$ is antiholomorphic
in the whole~$\domain$.
Furthermore, since multi-point correlations are obtained by Wick's formula,
$\GFFkernel{\domain}{\DorN}{z_1,\ldots,z_m}{w_1,\ldots,w_{m'}}$
is meromorphic in each of the variables $z_i$ with poles of order $2$ at 
$z_1, \ldots, z_{i-1}, z_{i+1}, \ldots, z_m$ only,
and antimeromorphic in each of the variables $w_i$ with (anti)poles of order $2$ at
$w_1, \ldots, w_{i-1}, w_{i+1}, \ldots, w_{m'}$ only.

\subsection{Correlation functions of CFT fields from current OPEs}
\label{subsec: local fields of GFF}

One possible way to construct higher degree fields would be by forming normally
ordered products of the linear fields, see, e.g., \cite[Section~1.2]{KM-GFF_and_CFT}.
Making sense of these as distribution-valued random variables and defining
pointwise correlation functions as integral kernels is genuinely more complicated than
for the linear fields, see~\cite[Section~2.2]{KM-GFF_and_CFT}.
There is, however, an even more serious drawback:
the (probabilistic) normal ordering is not intrinsic to the fields themselves.
Indeed, the probabilistic normal ordering depends on the Gaussian measure (the law of~$\gff$),
which depends on the domain~$\domain$ and boundary conditions, so,
from a CFT point of view,
normal ordering is \emph{not} the right definition of an abstract higher degree field
whose correlation functions can be evaluated in all domains and with any reasonable boundary conditions.

The intrinsic (domain-agnostic and boundary-condition-agnostic) way
to construct higher degree local fields is by recursively
extracting operator-product expansion (OPE) coefficients from lower order fields,
as discussed, e.g., in \cite[Section~3]{KM-GFF_and_CFT}.
Starting from the currents $\HolCurr(z)$ and $\AntiHolCurr(w)$,
general Fock space fields
are generated via such OPEs.
The holomorphic current $\HolCurr(z)$ has a purely meromorphic OPE with any Fock space field, so
an easy way to extract the OPE coefficients is by suitable weighted contour integrals.
Similarly the  antiholomorphic current $\AntiHolCurr(w)$ has a purely antimeromorphic OPE with any
Fock space field, and again OPE coefficients can be extracted by suitable contour integrals.
This is, roughly speaking, the route we take below to define arbitrary
$n$-point correlation functions
\begin{align}\label{eq: the type of the correlation function}
\CFTcorr{\domain}{\DorN}{\,\cdots\,} \; \colon \; (\FullFock)^{\tens n}
    \; \to \; & \; \ContFun^\omega(\ConfigSp{n}{\domain},\bC)  . \\ \nonumber
F_1 \tens \cdots \tens F_n
    \mapsto \; & \bigg( (z_1, \ldots, z_n) 
    \mapsto \CFTcorrBig{\domain}{\DorN}{F_1(z_1) \cdots F_n(z_n)} \bigg)
\end{align}
for a bosonic CFT with the Fock space~$\FullFock$ as its space of local fields.

\begin{prop}\label{prop: folklore}
Consider the domain~$\domain$ and boundary conditions~$\DorN$ fixed.
Then there exists a unique collection of
linear assignments of type~\eqref{eq: the type of the correlation function},
indexed by $n \in \bN$, such that the following conditions hold:
\begin{description}
\item[(BOS)] For any permutation $\sigma \in \SymmGrp_{n}$, we have
\begin{align*}
\CFTcorrBig{\domain}{\DorN}{F_1(z_1) \cdots F_n(z_n)}
= \CFTcorrBig{\domain}{\DorN}{F_{\sigma(1)}(z_{\sigma(1)}) \cdots F_{\sigma(n)}(z_{\sigma(n)})} \; .
\end{align*}
\item[(CUR)] For $(z_1, \ldots, z_m, w_1, \ldots, w_{m'}) \in \ConfigSp{m+m'}{\domain}$,
we have
\begin{align*}
& \CFTcorrBig{\domain}{\DorN}{(\HeiJ{-1} \FockId)(z_1) \cdots (\HeiJ{-1} \FockId)(z_{m})
    (\AntHeiJ{-1} \FockId)(w_1) \cdots (\AntHeiJ{-1} \FockId)(w_{m'}) } \\ \nonumber
= \; &
\GFFkernel{\domain}{\DorN}{z_1,\ldots,z_m}{w_1,\ldots,w_{m'}}
\end{align*}
where the right-hand side is the integral kernel~\eqref{eq: integral kernel for holom and antiholom currents}
of GFF current correlations in~$\domain$ with the chosen boundary conditions~$\DorN$.
\item[($\HolCurr$-OPE)] For any $F_1, \ldots,F_n \in \FullFock$ and any $(z_1, \ldots, z_n) \in \ConfigSp{n}{\domain}$,
we have the following Laurent series expansion of the $n+1$-point function
\begin{align}
& \CFTcorrBig{\domain}{\DorN}{ F_1(z_1) \cdots F_j(z_j) \cdots F_n(z_n) \, (\HeiJ{-1} \FockId)(z)} \\ \nonumber
= \; & \sum_{k \in \bZ} (z-z_j)^{-1-k}
    \CFTcorrBig{\domain}{\DorN}{ F_1(z_1) \cdots (\HeiJ{k} F_j)(z_j) \cdots F_n(z_n) } \, ,
\end{align}
in the region $0 < |z-z_j| < \min \big\{ \min_{i \ne j} |z_j-z_i| , \, \dist (z_j , \bdry \domain) \big\}$.
\item[($\AntiHolCurr$-OPE)] For any $F_1, \ldots,F_n \in \FullFock$ and any $(z_1, \ldots, z_n) \in \ConfigSp{n}{\domain}$,
we have the following anti-Laurent series expansion of the $n+1$-point function
\begin{align}
& \CFTcorrBig{\domain}{\DorN}{ F_1(z_1) \cdots F_j(z_j) \cdots F_n(z_n) \, (\AntHeiJ{-1} \FockId)(w)} \\ \nonumber
= \; & \sum_{k \in \bZ} (\cconj{w}-\cconj{z}_j)^{-1-k}
    \CFTcorrBig{\domain}{\DorN}{ F_1(z_1) \cdots (\AntHeiJ{k} F_j)(z_j) \cdots F_n(z_n) } \, ,
\end{align}
in the region $0 < |w-z_j| < \min \big\{ \min_{i \ne j} |z_j-z_i| , \, \dist (z_j , \bdry \domain) \big\}$.
\end{description}
\end{prop}

\begin{proof}
The uniqueness of such~$\CFTcorr{\domain}{\DorN}{\cdots}$
is straightforward by a recursive argument: (CUR) provides the base
case, and with repeated coefficient extraction using ($\HolCurr$-OPE) and ($\AntiHolCurr$-OPE)
one obtains
expressions for the general correlation functions. The notation introduced
in the existence proof will make the explicit expression clear.

To prove the existence of~$\CFTcorr{\domain}{\DorN}{\cdots}$ as
in~\eqref{eq: the type of the correlation function},
we start by constructing correlation functions for
$n$ formal expressions
\begin{align}\label{eq: Fock fields}
F_i := \; & \HeiJ{k_{i;m_i}} \cdots \HeiJ{k_{i;2}} \, \HeiJ{k_{i;1}} \,
    \AntHeiJ{k'_{i;m'_i}} \cdots \AntHeiJ{k'_{i;2}} \, \AntHeiJ{k'_{i;1}} \FockId , \\
\nonumber
\text{ with } \; &
    m_i, m_i' \in \Znn ,  \quad
    k_{i;1}, \ldots, k_{i;m_i}, k'_{i;1}, \ldots, k_{i;m'_i} \in \bZ .
\end{align}
indexed by $i = 1, \ldots, n$. To then finally
define~\eqref{eq: the type of the correlation function}, we extend
from the formal expressions $n$-multilinearly, and check that the
correlation functions become well-defined, i.e., in each tensorand,
they factor through a quotient which defines the Fock space~$\FullFock$.
The quotient has to account for both the Heisenberg commutation relations and the
quotient construction~\eqref{eq: Fock representation} of the Fock representation
in both chiralities.

The $n$-point correlation function
\begin{align*}
\CFTcorrBig{\domain}{\DorN}{F_1(z_1) \cdots F_n(z_n)}
\end{align*}
of formal expressions~\eqref{eq: Fock fields}
at $(z_1, \ldots, z_n) \in \ConfigSp{n}{\domain}$
is constructed starting from
the kernel~\eqref{eq: integral kernel for holom and antiholom currents}
for current correlations
\begin{align*}
\GFFkernel{\domain}{\DorN}
    {\overbrace{\underbrace{\zeta_{1;1},\ldots,\zeta_{1;m_1}}_{\text{$m_1$ variables}}, \ldots, \underbrace{\zeta_{n;1},\ldots,\zeta_{n;m_n}}_{\text{$m_n$ variables}}}^{\text{$n$ groups of variables}}}
    {\overbrace{\underbrace{\xi_{1;1},\ldots,\xi_{1;m'_1}}_{\text{$m'_1$ variables}}, \ldots, \underbrace{\xi_{n;1},\ldots,\xi_{n;m'_n}}_{\text{$m'_n$ variables}}}^{\text{$n$ groups of variables}}}
\end{align*}
with $\sum_i m_i$ holomorphic currents and $\sum_i m'_i$
antiholomorphic currents, by integrating each of the
variables~$\zeta_{i;s_i}$ and $\xi_{i;t_i}$ around~$z_i$
along suitable radially ordered contours
with the appropriate weight depending on~$k_{i;s_i}$ or~$k'_{i;t_i}$.
Specifically, for each $i=1,\ldots,n$, choose radii
\begin{align*}
0 < r'_{i;1} < r'_{i;2} < \cdots < r'_{i;m'_i} < r_{i;1} < r_{i;2} < \cdots < r_{i;m_i} < R ,
\end{align*}
where $R$ is small enough so that the disks $\diskRC{R}{z_i}$, $i=1,\ldots,n$,
are disjoint and contained in the domain~$\domain$. Then the defining formula for the
correlation function of the expressions~\eqref{eq: Fock fields} is
\begin{align}
\label{eq: CFT correlation}
& \CFTcorrBig{\domain}{\DorN}{F_1(z_1) \cdots F_n(z_n)} \\ \nonumber
:= \; & 
    \oint \!\cdot\!\cdot\!\cdot\! \oint
    \GFFkernel{\domain}{\DorN}{\zeta_{1;1},\ldots}{\ldots,\xi_{n;m'_n}} 
    \prod_{i=1}^n \Bigg( \prod_{s=1}^{m_i} (\zeta_{i;s}-z_i)^{k_{i;s}} \frac{\ud \zeta_{i;s}}{2 \pi \ii}
        \prod_{t=1}^{m'_i} (\overline{\xi}_{i;s}-\overline{z}_i)^{k'_{i;t}} \frac{\ud \bar{\xi}_{i;t}}{(-2 \pi \ii)} \Bigg) ,
\end{align}
where $\zeta_{i;s_i}$ is integrated over the positively oriented circle~$\bdry\diskRC{r_{i;s_i}}{z_i}$
and $\xi_{i;t_i}$ over the positively oriented circle~$\bdry\diskRC{r'_{i;t_i}}{z_i}$.
By holomorphicity and antiholomorphicity of the correlation kernel $\GFFkernelWO{\domain}{\DorN}$
of the currents~$\HolCurr$ and~$\AntiHolCurr$,
the value of~\eqref{eq: CFT correlation} only depends on the homotopy class of these
integrations, and in particular any choice of radii with the prescribed ordering yields the same result.
Moreover, since there are no singularities at $\zeta = \xi$ of $\HolCurr$-$\AntiHolCurr$-correlations
$\GFFkernel{\domain}{\DorN}{\ldots, \zeta, \ldots}{\ldots, \xi, \ldots}$
the homotopies
can be used to move holomorphic current integrations past antiholomorphic current integrations, so
the radial ordering convention can be relaxed to separately requiring
\begin{align*}
0 < r'_{i;1} < r'_{i;2} < \cdots < r'_{i;m'_i} < R
\qquad \text{ and } \qquad 
0 < r_{i;1} < r_{i;2} < \cdots < r_{i;m_i} < R .
\end{align*}

We must then check that
the multi-linear extension of~\eqref{eq: CFT correlation}
gives rise to well-defined correlation functions
\begin{align*}
\CFTcorr{\domain}{\DorN}{\,\cdots\,} \; \colon \; (\FullFock)^{\tens n}
    \; \to \; \ContFun^\omega(\ConfigSp{n}{\domain},\bC)  .
\end{align*}
The first thing to check is that
the correlation functions~\eqref{eq: CFT correlation} of the formal
expressions~\eqref{eq: Fock fields}
satisfy equations
corresponding
to the commutation relations~\eqref{eq: Heisenberg bracket} of the two commuting Heisenberg
algebras in each of the $n$ tensorands.
Concretely, for
\begin{align*}
G_i := \HeiJ{k} \HeiJ{\ell} F_i
\qquad \text{ and } \qquad
H_i := \HeiJ{\ell} \HeiJ{k} F_i ,
\end{align*}
by standard satellite integral arguments and the explicit poles of the current
correlation functions~\eqref{eq: integral kernel for holom and antiholom currents},
one can prove equalities of correlation functions~\eqref{eq: CFT correlation} of the form
\begin{align*}
\CFTcorrBig{\domain}{\DorN}{\cdots G_i(z_i) \cdots}
- \CFTcorrBig{\domain}{\DorN}{\cdots H_i(z_i) \cdots}
\; = \; k \delta_{k+\ell} \CFTcorrBig{\domain}{\DorN}{\cdots F_i(z_i) \cdots} .
\end{align*}
Similar equations hold relating correlations of $\AntHeiJ{k} \AntHeiJ{\ell} F_i$,
$\AntHeiJ{\ell} \AntHeiJ{k} F_i$ and $F_i$.
From such commutation equations it follows
that the correlation functions of formal linear combinations
of symbols~\eqref{eq: Fock fields} at least factor through
$\UHei \tens \AntUHei$ in each tensorand, and~\eqref{eq: CFT correlation}
gives rise to a map
\begin{align*}
(\UHei \tens \AntUHei)^{\tens n} \to \ContFun^\omega(\ConfigSp{n}{\domain},\bC) .
\end{align*}
For well-definedness on the Fock space~$\FullFock$, it therefore remains to check
equations corresponding to the quotient construction~\eqref{eq: Fock representation}
of the Fock representations in the two chiralities.

If $F_i = \HeiJ{k}\FockId$ for some~$i$, then there is exactly one integration around $z_i$
in~\eqref{eq: CFT correlation}, with
$\zeta_{i;1}$ as the integration variable and $(\zeta_{i;1} - z_i)^k$ as the weight.
But the correlation kernel~\eqref{eq: integral kernel for holom and antiholom currents}
of the currents is holomorphic in~$\zeta_{i;1}$, so this integration
(alone) picks the coefficient of $(\zeta_{i;1} - z_i)^{-1-k}$ in the Taylor series
expansion of the holomorphic function
${\zeta_{i;1} \mapsto \GFFkernel{\domain}{\DorN}{\zeta_{1;1},\ldots}{\ldots,\xi_{n;m'_n}}}$
at~$z_i$, which is zero for $k \ge 0$.
This shows factorization through $\ChiralFock = \UHei \big/ (\HeiJ{k} : k \ge 0)$
in the holomorphic chirality of the $i$th tensorand.
For the factorization
\begin{align*}
(\FullFock)^{\tens n} \to \ContFun^\omega(\ConfigSp{n}{\domain},\bC) 
\end{align*}
through the full Fock space in each factor, one repeats an entirely similar
argument with the antiholomorphic current.

We thus obtain a collection of well-defined maps~\eqref{eq: the type of the correlation function}
indexed by~$n$. It remains to check the asserted properties
(BOS), (CUR), ($\HolCurr$-OPE), and ($\AntiHolCurr$-OPE).

Symmetricity~(BOS) is evident from the symmetricity of the
kernel~$\GFFkernelWO{\domain}{\DorN}$ and of the integral~\eqref{eq: CFT correlation}.

The property~(CUR) is evident
by choosing ${m_1=\cdots=m_m=1}$, ${m'_1=\cdots=m'_{m'}=1}$,
and ${k_{1;1}=\cdots=k_{m;1}=-1}$, ${k'_{1;1}=\cdots=k'_{m';1}=-1}$, and
noticing that the current kernel
$\GFFkernel{\domain}{\DorN}{z_1,\ldots,z_m}{w_1,\ldots,w_{m'}}$ is then
obtained as the residue of the defining integrals~\eqref{eq: CFT correlation}.

The proofs of ($\HolCurr$-OPE) and ($\AntiHolCurr$-OPE) are similar, so let us
only comment on the former.
Consider a correlation function of the form~\eqref{eq: CFT correlation}
with $n+1$ fields of the form~\eqref{eq: Fock fields}, where the $(n+1)$st one is taken to be
simply $F_{n+1} = \HeiJ{-1}\FockId$.
This correlation function is meromorphic in the variable~$z_{n+1}$
(as shown by induction on the total degree $\sum_i (m_i + m'_i)$).
Its Laurent series expansion at~$z_i$ is of the form
\begin{align*}
\CFTcorrBig{\domain}{\DorN}{F_1(z_1) \cdots F_n(z_n) \, (\HeiJ{-1}\FockId) (z_{n+1})}
= \sum_{p=-p_0}^\infty (z_{n+1} - z_i)^p \; C_p(z_1, \ldots, z_n) .
\end{align*}
The Laurent series coefficients admit contour integral expressions,
\begin{align*}
C_p(z_1, \ldots, z_n) = \frac{1}{2\pi\ii} \oint_{\bdry \diskRC{R}{z_i}}
    \CFTcorrBig{\domain}{\DorN}{F_1(z_1) \cdots F_n(z_n) \, (\HeiJ{-1}\FockId) (z_{n+1})}
    \, (z_{n+1}-z_i)^{-1-p} \, \ud z_{n+1} .
\end{align*}
Writing out the definition of the correlation function and extracting the residue in
its single integration around $z_{n+1}$ as in~(CUR), we obtain an expression
for the coefficient~$C_p$ which coincides,
by the definition~\eqref{eq: CFT correlation}, with the
$n$-point correlation function
\begin{align*}
C_p(z_1, \ldots, z_n) = \CFTcorrBig{\domain}{\DorN}{F_1(z_1) \cdots (\HeiJ{-k} F_i)(z_i) \cdots F_n(z_n)}
\qquad \text{ with } k = {-1-p}
\end{align*}
involving the higher degree field~$\HeiJ{-k} F_i$.
\end{proof}

Finally, let us make a few further remarks on the interpretations
of the fields in the correlation functions~\eqref{eq: CFT correlation}.

\emph{Identity as a Fock space field:}
If for some $i$ we have $F_i = \FockId$, i.e., $m_i = 0 = m'_i$, then there are no integrations
around $z_i$ and no $z_i$-dependent weights in~\eqref{eq: CFT correlation}. The resulting correlation
function does not depend on~$z_i$, and its value is the same as the correlation obtained by
omitting~$F_i(z_i)$. Such an~$F_i$ is interpreted as the ``identity field''.

\emph{Holomorphic and antiholomorphic currents and their derivatives as Fock space fields:}
The fields $\HeiJ{-1}\FockId$ and $\AntHeiJ{-1}\FockId$ should be interpreted as
the holomorphic and antiholomorphic current fields in the bosonic CFT.
Correlation functions which only involve these are exactly the probabilistic GFF
current correlation kernels,
\begin{align*}
\CFTcorrBig{\domain}{\DorN}{(\HeiJ{-1} \FockId)(z_1) \cdots
    (\AntHeiJ{-1} \FockId)(w_{m'}) } 
= \; & \;\, \GFFkernel{\domain}{\DorN}{z_1,\ldots,z_m}{w_1,\ldots,w_{m'}} \\
= \; & \text{\large{``}} \EX_{\domain}^{\DorN} \big[ \HolCurr (z_1)
            \cdots \HolCurr (z_m) \AntiHolCurr (w_1) \cdots \AntiHolCurr (w_{m'}) \big]
    \text{\large{''}} \, ,
\end{align*}
according to the property~(CUR). The CFT fields
$\HeiJ{-1}\FockId$ and $\AntHeiJ{-1}\FockId$, however,
have correlation functions not only among themselves, but more generally with any other
Fock space fields. In view of this, the holomorphic and antiholomorphic current can be interpreted as fields in the Fock space by the identifications
\begin{align*}
\HolCurr \coloneqq \HeiJ{-1}\FockId \in \FullFock
\qquad \text{ and } \qquad
\AntiHolCurr \coloneqq \AntHeiJ{-1}\FockId \in \FullFock \; .
\end{align*}
Similarly, for $p \ge 0$, the Fock space fields
$\HeiJ{-1-p} \, \FockId$ and $\AntHeiJ{-1-p} \, \FockId$
have the interpretations as the higher order derivative linear fields
$\frac{1}{p!} {\dee} {}^p \HolCurr$  and
$\frac{1}{p!} {\deebar} {}^p \AntiHolCurr$, respectively.
Indeed, as in the proof of (CUR) above, the construction of the correlation
functions of these involves picking a residue which recovers the corresponding
derivative of the GFF correlation kernel.

\subsection{Renormalized limits of correlations of lattice local fields}
\label{subsec: renormalized correlations}
In the scaling limit statements, we let the lattice mesh $\meshsz>0$ tend to zero.
Continuum domains $\domain \subset \bC$ are approximated by discrete
domains $\domain_\meshsz \subset \meshsz \bZ^2$ on the $\meshsz$-mesh square grid.
We will show that expected values $\EX_{\ddomain}^{\DorN} [\,\cdots]$
of products of lattice local fields w.r.t. the DGFF measures of
Section~\ref{sub: DGFF def} converge, when suitably renormalized,
to the CFT correlation functions $\CFTcorr{\domain}{\DorN}{\cdots}$ given
by Proposition~\ref{prop: folklore}.

The correspondence between CFT fields and lattice local fields, used
in the expected values and CFT correlations, respectively, is provided by the
isomorphism
\begin{align*}
\dFields \cong \FullFock
\end{align*}
of Theorem~\ref{thm: main theorem about Fock space structure}.
We (ab)use this one-to-one correspondence to interpret a given lattice local field
$F = P + \Null \in \dFields$
as a Fock space field $F \in \FullFock$, or vice versa.
The former has well-defined correlations
$\EX_{\ddomain}^{\DorN} \big[ \ev_{\zs}^{\domain_\meshsz}(P) \cdots \big]$
independent of the chosen representative~$P$
in cases relevant to the scaling limit, and the latter has
meaningful CFT correlations $\CFTcorr{\domain}{\DorN}{F(z)\cdots}$.
To facilitate stating the scaling limit as locally uniform convergence on
the configuration spaces $\ConfigSp{n}{\domain}$, we extend the definition
of evaluation 
of field polynomials~$P$ to points~$z \in \domain$ in the continuum domain,
by setting~$\ev^{\domain_\meshsz}_{z}(P) \coloneqq \ev^{\domain_\meshsz}_{\zs}(P)$
with $\zs\in\domain_\meshsz$ a closest lattice point to~$z\in\domain$.
The lattice local fields need to be renormalized in the scaling limit according
to their~$\VirL{0} + \VirBarL{0}$ eigenvalue, so we phrase the statement in terms of
homogeneous fields ${F \in \dHomFi{\Delta}{\bar\Delta}}$ defined
in~\eqref{eq: homogeneous local fields}. Note that by virtue of the isomorphism
$\dFields \cong \FullFock$, any local field can
be written as a finite linear combination of such homogeneous local fields.

\begin{thm}\label{thm: main theorem about scaling limits}
Fix boundary conditions~$\DorN$ and an approximation~$(\ddomain)_{\meshsz > 0}$
of a domain~$\domain$
in the Carath\'eodory sense
by discrete domains~$\ddomain \subset \SqLatMesh$.
Let ${F_1 , \ldots , F_n \in \dFields \cong \FullFock}$ be local fields
satisfying $F_i\in\dHomFi{\Delta_i}{\bar\Delta_i}$, and
fix representative field polynomials~$P_i \in \dLocFi$ so that
$F_i = P_i + \Null$.
Then, as $\meshsz \to 0$, we have
\begin{align}\label{eq: main scaling limit result}
	\meshsz^{-\sum_j (\Delta_j+\bar\Delta_j)} \;
	\EX_{\domain_\meshsz}^{\DorN} 
	\Big[\,
	\ev_{z_1}^{\domain_\meshsz}(P_1)
	\,\cdots\,
	\ev_{z_n}^{\domain_\meshsz}(P_n)
	\,\Big]
	\; \longrightarrow \;
	\CFTcorrBig{\domain}{\DorN}{F_1(z_1) \cdots F_n(z_n)}
\end{align}
uniformly for $(z_1,\ldots,z_n)$ in compact subsets of~$\ConfigSp{n}{\domain}$.
\end{thm}

The proof strategy should already naturally suggest itself.
If~$F_i$ are basis vectors of the Fock space of fields, then
the right hand side of~\eqref{eq: main scaling limit result} is explicitly
given by the multiple contour integral~\eqref{eq: CFT correlation}. The left hand side
expected value is given by analogous discrete contour integrations, with discrete
monomial weights.
The discrete integration contours can be chosen
to approximate rectangular continuum contours.
The DGFF expected values inside the discrete integrations involve only
discrete currents, so they are simply
written in terms of discrete double derivatives of discrete Green's functions,
which converge to their continuum counterparts, and in the limit they reproduce the integral
kernel~$\GFFkernelWO{\domain}{\DorN}$ of GFF currents.
The discrete monomial weights also converge to the monomial weights in~\eqref{eq: CFT correlation},
when an appropriate scaling by lattice mesh is taken into account.
Together, the scaling factors from the discrete monomials account for the
renomalization factor~$\meshsz^{-\sum_j (\Delta_j+\bar\Delta_j)}$.
Finally, the discrete integrations can be viewed as Riemann sum approximations, so the
locally uniform convergence of the integrands implies the locally uniform
convergence~\eqref{eq: main scaling limit result}. The proof below provides
concrete details following this natural strategy.

\begin{proof}
By linearity of expected values, it suffices to prove the statement for basis vectors of
the joint eigenspaces~$\dHomFi{\Delta_i}{\bar\Delta_i}$ of $\VirL{0}$ and $\VirBarL{0}$
with eigenvalues $\Delta_i$ and $\bar{\Delta}_i$.
Fields of the form
$F_i = \HeiJ{-k_{i;m_i}} \cdots \HeiJ{-k_{i;1}} \HeiJ{-k'_{i;m'_i}} \cdots \HeiJ{-k'_{i;1}} \FockId$,
with $\sum_{s = 1}^{m_i} k_{i;s} = \Delta_i$ and $\sum_{s=1}^{m'} k'_{i;s} = \bar{\Delta}_i$
form the convenient basis for us.
The essence of the proof becomes clear already from the case
$F_i = \HeiJ{-k_{i;m_i}} \cdots \HeiJ{-k_{i;1}} \FockId$, for $i=1,\ldots,n$ so
for notational simplicity let us assume this.

The CFT correlation on the right-hand side of the assertion is defined
by~\eqref{eq: CFT correlation}, and for
fields of the form $F_i = \HeiJ{-k_{i;m_i}} \cdots \HeiJ{-k_{i;1}} \FockId$ it
simplifies to
\begin{align*}
& \CFTcorrBig{\domain}{\DorN}{F_1(z_1) \cdots F_n(z_n)} \\
:= \; & 
    \oint \!\cdot\!\cdot\!\cdot\! \oint
    \GFFkernel{\domain}{\DorN}{\zeta_{1;1} , \ldots , \zeta_{n;m_n}}{} 
    \prod_{i=1}^n \prod_{s=1}^{m_i} (\zeta_{i;s}-z_i)^{-k_{i;s}} \frac{\ud \zeta_{i;s}}{2 \pi \ii} \\
= \; & \sum_{P} \prod_{\{(i, s), (j, t)\} \in P} 4\pi \oint\oint
        \frac{\partial^2 \Green^{\DorN}_\domain (\zeta_{i;s}, \zeta_{j;t})}{\partial{\zeta_{i;s}} \partial{\zeta_{j;t}}}
        (\zeta_{i;s}-z_{i})^{-k_{i;s}}
        (\zeta_{j;t}-z_{j})^{-k_{j;t}}
        \frac{\ud \zeta_{i;s}}{2 \pi} \frac{\ud \zeta_{j;t}}{2 \pi} ,
\end{align*}
where $P$ is summed over the pairings of the set
$\set{(i,s) \; \big| \; 1 \le i \le n, \, 1 \le s \le m_i}$
and each $\zeta_{i;s}$ is integrated in the positive direction around~$z_i$ along
radially ordered circular contours as described below~\eqref{eq: CFT correlation}.
We can, however, deform the
integrations to 
radially-ordered concentric square contours~$\Gamma_{i;s}$
contained in $\domain$  and centered at $z_i$~---
see Figure~\ref{fig: square contours}.
We fix a choice of such contours for the rest of the proof.
We will not write full details about the local uniformity of the convergence below,
but locally uniform error estimates are routine to obtain if one
ensures here that the chosen $\Gamma_{i;s}$ is only changed slightly for nearby values of~$z_i$.

To write down the expected value on the left-hand side of the assertion for a given
small mesh size~$\meshsz>0$, we choose the unit mesh
corner contours $\gamma_{i;s}^{(\meshsz)}$
that are closest from inside to
the blown-up and shifted continuum square contours $\meshsz^{-1}(\Gamma_{i;s}-z_i)$
such that each of their sides has an even number of corner steps~--- see
Figure~\ref{fig: square contours}.
Then, for small $\meshsz>0$, the expected value involving evaluations of
specific representatives~$P_i$ of $F_i$ can, by discrete contour deformation
(without changing the value of the expectation), be written in the form
\begin{align}\label{eq: main discrete formula}
	& \EX_{\domain_\meshsz}^{\DorN} 
	\Big[\,
	\ev_{z_1}^{\domain_\meshsz}(P_1)
	\,\cdots\,
	\ev_{z_n}^{\domain_\meshsz}(P_n)
	\,\Big]
	\\
	\nonumber
	= \; &	
    \dcoint{} \!\cdot\!\cdot\!\cdot\!\, \dcoint{}
    \EX_{\domain_\meshsz}^{\DorN} \left[\, \prod_{i=1}^n \prod_{s=1}^{m_i}
	\dHolCurr \big( \zs^\meshsz_i + \meshsz \zu^{i;s}_\medial \big) \,\right]
    \prod_{i=1}^n \prod_{s=1}^{m_i} \big( \zu^{i;s}_\diamond \big)^{[-k_{i;s}]}
    \frac{\dd{\zu^{i;s}}}{2\pi\ii} \\ \nonumber
= \; & \sum_{P} \prod_{\{(i, s), (j, t)\} \in P}
    4\pi \dcoint{\gamma_{i;s}^{(\meshsz)}} \dcoint{\gamma_{j;t}^{(\meshsz)}}
        (\pdee)_{\zu^{i;s}} (\pdee)_{\zu^{j;t}}
        \Green_{\domain_\meshsz}^{\DorN}
            \big( \zs_i^\meshsz + \meshsz\zu^{i;s}_\medial , \zs_j^\meshsz + \meshsz\zu^{j;t}_\medial \big)
    \\  \nonumber
& \mspace{180mu}        \big( \zu^{i;s}_\diamond \big)^{[-k_{i;s}]}
        \big( \zu^{j;t}_\diamond \big)^{[-k_{j;t}]}
        \frac{\dd{\zu^{i;s}}}{2\pi} \frac{\dd{\zu^{j;t}}}{2\pi} ,
\end{align}
where $\zs^\meshsz_i\in\domain_\meshsz$ is the closest vertex to $z_i\in\domain$.

Both the CFT correlation and the discrete expectation were written above as
sums of products of double-integral factors,
so it suffices to handle the convergence of each such factor.
More precisely, remembering the form of the renormalizing prefactor
\[{\meshsz^{-\sum_i \Delta_i} = \meshsz^{-\sum_i \sum_{s = 1}^{m_i} k_{i;s}}} , \]
it suffices to prove, with $k = k_{i;s}$ and $\ell = k_{j;t}$, that as $\meshsz \to 0$ we have
\begin{align}\label{eq: foo}
& \meshsz^{- k - \ell} \dcoint{\gamma_{i;s}^{(\meshsz)}} \dcoint{\gamma_{j;t}^{(\meshsz)}}
        (\pdee)_{\zu} (\pdee)_{\zubis}
        \Green_{\domain_\meshsz}^{\DorN}
            \big( \zs_i^\meshsz + \meshsz\zu_\medial , \zs_j^\meshsz + \meshsz\zubis_\medial \big)
        \, \zu_\diamond^{[-k]} \, \zubis_\diamond^{[-\ell]}
        \; {\dd{\zu}} \, {\dd{\zubis}} \\
\label{eq: bar}
\longrightarrow \; & \oint_{\Gamma_{i;s}} \oint_{\Gamma_{j;t}}
        \frac{\partial^2 \Green^{\DorN}_\domain (\zeta, \xi)}{\partial{\zeta} \partial{\xi}}
        (\zeta-z_{i})^{-k} (\xi-z_{j})^{-\ell} \; \ud \zeta \, \ud \xi \; ,
\end{align}
with error estimates of the form that ensure local uniformity of the convergence.

It is natural to split the discrete and continuous integrals
in~\eqref{eq: foo} and \eqref{eq: bar} over rectangle contours to the four
sides of the rectangles; with 16 terms in total from the double integrations.
Let us for concreteness focus on the term when both integration variables
are on the bottom side of their corresponding rectangles.

\begin{figure}[h!]
	\centering
	\begin{overpic}[scale=0.78]{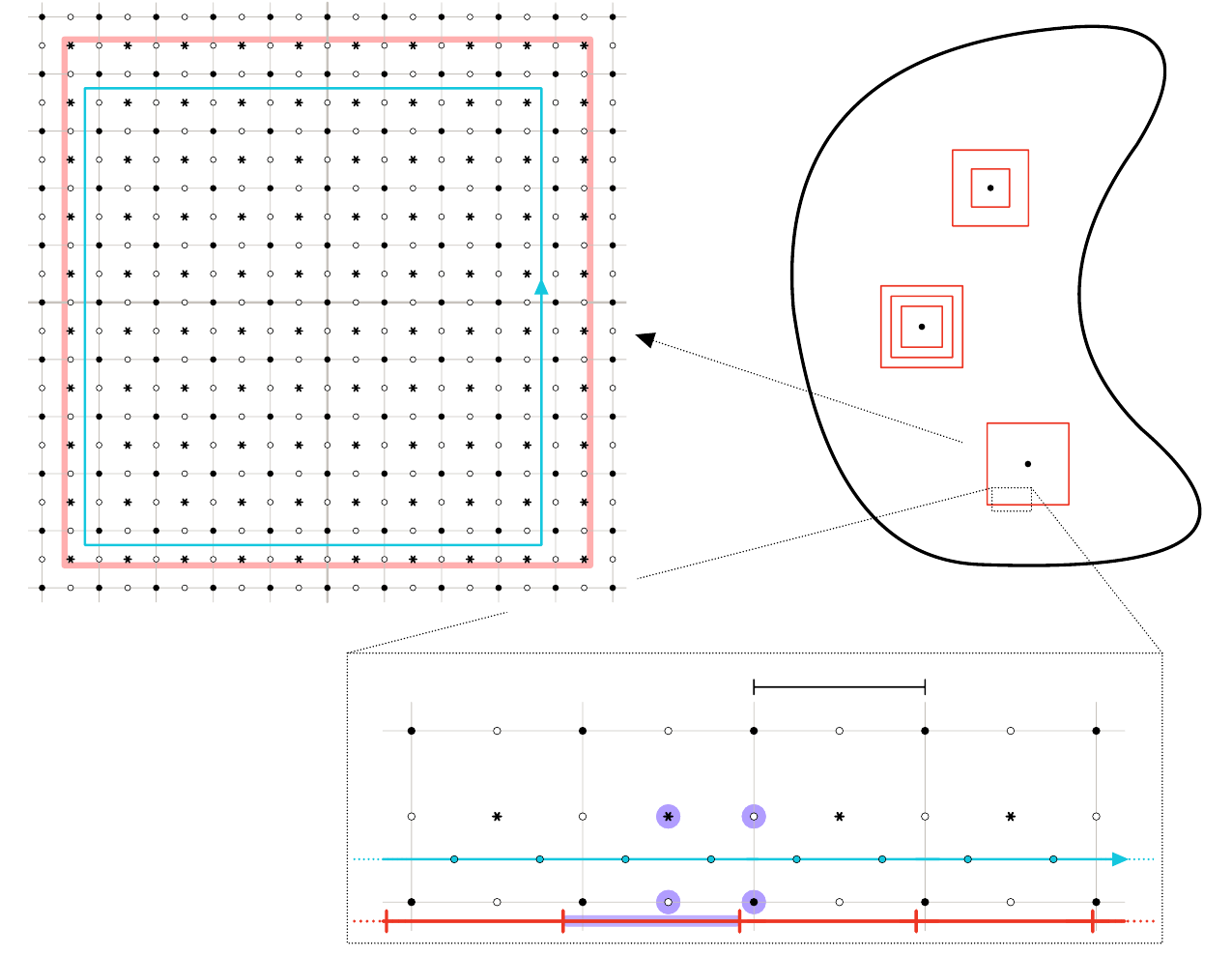}
		\put(65.8,23.6){
			\pgfsetfillopacity{0.9}
			\colorbox{white}{\centering
				\parbox{0pt}{\pgfsetfillopacity{1}}\color{black} \hspace{0pt}$\meshsz$}}
		\put(87.5,42){$\Gamma_{i;s}$}
		\put(81,41){$z_i$}
		\put(61.3,13.5){
			\pgfsetfillopacity{0.9}
			\colorbox{white}{\centering
				\parbox{0pt}{\pgfsetfillopacity{1}}\color{black} \hspace{0pt}$\zs_i^\meshsz+\meshsz\zu^{\meshsz}_{\vertical}$}}  
		\put(41.5,13.5){
			\pgfsetfillopacity{0.9}
			\colorbox{white}{\centering
				\parbox{0pt}{\pgfsetfillopacity{1}}\color{black} \hspace{0pt}$\zs_i^\meshsz+\meshsz\zu^{\meshsz}_{\dual}$}}
		\put(41.5,6.2){
			\pgfsetfillopacity{0.9}
			\colorbox{white}{\centering
				\parbox{0pt}{\pgfsetfillopacity{1}}\color{black} \hspace{0pt}$\zs_i^\meshsz+\meshsz\zu^{\meshsz}_{\horizontal}$}}
		\put(61.3,6.2){
			\pgfsetfillopacity{0.9}
			\colorbox{white}{\centering
				\parbox{0pt}{\pgfsetfillopacity{1}}\color{black} \hspace{0pt}$\zs_i^\meshsz+\meshsz\zu^{\meshsz}_{\primary}$}}  
		\put(52,76){$\Z^2$}
		\put(92,60){$\domain$}
		\put(9.5,77){
			\pgfsetfillopacity{0.9}
			\colorbox{white}{\centering
				\parbox{0pt}{\pgfsetfillopacity{1}}\color{black} \hspace{0pt}$\meshsz^{-1}\big(\Gamma_{i;s}-z_i\big)$}} 
		\put(10.5,67.1){
			\pgfsetfillopacity{0.9}
			\colorbox{white}{\centering
				\parbox{0pt}{\pgfsetfillopacity{1}}\color{black} \hspace{0pt}$\gamma^{(\meshsz)}_{i;s}$}} 
	\end{overpic}
	\centering
	\caption{The top right figure illustrates a domain~$\domain$ with marked points $z_1, \ldots, z_n$ and
	concentric radially-ordered square contours~$\Gamma_{i;s}$ around~$z_i$ for the integrals
	in~\eqref{eq: CFT correlation}. The top left figure illustrates a
	corner lattice discrete integration contour~$\gamma^{(\meshsz)}_{i;s}$ with even side
	lengths approximating a shifted and rescaled version of~$\Gamma_{i;s}$ from the inside.
	The bottom figure shows a part of the bottom side of the discrete and continuous
	integration contours,
	and illustrates the lattice points $\zu^{\meshsz}_{\dual}$,
	$\zu^{\meshsz}_{\horizontal}$,
	$\zu^{\meshsz}_{\primary}$, and $\zu^{\meshsz}_{\vertical}$
	involved in two consecutive steps of discrete integration.
	The piecewise constant function~$Q^{(\delta)}_{k,l}$ is defined in terms of these.
	}
	\label{fig: square contours}
\end{figure}

The remaining slight complication in interpreting~\eqref{eq: foo} as a Riemann sum
approximation is related to the combinatorics of the corner contour
integrations and the primal lattice discrete Wirtinger derivatives~$\pdee$
of \eqref{eq: primal lattice dee}.
Namely, every other step of a single corner contour integral
contributes with a horizontal discrete derivative of the Green's function,
whereas the subsequent step contributes with a vertical discrete derivative,
and one of these comes with a real and the other with an imaginary prefactor.
For this reason, in the discrete integrations with even numbers of steps, we group together the four
terms of two consecutive steps in each of the two variables, and
interpret the sum as $\meshsz^{-2}$ times a double integral of a function
$Q^{(\meshsz)}_{k,\ell}$, to be described below, which is constant on
$\meshsz \times \meshsz$ squares.
We then estimate the error
by triangle inequality for integrals
\begin{align*}
& \Big| \text{\eqref{eq: foo}} - \text{\eqref{eq: bar}} \Big| \\
\le \; & \iint
        \bigg| \meshsz^{- k - \ell - 2}  Q^{(\meshsz)}_{k,\ell}(\zeta, \xi)
        - \frac{\partial^2 \Green^{\DorN}_\domain (\zeta, \xi)}{\partial{\zeta} \partial{\xi}}
        (\zeta-z_{i})^{-k} (\xi-z_{j})^{-\ell}  \bigg| \; |\ud \zeta| \, |\ud \xi| + \OO(\meshsz) ,
\end{align*}
where the last error term is from the $\OO(\meshsz)$ discrepancy of lengths of the
rectangle side length and its discretization.
The terms combined in the piecewise constant function $Q^{(\meshsz)}_{k,\ell}$ are
the contributions of the following form. 
Among two consecutive corner lattice steps along the bottom side of the rectangle
contour~$\gamma^{(\meshsz)}_{i;s}$,
one separates a horizontal edge
$\zu^{\meshsz}_{\horizontal}\in\ZMedial$ from a dual vertex
$\zu^{\meshsz}_{\dual}\in\ZDual$, and the other separates 
a vertical edge $\zu^{\meshsz}_{\vertical} \in\ZMedial$ from a primal
vertex $\zu^{\meshsz}_{\primary}\in\ZPrimary$~--- see Figure~\ref{fig: square contours}.
Similarly two consecutive steps along the bottom side of the other
rectangle~$\gamma^{(\meshsz)}_{j;t}$,
involve $\zubis^{\meshsz}_{\horizontal}\in\ZMedial$ separated from
$\zubis^{\meshsz}_{\dual}\in\ZDual$, and
$\zubis^{\meshsz}_{\vertical} \in\ZMedial$ separated from
$\zubis^{\meshsz}_{\primary}\in\ZPrimary$. These lattice points
depend (in a piecewise constant way) on the variables $\zeta$
and $\xi$, respectively, but we omit writing the explicit
dependence~$\zu^{\meshsz}_{\horizontal} = \zu^{\meshsz}_{\horizontal}(\zeta), \ldots, \zubis^{\meshsz}_{\primary} = \zubis^{\meshsz}_{\primary}(\xi)$ below. The value
of~$Q^{(\meshsz)}_{k,\ell}$ is, then, the combined contribution
to the discrete integration from the two pairs of steps,
\begin{align}\label{eq: individual steps}
Q^{(\meshsz)}_{k,\ell}(\zeta,\xi) := \;
	\frac{1}{4}\,
	\bigg( 
	\big( \zu_\primary^{\meshsz} & \big)^{[-k]}
	\big( \zubis_\primary^{\meshsz} \big)^{[-\ell]}\,
	(\pdee)_{\zu}
	(\pdee)_{\zubis}
	\Green_{\domain_\meshsz}^\DorN
	\big( \zs_i^\meshsz +  \meshsz \zu^{\meshsz}_\vertical,\zs_j^\meshsz + \meshsz \zubis^{\meshsz}_\vertical \big)
	\\
	\nonumber
	+\;\; &
	\big( \zu_\dual^{\meshsz} \big)^{[-k]}
	\big( \zubis_\primary^{\meshsz} \big)^{[-\ell]}\,
	(\pdee)_{\zu}
	(\pdee)_{\zubis}
	\Green_{\domain_\meshsz}^\DorN
	\big( \zs_i^\meshsz +  \meshsz \zu^{\meshsz}_\horizontal,\zs_j^\meshsz + \meshsz \zubis^{\meshsz}_\vertical \big)\phantom{\bigg\vert}
	\\
	\nonumber
	+\;\; &
	\big( \zu_\primary^{\meshsz} \big)^{[-k]}
	\big( \zubis_\dual^{\meshsz} \big)^{[-\ell]}\,
	(\pdee)_{\zu}
	(\pdee)_{\zubis}
	\Green_{\domain_\meshsz}^\DorN
	\big( \zs_i^\meshsz +  \meshsz \zu^{\meshsz}_\vertical,\zs_j^\meshsz + \meshsz \zubis^{\meshsz}_\horizontal \big)\phantom{\bigg\vert}
	\\
	\nonumber
	+\;\; &
	\big( \zu_\dual^{\meshsz} \big)^{[-k]}
	\big( \zubis_\dual^{\meshsz} \big)^{[-\ell]}\,
	(\pdee)_{\zu}
	(\pdee)_{\zubis}
	\Green_{\domain_\meshsz}^\DorN
	\big( \zs_i^\meshsz +  \meshsz \zu^{\meshsz}_\horizontal,\zs_j^\meshsz + \meshsz \zubis^{\meshsz}_\horizontal \big)
	\bigg)
	\,.
\end{align}
The discrete monomial asymptotics~\eqref{eq: discrete monomial asymptotics} from
Proposition~\ref{prop: monomials} yield the following locally uniform
convergence as $\meshsz \to 0$
\begin{align*}
	\meshsz^{-k}\big(\zu^{\meshsz}_{\dual/\primary}\big)^{[-k]}
	=
	(\zeta-z_i)^{-k}
	+
	o(1)
	\mspace{20mu}\textnormal{ and }\mspace{20mu}
	\meshsz^{-\ell}\big(\zubis^{\meshsz}_{\dual/\primary}\big)^{[-\ell]}
	=
	(\xi-z_j)^{-\ell}
	+
	o(1)\, .
\end{align*}
By locally uniform convergence of derivatives of discrete Green's
functions,
Lemmas~\ref{lem: discrete Dirichlet Green function convergence}
and~\ref{lem: discrete Neumann Green function convergence},
we also have
\begin{align*}
	\meshsz^{-2}(\pdee)_{\zu}(\pdee)_{\zubis}
	\Green_{\domain_\meshsz}^\DorN
	\big( \zs_i^\meshsz + \meshsz\zu_\horizontal , \zs_j^\meshsz + \meshsz\zubis_\vertical \big)
	\, = & \ 
	\big(\dee^x_1\big)\big(-\ii\dee^y_2\big)
	\Green_{\domain}^\DorN
	\big( \zeta , \xi \big)
	+
	o(1)\,,\phantom{\Big\vert}
	\\
	\meshsz^{-2}(\pdee)_{\zu}(\pdee)_{\zubis}
	\Green_{\domain_\meshsz}^\DorN
	\big( \zs_i^\meshsz + \meshsz\zu_\vertical , \zs_j^\meshsz + \meshsz\zubis_\horizontal \big)
	\, = & \ 
	\big(-\ii\dee^y_1\big)\big(\dee^x_2\big)
	\Green_{\domain}^\DorN
	\big( \zeta , \xi \big)
	+
	o(1)\,,\phantom{\Big\vert}
	\\
	\meshsz^{-2}(\pdee)_{\zu}(\pdee)_{\zubis}
	\Green_{\domain_\meshsz}^\DorN
	\big( \zs_i^\meshsz + \meshsz\zu_\horizontal , \zs_j^\meshsz + \meshsz\zubis_\horizontal \big)
	\, = & \ 
	\big(\dee^x_1\big)\big(\dee^x_2\big)
	\Green_{\domain}^\DorN
	\big( \zeta , \xi \big)
	+
	o(1)\,,\phantom{\Big\vert}
	\\
	\meshsz^{-2}(\pdee)_{\zu}(\pdee)_{\zubis}
	\Green_{\domain_\meshsz}^\DorN
	\big( \zs_i^\meshsz + \meshsz\zu_\vertical , \zs_j^\meshsz + \meshsz\zubis_\vertical \big)
	\, = & \ 
	\big(-\ii\dee^y_1\big)\big(-\ii\dee^y_2\big)
	\Green_{\domain}^\DorN
	\big( \zeta , \xi \big)
	+
	o(1)\,.\phantom{\Big\vert}
\end{align*}
Regrouping these derivatives into Wirtinger derivatives, we obtain
the desired estimate for the error integrand,
\begin{align*}
\bigg| \meshsz^{- k - \ell - 2}  Q^{(\meshsz)}_{k,\ell}(\zeta, \xi)
        - \frac{\partial^2 \Green^{\DorN}_\domain (\zeta, \xi)}{\partial{\zeta} \partial{\xi}}
        (\zeta-z_{i})^{-k} (\xi-z_{j})^{-\ell}  \bigg| \; = \; \oo(1)
\end{align*}
when $\meshsz \to 0$. This error integrand is integrated over a fixed finite
pair of rectangle sides, so the total error is
$\big| \text{\eqref{eq: foo}} - \text{\eqref{eq: bar}} \big| = \oo(1)$.
This completes the proof.
\end{proof}

\newpage

\titleformat{\section}
{\normalfont\Large\bfseries}{\thesection}{0pt}{}

\end{document}